\documentclass[12pt,a4paper]{report}
\usepackage{amsmath}
\usepackage{amssymb}
\usepackage{graphicx}
\usepackage{bm}
\usepackage{setspace}
\usepackage{subfig}
\usepackage{cite}
\usepackage{bbold}
\usepackage{appendix}
\usepackage{pdflscape}
\usepackage[T1]{fontenc}
\usepackage{inslrmin}

\def\half{\frac{1}{2}}

\newcommand{\non}{\nonumber}
\newcommand{\sbp}{\subparagraph*{}}
\newcommand{\be}{\begin{equation} }
\newcommand{\ee}{\end{equation}}
\newcommand{\bea}{\begin{eqnarray} }
\newcommand{\eea}{\end{eqnarray} }
\newcommand{\sg}{\sigma} 
\newcommand{\gm}{\gamma}
\newcommand{\Gm}{\Gamma}
\newcommand{\hb}{\hbar}
\newcommand{\ap}{\alpha}
\newcommand{\og}{\omega}
\newcommand{\Og}{\Omega}
\newcommand{\lb}{\lambda}
\newcommand{\bt}{\beta}
\newcommand{\tbf}{\textbf}
\newcommand{\bds}{\boldsymbol}

\newcommand{\dt}{\delta}
\newcommand{\ep}{\epsilon}

\newcommand{\dg}{\dagger}
\newcommand{\comment}[1]{}
\newcommand{\alld}{\allowdisplaybreaks}
\newcommand{\ta}{\theta}
\newcommand{\zt}{\zeta}
\newcommand{\vp}{\varpi}
\newcommand{\im}{\imath}
\newcommand{\vep}{\varepsilon}

\newcommand{\vsg}{\varsigma}
\newcommand{\Lb}{\Lambda}

\newcommand{\Dt}{\Delta}
\newcommand{\ptl}{\partial}

\newcommand{\Tr}{{\rm Tr}}
\newcommand{\Rl}{{\rm Re}}
\newcommand{\diag}{{\rm diag}}
\newcommand{\eig}{{\rm eig}}
\newcommand{\gins}{\text{\iminfamily g \normalfont}}
\newcommand{\mins}{\text{\iminfamily m \normalfont}}
\newcommand{\nins}{\text{\iminfamily n \normalfont}}

\newtheorem{theorem}{Theorem} 
\newtheorem{lemma}{Lemma}
\newtheorem{defn}{Definition}
\newtheorem{proposition}{Proposition}
\newenvironment{proof}{{\bf Proof:}}{\hfill$\square$\vskip.5cm}

\begin{document} 


\singlespacing
\begin{center}
\vspace*{-2cm}
\includegraphics[width=6cm]{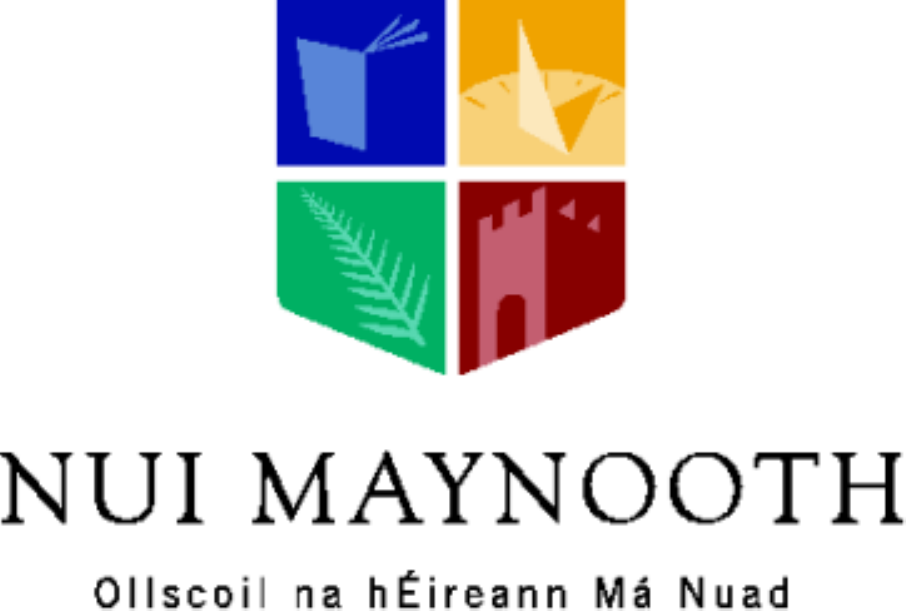}\\
\vfill
{\LARGE Entanglement in a Bipartite Gaussian State}\\

\vspace{1cm}
{\bf \Large Anne Ghesqui\`{e}re}

{\scshape Based on research conducted in the\\
Dublin Institute for Advanced Studies}\\
\vspace{0.5cm}under the supervision of\\
\vspace{0.2cm}{\bf\large Prof.~Tony Dorlas}\\

\vspace{7mm}

{\scshape Submitted in partial fulfilment of the requirements for the degree of \\Doctor of Philosophy
in the \\Department of Mathematical Physics
of the \\Faculty of Science at the\\
National University of Ireland, Maynooth}\\

\vspace{5mm}

{\it 2009}

\vspace{5mm}

under the supervision of\\
\vspace{0.2cm}{\bf\large Prof.~Daniel Heffernan}\\
\vspace{7mm}

\vfill
\includegraphics[width=4cm]{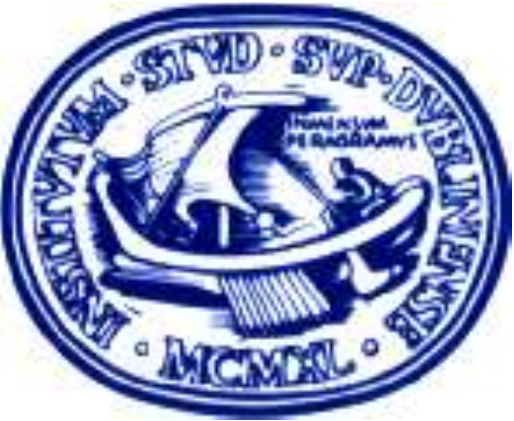}
\vspace{-2.0cm}
\end{center}

\addtocounter{page}{-1}

\doublespacing


	\pagenumbering{Roman}
		
		\tableofcontents


\chapter*{Acknowledgments}

\thispagestyle{empty}

\sbp I would like to thank my supervisors Professor Tony Dorlas and Professor Daniel Heffernan for giving me the opportunity to explore the field of entanglement and for their help in both academic and administrative matters. I would also like to thank Professor Bob O'Connell and Professor Bill Ford for their help and challenging discussions. Special thanks to Dr. Graham Kells for his help and advice.

I would like to thank my examiners Professor Bob O'Connell, Dr. Jonivar Skullerud and Professor Tony Dorlas for taking the time to examine this work.

I thank my family for her support and advice. I would also like to thank my friends for always being there for a chat, help, lunch, a game or a cup of tea, especially Ciara Morgan, Dr. Gerry Boland, Fiona O'Brien, Teresa Hughes, Sorcha Power, Sandrine Fargeat and Chie Nakamura. Warm thanks to Diarmuid \'{O} Math\'{u}na, Dr. Se\'{a}n Murray and Sepanda Pouryarya for livening up the office. For their love and support, I would also like to thank Stephen Burnell and his family.
 
For their help in all technological matters, I would like to thank Andres Jimenez and Jean-Fran\c{c}ois Bucas. 

I would like to thank the Dublin Institute for Advanced Studies and the Kildare County Council for their financial support.

\pagestyle{plain}
\newpage

\chapter*{}

\thispagestyle{empty}

\begin{center}
  \`{a} grand-p\`{e}re
\end{center}

\pagestyle{plain}		

\chapter*{Abstract}

\addcontentsline{toc}{chapter}{Abstract}

\thispagestyle{empty}

\sbp To examine the loss of entanglement in a two-particle Gaussian system, we couple it to an environment and use the Non-Rotating Wave master equation to study the system's dynamics. We also present a derivation of this equation.

We consider two different types of evolution. Under free evolution we find that entanglement is lost quickly between the particles. When a harmonic potential is added between the particles, two very different behaviours can be observed, namely in the over and under-damped cases respectively, where the strength of the damping is determined by how large the coupling to the bath is with respect to the frequency of the potential. 

In the over-damped case, we find that the entanglement vanishes at even shorter times than it does in the free evolution. In the (very) under-damped case, we observe that the entanglement does not vanish. Instead it oscillates towards a stable value.

\pagestyle{plain}

	\pagenumbering{arabic}


\chapter*{Introduction}
\stepcounter{chapter}
\addcontentsline{toc}{chapter}{Introduction}

\setcounter{page}{1}

\sbp Freud said that mankind's ego underwent three downfalls caused by Galileo's heliocentrism, Darwin's theory of evolution and his own theory of the unconscious mind. Similarly, it can be said that the field of physics was "reduced" three times with Newtonian mechanics, statistical mechanics (followed by quantum mechanics) and the theory of relativity. Whereas Newtonian mechanics concerns solid objects, statistical mechanics concerns ensembles of particles and is therefore much better suited to study the very small. Finally quantum mechanics allows for a study of particles and the evolution of their states in time. Entanglement is a property very specific to quantum mechanics and is proving a wonderful resource to the field of quantum information and quantum computing. This work sets to study how a quantum state, and the entanglement within it, evolve with time.

\section{Study of open quantum systems}

\sbp The mechanics of closed quantum mechanical systems are rather well-known and have been extensively studied. If one considers a system with wavefunction $\Psi(\tbf{r},t)$, the time-evolution of the system is obtained from the Schr\"{o}dinger equation 

\be H\mid\Psi(\tbf{r},t)\rangle= \imath\hb\frac{\partial}{\partial t}\mid\Psi(\tbf{r},t)\rangle \ee

where H is the Hamiltonian of the system, usually in terms of the Laplacian $\nabla$, a potential $V(\tbf{r})$, the mass of the particle $m$ and Planck's constant.

\be H= -\frac{\hb ^2}{2m}\nabla ^2+V(\tbf{r}) \ee

When the particle is coupled to an environment (which in the remainder of this work will also be known as a heat-bath or a reservoir), the Schr\"{o}dinger equation is no longer sufficient to describe the system's evolution, since it fails to take the bath's variables and behaviour into account. The Schr\"{o}dinger equation described above has to be modified to account for the damping caused by the environment. It is now appropriate to examine the relationship between the bath and the system and the dynamics of the dissipation from a general point of view. We now consider one-dimensional systems.

\sbp If the system is interacting with the bath via a random force F(t), the system evolution is called \textit{stochastic}. Classically, a stochastic equation may be written for a physical variable $q$ and the random force is real-valued. Quantum mechanically, the variable $q$ becomes an operator and the random force an operator-valued function. One can write a stochastic equation as

\be {\ddot q}+ \gm {\dot q} + \og ^2 q = F(t)/m \ee

A more general form of this equation is the Quantum Langevin Equation (Q.L.E.) \cite{FLO'C:QLE1988}

\be m {\ddot x} + \int_{-\infty} ^{t} dt' \mu(t-t'){\dot x}(t') + V'(x) = F(t)  \label{qle}\ee

where the dot and the prime are respectively the derivative with respect to $t$ and $x$. The function $\mu(t)$ is a ``memory'' function and depends on the way the bath is coupled to the system. It can be reduced to a constant $\gm$ in a Markov approximation. The QLE can be derived from the Hamiltonian of the complete system using the Heisenberg equations of motion and a specific bath model (see Chapter 3).

The state of a quantum system is best described by a density matrix $\rho$. A state is called \textit{pure} when its density matrix is the projection $\mid \Psi \, \rangle \langle \, \Psi \mid$ onto a vector in the Hilbert space. The time evolution of the complete system (system + bath) may then be expressed by the Von Neumann equation
\be \frac{\partial \rho}{\partial t} = \frac{1}{\im \hb} \left[H, \rho \right] \label{VNEquation} \ee
with a Hamiltonian chosen depending on the system and the type of coupling. Considering only the density matrix of the system, one may derive a macroscopic description in terms of the operators $x$ and $p$ from (\ref{VNEquation}), which is generally known as a \textit{master equation}. There are various master equations in common use. The types of master equations depend on the assumptions made concerning the bath, its coupling to the system and also the initial conditions. A general derivation may be found in \cite{FLO'C:1996, Gard:2000}. 

A very popular type is obtained in the rotating-wave approximation (RWA). The main reason for the popularity of this approximation resides in the fact that a master equation obtained in the RWA is of the following Lindblad form \cite{Lindblad:1976} 
\be \label{Lindb}
\frac{d\rho}{dt}= -\frac{\im}{\hb}[H,\rho] + \sum_{j} \left(2A_{j} ^{\dg} \rho A_{j} - \rho A_{j} A_{j}^{\dg}
 - A_{j} A_{j} ^{\dg} \rho\right) \ee
where H is the Hamiltonian and $A_{j}$ are operators depending on the system studied. The solutions of Lindblad type equations are guaranteed to be positive semi-definite matrices. The RWA is well suited for systems with well-defined energy levels and rather weak coupling, for instance a harmonic oscillator in a weak radiation field. 

One may wish to study more general systems. The coupling may be stronger, or the system may have energy levels much closer together. The approximations required in this case are quite different and one may derive a "pre-Lindblad" type of master equation. Various such equations have been derived using a variety of approaches. For instance, the Hu-Paz-Zhang equation was derived using path integral techniques \cite{HPZ:1992}, whereas Agarwal derived a general equation in phase space formalism \cite{Agar:1971, Agar:1969, Agar:1968}. An alternative derivation of the Hu-Paz-Zhang equation was proposed by Halliwell and Yu \cite{Halliwell:1996} using the Wigner function, a preferred method of study for most of these authors, mainly because it is real everywhere. The Hu-Paz-Zhang was also solved explicitly by Ford and O'Connell \cite{FO'C:PRD2001} and the solution was shown to be valid only for high temperatures. The solution of a pre-Lindblad equation is not always positive semi-definite and one must check in each case whether the resulting solution is a density matrix. 

The present work will use the so called Non-Rotating Wave (NRW) master equation \cite{Gard:2000, Gard:NRWME}
\be \label{NRWMEint}
\frac{\partial\rho(t)}{\partial t}=-\frac{\imath}{\hb}[H_{syst}, \rho(t)]-\frac{\imath\gm}{2\hb}[X,[{\dot X},\rho(t)]_{+}]-\frac{\gm kT}{\hb^{2}}[X,[X,\rho(t)]] \ee
with ${\dot X}=\frac{\imath}{\hb}[H_{syst},X]$ and X in principle an arbitrary system operator coupled to the bath. It is of pre-Lindblad nature and obtained in the high temperature limit. It is similar in its formulation to that derived by Agarwal and used by Savage and Walls in \cite{SavWall:85_1, SavWall:85_2}. The fourth chapter of this work contains an analysis of this equation for a single particle with results similar to those of Savage and Walls. This equation can be derived from the Quantum Langevin Equation for an independent-oscillator model of the heat-bath. In the third chapter of this work, such a derivation is presented using perturbation theory techniques.

\section{Entanglement}

\sbp Entanglement is one of quantum mechanics' most fascinating features. It is the property two quantum systems can share that allows one to get a piece of information about both systems while measuring it in only one. It was highlighted by the Einstein, Podolosky and Rosen (EPR) paradox in their famous 1935 paper \cite{EPR:1935}. The authors used a thought experiment to express their doubt that the wavefunction description is complete. Performing a measurement reduces the system's wavefunction. Performing a measurement on one of two correlated systems reduces the two systems' wavefunctions simultaneously. According to special relativity, information cannot travel faster than light, yet the systems can be infinitely far apart. Such correlation is known as entanglement, a term originally introduced by Schr\"{o}dinger \cite{Schrod:1935}. 

By its very nature, entanglement has proven to be a fantastic resource for quantum information. It is used extensively in quantum information theory, quantum coding and quantum cryptography, because it allows for very dense coding and totally secure encryption.

\subsection*{Decoherence versus Entanglement Sudden-Death} 

\sbp When subjected to an environment, a quantum state tends to decohere. Decoherence can be qualitatively defined as the destruction of the interference pattern of a quantum state \cite{FO'C:DWD2001, O'C:2005, MFO'C:2003, FLO'C:QMD2001, FO'C:2004}. This definition can be quantitatively represented in the interference term of the probability distribution of a particle coupled to an arbitrary reservoir for instance. Furthermore, O'Connell \cite{O'C:2005} points out that one may obtain different results depending on various conditions, for instance whether one assumes initial coupling between the particle and the reservoir or if external forces come into play. The decoherence time is typically much shorter than the relaxation time of the overall system and as such is a rather important quantity \cite{Zurek:1991, Zurek:2003}. The degree of decoherence is represented by the variance of the off-diagonal terms of the density matrix.

\sbp If one studies the evolution of the entanglement in a two-particle system coupled to two independent reservoirs, one finds that the entanglement measures typically show a sharp decrease, vanishing at a finite time whereas the coherence merely vanishes asymptotically. Such phenomenon is known as entanglement sudden-death and has been the object of much study in recent years \cite{YE:2003, YE:2004, YE:2006, PE:2001, RM:2006}. One can typically observe this decay by evolving the initial state with a master equation then estimating the degree of entanglement using an entanglement measure. Eberly et al. uses the concurrence \cite{Mintert:2005, Woot:1998} whereas Eisert et al. observed entanglement decay and entanglement transfer with the logarithmic negativity \cite{Eisert:2004}. The following subsection explores some entanglement measures in more detail. Entanglement arising between the two reservoirs (known as sudden birth of entanglement) has also been observed \cite{Ficek:2006, Ficek:2008, LRLSR:2008}. One finds that it increases as the entanglement between the two particles decreases.

\subsection*{Entanglement Measure}

\sbp Entanglement of a general quantum state is defined as the opposite of separable. Separability can be expressed simply as follows

\begin{defn} A state $\rho$ of a composite system of two parts is separable if it can be written as 
\be \label{cdt1} \rho_{AB} = \sum_{k=0} ^{n}\rho_A ^k \otimes \rho_B ^k \ee
\end{defn}

When the state under study is that of a pure bipartite state $\rho_{AB}$, the Von Neumann entropy of the reduced density operator, also called the \textit{entropy of entanglement}, is a good measure of entanglement. It is the quantum analogue of the Shannon entropy and is defined by $ S( \rho_A ) = - \Tr(\rho_A \,\, \ln (\rho_A) )$, where $\rho_A = \Tr_B(\rho)$. The Von Neumann entropy is invariant under a closed system time evolution, as will be proven later in this work.

When studying mixed systems or systems with more than two particles, the Von Neumann entropy is no longer sufficient to fully quantify the entanglement. To be called an \textit{entanglement measure}, a function $\mathcal{E}$ must satisfy the following three conditions \cite{Vedral:1997}.

\begin{enumerate}
\item $\mathcal{E}(\rho)$ vanishes if the state $\rho$ is separable. \label{1}
\item $\mathcal{E}$ does not increase on average under local operations and classical communications (LOCC). \label{2}
\item $\mathcal{E}$ is invariant under local unitary transformations. \label{3}
\end{enumerate}

These conditions however do not uniquely specify a measure for mixed states. Numerous measures have been proposed, such as the \textit{entanglement of formation} \cite{Woot:1998}, the \textit{entanglement of distillation} \cite{Bennett:1996}, the \textit{concurrence} \cite{Mintert:2005, Woot:1998} or the \textit{global entanglement for multipartite system} \cite{Meyer:2002, Vedral:2008}, to name but a few. The concurrence in particular has been quite popular \cite{YE:2003, YE:2004, YE:2006, PE:2001, RM:2006}.

One may write the entangled state $\rho _{AB}$ as $\rho = \sum_i p_i \parallel \psi_i \rangle \langle \psi_i \parallel$. The entanglement of formation is the amount of entanglement needed to create the entangled state $\rho$ and is defined as $\mathcal{E}(\rho) = \min \sum_i p_i S(\rho_i)$. The entanglement of distillation is the amount of entanglement that one obtains after purifying the state. The concurrence is related to the entanglement of formation and provides a formula for the an abritrary state of two-qubits. It is given as $\mathcal{C}(\rho) = \max \left\{ 0, \lb_1 - \lb_2 - \lb_3 - \lb_4 \right\}$ where the $\lb_i$ are the eigenvalues, in decreasing order of the Hermitian matrix $R \equiv \sqrt{\sqrt{\rho}{\tilde \rho} \sqrt{\rho}}$. The global entanglement for multipartite systems provides the entanglement of one of the particles to the rest of the system and can be written as $\mathcal{E}(\rho) = 2 - \frac{2}{N} \sum_{j=1}^N \Tr \rho_j ^2$.

Chosen here for the ease of computation with which it is computed for continuous variables and its widespread use is the negativity, seconded by its close cousin, the logarithmic negativity \cite{Werner:2002}. They are based on the trace norm of the partial transpose of the density operator $\rho$ where $\rho$ represents a generic state of a bipartite system. The partial transpose is obtained by \cite{Werner:2002, Anders:2003, Eisert:2003, Eisert:2004} 
\be \langle i_A, j_B \mid \rho^{T_A} \mid k_A, l_B \rangle \equiv \langle k_A, j_B \mid \rho \mid i_A, l_B \rangle \non \ee 
The trace norm of a Hermitian  operator A is $\parallel A \parallel_1 \equiv \Tr \sqrt{A^{\dg} A} \equiv \sum | \lb_i | $ where the $\lb_i$ are the eigenvalues of A. Density matrices are Hermitian matrices and as such have positive eigenvalues and $\parallel \rho \parallel_1 = \Tr \rho = 1$. The partial transpose $\rho^{T_A}$ also has trace 1 but since it may have negative eigenvalues $\mu_i$, its trace norm reads

\be \parallel \rho ^{T_A} \parallel_1 = 1 + 2 \mid \sum_i \mu_i \mid \equiv 1+ 2\mathcal{N}(\rho) \ee

 where $\mathcal{N}(\rho)$ is the negativity. One can write the \textit{partial transpose separability criterion} as \cite{Werner:2002, Mintert:2005}

\begin{theorem}Let $\rho^{T_A}$ be the partial transpose of a state $\rho$ with eigenvalues $\lb_i$. If one of the $\lb_i$ is negative, then the state is entangled.
\end{theorem}

\sbp However, the converse has been shown to be true only for $2 \times 2$ and $2 \times 3$ dimensional systems \cite{Mintert:2005, Simon:2000}, i.e. in systems of this size, the lack of a negative eigenvalue is not enough to guarantee that the state is separable. The negativity of a separable state $\rho_s$ can be shown to be  $\mathcal{N}(\rho_s)=0$. It is also monotonous under LOCC. 

The logarithmic negativity is expressed as 

\be \mathcal{L}_{\mathcal{N}}(\rho) \equiv \log_2 \parallel \rho ^{T_A} \parallel_1\ee  

and has an interpretation as an asymptotic entanglement cost, which itself is the asymptotic version of the entanglement of formation \cite{Eisert:2003}. Since $\mathcal{N}(\rho_s)=0$ for a separable state, it is easy to see that $\mathcal{L}_{\mathcal{N}}(\rho) = 0$ also.

\section*{Motivation}

\sbp Entanglement is a remarkable resource. It is, however, rather fragile and examining just how fragile is the aim of this work. To this end, the dynamics of the entanglement in a bipartite Gaussian state will be examined.

The state is prepared such that it is initially entangled. Using the Von Neumann entropy, the second chapter will establish that in a closed system, this entanglement does not vary. The second part of this chapter will present the covariance matrix formalism that will be used in subsequent chapters.

The state is then subjected to an open system evolution and the Non-Rotating Wave master equation is solved using two Hamiltonians. The first Hamiltonian considered is a free-particle Hamiltonian and the results are contained in the fifth chapter. The final chapter concerns a harmonic potential Hamiltonian. In both cases, the covariance matrices are determined and the logarithmic negativity obtained. This provides some insight into the influence of a harmonic potential over the sudden death of entanglement

Finally, this work's main results are summarised in a short conclusion.


		 \chapter*{Gaussian States}
\stepcounter{chapter}
\addcontentsline{toc}{chapter}{Gaussian States}

\sbp The Von Neumann entropy is an essential tool in the matter of studying the entanglement of pure states. In the current chapter, the entropy is proved invariant under closed system dynamics, showing that pure state entanglement is conserved in closed systems. A formalism is also introduced that allows for an easy study of Gaussian states.

\section{The entanglement entropy of a two-particle Gaussian state}

\sbp Let us consider the Gaussian state for two particles in one dimension, suggested by Ford and O'Connell \cite{FO'C:PC2008} and given by
\be\Psi(x_1,x_2)=\frac{1}{\sqrt{2\pi s d}}e^{-\frac{(x_1-x_2)^2}{4 s^2}-\frac{(x_1+x_2)^2}{16d^2}}
\ee
The corresponding density matrix is
\be\rho(x_1,x_2;x'_1,x'_2)=\Og e^{-\ep_+ ({x_1}^2+{x_2}^2+{x'_1}^2+{x'_2}^2)+2\ep_-(x_1x_2+x'_1x'_2)}
\ee
where $\Og=\frac{1}{2\pi s d}$, $\ep_+=\frac{1}{4 s^2}+\frac{1}{16d^2}$ and $\ep_- =\frac{1}{4 s^2}-\frac{1}{16d^2}$
The reduced density matrix is obtained by tracing over the position of particle 2
\be \rho_{1}(x_{1},x'_{1}) = \int \rho(x_{1},x_{2};x'_{1},x_{2})dx_{2} \ee
Explicit calculations give
{\alld
\begin{align}  \rho_{1}(x_{1},x'_{1}) &= \Og e^{-\ep_+ (x_1 ^2 + {x'_1} ^2)} \int e^{-2\ep_+ x_2 ^2 + 2 \ep_- (x_1 + x'_1) x_2} \, dx_2 \non 
\\ &= \Og' e^{-\vp x_1 ^2 - \vp {x'_1} ^2 + \nu x_1 x'_1} 
\end{align}}
with $\vp= \ep_+ - \frac{\ep_- ^2}{2\ep_+}$, $\nu = \frac{\ep_- ^2}{\ep_+}$ and $\Og' = \Og \sqrt{\frac{\pi}{2\ep_+}} $.

To calculate the Von Neumann entropy, one must calculate the eigenvalues for the state. The general eigenvalue equation is  \be \int \rho_1 (x,y) \Phi(y) dy= \lb \Phi(x)\ee
Let us try $\Phi(x'_1)=e^{- \vsg {x'_1} ^{2}}$ as an eigenvector.
{\alld
\begin{align}
 \int \rho_1 (x,x') e^{-\vsg {x'} ^2} dx' &= \Og' e^{- \vp x ^2} \int e^{-(\vp + \vsg){x'} ^2 + \nu x x'} dx' \non 
\\ &= \Og' \sqrt{\frac{\pi}{\vp + \vsg}} e^{-\vp  x ^2 + \frac{\nu ^2 x ^2}{4(\vp + \vsg)}}
\end{align}}
For this $\Phi$ to be an eigenfunction, we must have 
\be \vsg = \vp - \frac{\nu ^2}{4(\vp + \vsg)} \ee
and therefore 
\be \vsg = \sqrt{\vp ^2 - \frac{\nu ^2}{4}} = \sqrt{\ep_+ ^2 - \ep_- ^2} = \frac{1}{4 s d}
\ee
Hence we have the eigenvalue
\be \label{vnelb0} \lb_0 = \Og'\sqrt{\frac{\pi}{\vp + \vsg}} \ee

It is likely that the other eigenfunctions are given by Hermite polynomials, which can be defined by \be H_n(x)= (-1)^n e^{x^{2}}
\left(\frac{d}{dx}\right)^{n} e^{-x^{2}} \ee For example, $H_0(x)
= 1$, $H_1(x) = 2x$, $H_2(x) = 4 x^2-2$, $H_3(x) = 8x^3-12 x$,
etc.  The Hermite polynomials correspond to the Wick objects
defined in \cite{Sim:1974} for the arbitrary variable $f$ as follows. 
\bea :f^{0}:&=& 1 \non
\\ \frac{\partial}{\partial f} :f^{n}:&=&n:f^{n-1}: \non
\\ <:f^{n}:>&=&0 \eea
for $n=0,1,2,...$. As a simple example, let us consider two successive terms. 
\be \frac{\ptl }{\ptl f} :f^1: = 1 \quad \text{so by integration} \quad :f^1: = f - \langle f \rangle \non \ee
and
\be \frac{\ptl }{\ptl f} :f^2: = 2 \left( f - \langle f \rangle \right) \non \ee 
so by integration
\be :f^2: = f^2 - 2 \langle f \rangle f - \langle f^2 \rangle + 2 \langle f \rangle ^2 \non \ee

We have the following relation : 
\be :f^n: = c^n \,H_n \left( \frac{f}{2c} \right)  \ee 
where $c^2 = <f^2>$. The generating function for these objects is : 
\be :e^{zf}: = \sum_{n=0}^{\infty} \frac{z^{n}:f^{n}:}{n!} =
\frac{e^{zf}}{<e^{zf}>}, \ee \ where \be <e^{zf}> = \exp
\left[\frac{1}{2}z^{2}<f^{2}> \right]. \ee 

Writing $\Phi(y) = e^{-\vsg y^2}$, we have 
{\alld
\begin{align} \sum_{n=0}^\infty
\frac{z^n}{n!}\int \rho_1 (x,y) & :y^n:  \Phi(y) dy 
\non \\ =& \Og' \sum_{n=0}^\infty \frac{z^{n}}{n!} \int :y^{n}:e^{-\vp x^{2}- \vp y^{2}+ \nu xy - \vsg y^{2}} dy \non
\\ =& \Og' e^{-\frac{1}{2}z^{2}<y^{2}>} e^{-\vp x^{2}}
\int e^{zy} e^{-\vp y^{2}+\nu xy-\vsg y^{2}} dy \non
\\ =& \Og' e^{-\frac{1}{2}z^{2}<y^{2}>} e^{-\vp x^{2}}
\sqrt{\frac{\pi}{\vp +\vsg}} e^{\frac{(\nu x+z)^{2}}{4(\vp +\vsg)}} \non
\\ =& \Og' e^{-\frac{1}{2}z^{2}<y^{2}>} e^{-\vp x^{2}}
\sqrt{\frac{\pi}{\vp + \vsg}} e^{\frac{z^{2}}{4(\vp +\vsg)}}
e^{\frac{\nu xz}{2(\vp +\vsg)}} e^{\frac{\nu ^{2}x^{2}}{4(\vp +\vsg)}}, \non
\\ =& \Og' \sqrt{\frac{\pi}{\vp+\vsg}} e^{-\frac{1}{2}z^{2}<y^{2}>}
 e^{\frac{z^{2}}{4(\vp +\vsg)}} e^{\frac{\nu xz}{2(\vp +\vsg)}} \Phi(x), \end{align}}

From the Wick objects' generating function : 
\be \exp \left[\frac{\nu xz}{2(\vp +\vsg)} \right] = :\exp \frac{\nu xz}{2(\vp +\vsg)}: <
\exp\frac{\nu xz}{2(\vp +\vsg)}> \ee 
with \be < \exp\frac{\nu xz}{2(\vp +\vsg)}> =
\exp \left[\frac{z^{2}\nu ^{2}<x^{2}>}{8(\vp +\vsg)^2} \right] \ee
Hence the previous result becomes : 
{\alld
\begin{align} \lefteqn{
\sum_{n=0}^\infty \int \rho_1(x,y) :y^n: \Phi(y) dy} \non
\\ =& \Og' \sqrt{\frac{\pi}{\vp + \vsg}} e^{-\frac{1}{2} z^{2}<x^{2}>}
e^{-\vsg x^{2}} e^{\frac{z^{2}}{4(\vp +\vsg)}}
:\exp\frac{\nu xz}{2(\vp +\vsg)}: \, \exp
\left[\frac{z^{2}\nu ^{2}<x^{2}>}{8(\vp +\vsg)^2} \right] \non
\\ =&  \Og' \sqrt{\frac{\pi}{\vp +\vsg}} \Phi(x)
e^{-\frac{1}{2}z^{2}<x^{2}>}  e^{\frac{z^{2}}{4(\vp +\vsg)}}
e^{\frac{z^{2}\nu ^{2}<x^{2}>}{8(\vp +\vsg)^2}}
\sum_{n=0}^{\infty}\left( \frac{\nu }{2(\vp +\vsg)} \right)^{n}\frac{z^{n}}{n!}:x^{n}:
\end{align}}

This becomes of the form $ \sum_{n=0}^\infty \lb_n \frac{z^n}{n!} :x^n: \Psi(x) $, i.e. the z-dependence outside the sum cancels provided that
\be <x^2> \left( \half - \frac{\nu ^2}{8 (\vp + \vsg)^2} \right) = \frac{1}{4(\vp + \vsg)} , \ee i .e. if
\be  <x^2> = \frac{\vp + \vsg}{2 \left( (\vp +\vsg)^2 - \nu ^2/4 \right)}. \ee  

The eigenvalues are then 
\be \label{vnelbn} \lb_n = \Og \sqrt{\frac{\pi ^2}{2\ep_+ (\vp + \vsg)}} \left( \frac{\nu/2}{\vp + \vsg} \right) ^n\ee
Using $\vp + \vsg = \sqrt{\vp ^2 - \nu ^2 /4} + \vp = \half (\sqrt{\vp + \nu/2} + \sqrt{\vp - \nu/2 })^2$ and $\vp + \vsg - \nu/2 = \sqrt{\vp - \nu/2 }(\sqrt{\vp + \nu/2} + \sqrt{\vp - \nu/2 })$, we can check that $\sum_n \lb_n =1$ as
{\alld
\begin{align}
\sum _{n=0}^{\infty} \lb_n =& \Og' \sqrt{\frac{\pi}{\vp + \vsg}} \sum_{n=0} ^{\infty} \left( \frac{\nu/2}{\vp + \vsg}\right) ^n  \non 
\\ =& \Og' \sqrt{\frac{\pi}{\vp + \vsg}} \frac{1}{1-\frac{\nu / 2}{\vp + \vsg}}
\\ =& \Og \sqrt{\frac{\pi}{2\ep_+}} \sqrt{\frac{\pi}{\vp + \vsg}} \frac{\vp + \vsg}{\vp + \vsg- \nu/2} \non 
\\ \frac{\vp + \vsg}{\vp + \vsg - \nu/2} =& \half \left( 1 + \frac{\ep_+}{\ep_+ ^2 - \ep_- ^2} \right) \non 
\\ \sqrt{\frac{\pi}{\vp + \vsg}} =& \frac{\sqrt{2 \pi}}{\sqrt{\ep_+} + \sqrt{\ep_+ - \frac{\ep_- ^2}{\ep_+}}} \non
\\ \Og' \half \left( 1 + \frac{\ep_+}{\ep_+ ^2 - \ep_- ^2} \right) & \frac{\sqrt{2 \pi}}{\sqrt{\ep_+} + \sqrt{\ep_+ - \frac{\ep_- ^2}{\ep_+}}} =1 
\end{align}}

To calculate the entropy, let us introduce a lemma.

\begin{lemma} \label{LemmaEntropy} Suppose that the eigenvalues of a density matrix $\rho$ are given by $\lb_n = \vep \rho ^n $ where $\vep = 1 - \rho$ by normalisation. Then the Von Neumann entropy is given by
\be S(\rho) = - \sum_{n=0} ^{\infty} \lb_n \ln \lb_n = - \ln(1 - \rho) - \frac{\rho}{1 - \rho} \ln \rho \ee
\end{lemma}

\begin{proof} It is quite easy to see that 
\begin{align} S(\rho) = - \sum_{n=0} ^{\infty} \lb_n \ln \lb_n =& - \ln \ep - \ep \ln \rho \sum_{n=0}^{\infty} n \rho^n  \non 
\\ =& - \ln \ep - \ep \ln \rho \frac{\rho}{(1 - \rho)^2} \non 
\\ =& - \ln \ep - \ep \frac{\rho}{1 - \rho} \frac{1}{1 - \rho} \ln \rho 
\end{align}

One can recall that 
\be \sum_{n=0} ^{\infty} \rho ^n = \frac{1}{1 - \rho} \non \ee 
to notice
\be \sum_{n=0} ^{\infty} \lb_n = \frac{\ep}{1 - \rho} = 1 \non \ee
Hence one can write $\ep = 1 - \rho$ and conclude the proof.
\end{proof}

\sbp The entanglement entropy for the initial state is then given by
{\alld
\begin{align}
S(\rho_1) = - \sum_{n=0} ^{\infty} \lb_n \ln \lb_n = - \ln \left( 1- \frac{\nu}{2(\vp + \vsg)} \right) - \frac{\frac{\nu}{2(\vp + \vsg)}}{1- \frac{\nu}{2(\vp + \vsg)}} \ln \left( \frac{\nu}{2(\vp + \vsg)} \right)
\end{align}}

The entropy is plotted in Figure~\ref{VNE0}

\begin{figure}[h]
 	\begin{center}
 		\includegraphics[scale=0.75]{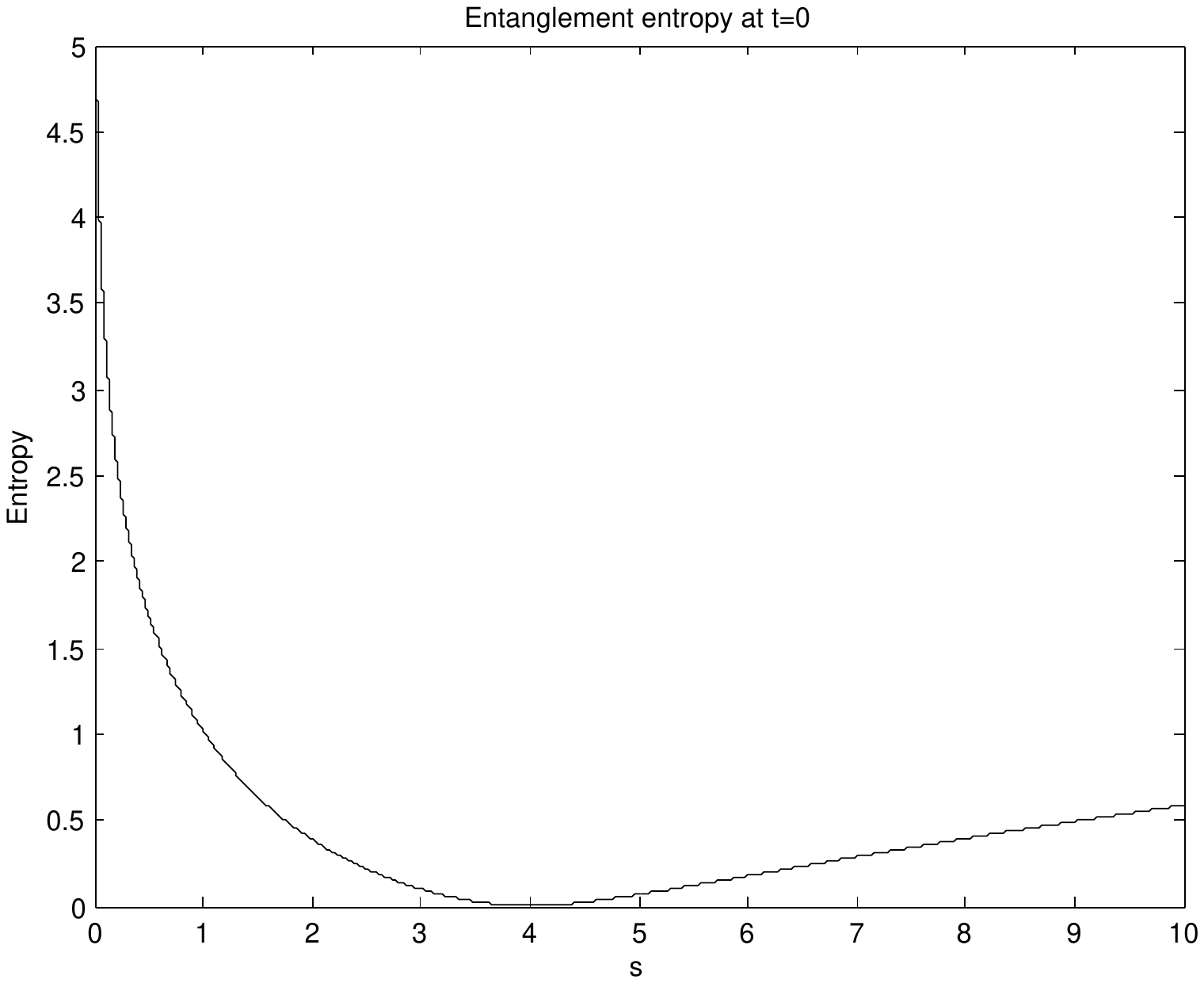}
	\end{center}
	\caption{Entanglement entropy for the state at $t=0$ and with $d=2$}
	\label{VNE0}
\end{figure}

As expected from the original state, the entanglement entropy vanishes for $s = 2d$ and increases again as the separation $s$ increases.

\begin{proposition}
The Von Neumann entropy is invariant under closed system dynamics.
\end{proposition}

\begin{proof} As a short proof, let us recall that since
\be \rho(t) = e^{\im H t} \rho_0 e^{- \im H t}\non \ee
we have
\be  S(t) = \Tr \left[ - \rho(t) \ln(\rho(t)) \right] \non \ee
Let us call $f(\rho(t)) = \rho(t) \ln (\rho(t))$ and apply the Weierstrass approximation theorem \cite{WeierTH:2006} so that we can write, in terms of an arbitrary smallness parameter $\ep$,
\begin{align}
f(\rho(t)) &\sim \sum_{k=0}^n a_k \rho(t)^k \non 
\\ \rho(t)^k &= e^{\im H t} \rho_0 e^{-\im H t} e^{\im H t} \rho_0 e^{-\im H t} ... e^{\im H t} \rho_0 e^{-\im H t} \non 
\\ &=e^{\im H t} \rho_0^k e^{-\im H t} \non
\end{align}
which as $\ep \rightarrow 0$ becomes $f(\rho(t)) = e^{\im H t} f(\rho_0) e^{-\im H t}$. Then we can write
\begin{align}
 S(t) =& \Tr \left[ - \rho(t) \ln(\rho(t)) \right] \non 
\\ =& - \Tr \left[ e^{\im H t} \rho_0 \ln(\rho_0) e^{- \im H t} \right] \non 
\\ S(t)=& - \Tr \left[ \rho_0 \ln(\rho_0) \right]
\end{align}
which concludes this proof.
\end{proof}

\sbp The Von Neumann entropy is well suited to pure entangled states but does not suffice to quantify the entanglement of mixed states, such as those obtained through open system dynamics \cite{Vedral:1997}. The logarithmic negativity is a common choice for Gaussian states amongst the many entanglement measures that have been proposed \cite{Mintert:2005, Woot:1998, Meyer:2002} because it is easy to compute, especially in the covariance matrix formalism that will be presented in the next section. 

\section{General Gaussian states}

\sbp Continuous variables are of growing interest in the field of quantum optics and Gaussian states are probably the most widely used states for such variables. Indeed, while theoretically easy to handle, the latter can also be experimentally prepared and manipulated. A common experimental representation of Gaussian states would be modes of light. Another major advantage of Gaussian states is their relatively simple mathematical formulation and the fact that this formulation allows for explicit calculations. 

\subsection*{Continuous variables}

\sbp A quantum system of N particles, each with one degree of freedom, has position and momentum variables $R_j$ satisfying (with $\hb=1$)
\be \left[ {\hat R}_j, {\hat R}_k \right] = \imath \sg_{jk} \mathbb{1} \label{Rcomrlt} \quad j,k = 1...2N\ee
where ${\hat R} = ({\hat x}_1, {\hat p}_1, ... , {\hat x}_N, {\hat p}_N) ^T$, ${\hat x}$ and ${\hat p}$ are the usual canonical position and momentum and $\sg$ is the fixed, non-singular, skew symmetric matrix defined as
\be \sg =  \bigoplus _{j=1} ^N \left( \begin{array}{cc} 0 & 1 \\ -1 & 0 \end{array} \right) \non \ee
One can note two additional properties of $\sg$, namely $\det(\sg)=1$ and $\sg^{-1}= \sg^T = -\sg$.

\sbp The state of the system is best described by its density matrix. The first moments are collected in the displacement vector $\tbf{d}$ as
\be d_j \equiv \Tr [ \rho {\hat R}_j]  \equiv \langle {\hat R}_j \rangle\ee
The second moments are collected in the $2N \times 2N$ real, symmetric covariance matrix $\gm$ through
\begin{align} \label{cm} \gm_{jk} &= \Tr \left[ \rho \left[ ( {\hat R}_j - \langle {\hat R}_j \rangle ), ({\hat R}_k - \langle {\hat R}_k \rangle) \right]_+ \right]
\\ \gm_{jk} + \imath \sg_{jk} &= 2\, \Tr [\rho({\hat R}_j - \langle {\hat R}_j \rangle )({\hat R}_k - \langle {\hat R}_k \rangle)]
\end{align}
(Note that the imaginary part is equal to $\sg$ by the commutation relations (\ref{Rcomrlt})). Not all such matrices are proper covariance matrices ; by (\ref{cm}), they must in addition satisfy $ \gm + \imath \sg \geq 0$, i.e. $\gm + \imath \sg$ must be positive definite. 
The covariance matrix $\gm$ is in particular given by
\be \label{cm2} \gm_{jk} = 2 \Rl \hspace{5 pt} \Tr \left[\rho({\hat R}_j - \langle {\hat R}_j \rangle )({\hat R}_k - \langle {\hat R}_k \rangle)\right] \ee

It may be useful to introduce the Weyl operators $ {\hat W}_{{\hat R},\xi} = e^{\imath \, \xi^T \sg {\hat R}} $, in terms of which a state $\rho$ of N modes can be expressed as 
\be \rho_{{\hat R}} = \frac{1}{(2\pi) ^N} \int _{\mathbb{R}^{\,2N}} d^{2N}\, \xi \hspace{5 pt}\chi(-\xi) \hspace{5 pt} {\hat W}_{{\hat R},\xi} \ee
where $\xi$ is a vector over the phase space $\mathbb{R}^{\,2N}$. $\chi(-\xi)$ is called the characteristic function and is defined as $\chi(\xi) = \Tr \left[\rho {\hat W}_{\xi} \right]$

\subsection*{Symplectic transformations}

\sbp A transformation S on a quantum mechanical state is called \textit{symplectic} if it leaves the canonical commutation relations unchanged. So if 
\be S : R \rightarrow R' = SR \non \ee 
is a real linear transformation such that $\left[R'_j, R'_k\right] = \imath \sg_{jk} \mathbb{1}$, then $S \sg S^T = \sg$. 

\begin{theorem}Any real, symmetric, positive definite matrix A can be transformed into its diagonal form (the so-called Williamson normal form) via a symplectic transformation S
\be A_{WF} = S A S^T = \diag(a_1, a_1, ... , a_N, a_N) \ee
where the $a_j$'s are the symplectic eigenvalues of A.
\end{theorem} 

\sbp The proof of this theorem can be found in \cite{Will:1936}. The symplectic eigenvalues $a_j$ can be calculated as the positive eigenvalues of $\imath \sg A$. In fact, using $\sg^{-1} = -\sg$
\begin{align} \sg^{-1} &= {S^{-1}}^T \sg^{-1} S^{-1} \non \\ S^T \sg^{-1} &= \sg^{-1} S^{-1} \non \\ S^T \sg &= \sg S^{-1} \end{align}
\begin{align} \eig( \im \sg A ) =& \eig( \im \sg S^{-1} A_{WF} {S^{-1}}^T ) = \eig( \im S^T \sg A_{WF} {S^{-1}}^T) = \eig(\im \sg A_{WF} ) 
\end{align}

We can also write
\begin{align}
 A V_{\lb} =& \lb V_{\lb} \non \\ A^2 V_{\lb} =& A A V_{\lb} = \lb A V_{\lb} = \lb ^2 V_{\lb}
\end{align}
enabling us to see that the $a_j$'s can also be easily calculated as the positive square root of the matrix $-\sg A \sg A$.

\subsection*{Gaussian States}

\sbp A Gaussian state $\rho$ with N modes is a state whose characteristic function $\chi_{\rho}(\xi)$ can be written as 
\be \label{cfn} \chi_{\rho}(\xi) = \exp \left[ -\frac{1}{4} \xi^T \, \Gm \, \xi + \imath \tbf{D}^T \, \xi \right] \ee
where $\Gm= \sg \gm \sg^T$ is the covariance matrix of the state and $\tbf{D}= \sg \tbf{d}$ the displacement vector. Considering \be S (\gm + \im \sg) S^T = S \gm S^T + \im \sg \geq 0 \ee
one can write the Heisenberg uncertainty relation for $\Gm$ as $\Gm + \im \sg \geq 0$. It follows that it can be brought to Williamson normal form by symplectic transformation without the symplectic eigenvalues
\be  S [\Gm + \im\sg] S^T = \Gm_{WF} + \imath \sg \non  = \bigoplus_j \left[ \begin{array}{cc} \Gm_j & \im \\ -\im & \Gm_j \end{array} \right] \ee
Since the matrix must be positive, all the eigenvalues $\mu_k$, $k=1, ... , 2N$ must be positive. They are determined by the usual characteristic polynomial
\be  \prod_{j=1}^N \left[ (\Gm_j - \mu_k )^2 - 1 \right] =0 \label{proof1} \ee
which one can use to determine that
\begin{align}
 (\Gm_j - \mu_k)^2 - 1 =& 0 \non 
\\ \Gm_j - \mu_k =& \pm 1 \non 
\\ \mu_k =& \Gm_j \pm 1 \geq 0 \quad \forall k=1 ... 2N 
\end{align}
Hence the symplectic eigenvalues of $\Gm$ are 
\be \Gm_j \geq 1 \quad \forall j = 1, ... , N \label{proof2}\ee

Thus the symplectic eigenvalues of a $\Gm$ satisfying the Heisenberg uncertainty relations are greater or equal to one. Since the symplectic eigenvalues of $\Gm$ can be calculated as the positive square root of $\lb_i$ where $\lb_i$ are the eigenvalues of $- \sg \gm \sg \gm$, then the $\lb_i \geq 1$.

 Another important property is that since the first moments of a Gaussian state can always be made to vanish via local operations, they are irrelevant in the determination of entanglement. Hence for the remainder of the present work, covariance matrices will be expressed as follows
\be \gm_{jk} = 2\Rl \hspace{5 pt} \Tr\left[ \rho {\hat R}_j {\hat R}_k \right] \ee
It should also be noted that due to  conflicting notation, the letter $\gm$ will in the subsequent chapters refer solely to the coupling constant, whereas  G will be used to designate the covariance matrix.

\subsection*{Separability}
\sbp We recall (\ref{cdt1}) to determine if a state is separable. In terms of covariance matrices, separability of Gaussian states can be written as 

\begin{theorem} A Gaussian state with covariance matrix $\gm$ is separable if there exists covariance matrices $\gm_1$ and $\gm_2$ such that
\be \gm \geq \left( \begin{array}{cc} \gm_1 & 0 \\ 0 & \gm_2 \end{array} \right)\ee
\end{theorem}

\sbp A proof of this theorem can be found in \cite{Anders:2003}. For low-dimensional systems, one can also check the positivity of the partial transpose (PPT). Partial transposition effectively results in time reversal for one of the particles'momentum operators. Hence if we have a bipartite system with operators ${\hat X}_1, {\hat P}_1, {\hat X}_2, {\hat P}_2$, partial transposition over the first particle results in sending ${\hat P}_1$ to $-{\hat P}_1$. With this, one can write \cite{Anders:2003}

\begin{theorem}Let $\Gm^{T_A}$ be the partially transposed covariance matrix of a state $\rho$. If $\Gm^{T_A}$ fails to fulfill the Heisenberg uncertainty relation, then the state is entangled.
\end{theorem}

\sbp As was seen in the previous section, satisfying the Heisenberg uncertainty relation is equivalent to restricting the symplectic eigenvalues to values greater than one. Hence the above theorem can written a follows.

\begin{theorem}Let $\Gm^{T_A}$ be the partially transposed covariance matrix of a state $\rho$. If one of its symplectic eigenvalues is less than 1, then the state is entangled.
\end{theorem}

\sbp For more formal definitions, one should refer to \cite{Anders:2003, Werner:2001, Simon:2000, Werner:2002}. Although the PPT criterion is not sufficient to completely establish entanglement when the system is of dimensions greater than $3 \times 3$ \cite{Werner:2001}, these conditions allow us to determine qualitatively whether the state studied is entangled. To quantify the degree of entanglement in a system, one requires an entanglement measure.

If we write a Gaussian state in terms of its covariance matrix and limit our study to second moments, we have $\gm_{jk}=2\Rl \hspace{5 pt} \Tr[\rho{\hat R}_j{\hat R}_k]$. Partial transposition may result in one or more eigenvalues to be less than 1, so that the resulting covariance matrix may not be positive. Yet, it can still be brought to Williamson normal form \cite{Werner:2002}, so that its symplectic eigenvalues $\lb_i ^T$ can be calculated using $-\sg \gm ^{T_1} \sg \gm ^{T_1}$. The logarithmic negativity is then defined by 

\be \mathcal{L}_{\mathcal{N}}(\rho) = - \sum _{i=1} ^{n+m} \log_2 \, (\min\,(1\, ,\mid \lb_i ^{T_1} \mid)) \ee

This expression makes $\mathcal{L}_{\mathcal{N}}(\rho)$ easy to calculate and will be used for the remainder of the present work. It is again easy to see that since for a separable state $\lb_i ^{T_1} \geq 1$, $\mathcal{L}_{\mathcal{N}}(\rho) = 0$.


		\chapter*{The Master Equation}
\stepcounter{chapter}
\addcontentsline{toc}{chapter}{The Master Equation}

\sbp In this chapter, the Non-Rotating Wave master equation is derived for a single particle system. The method itself follows that of Gardiner et al. \cite{Gard:NRWME, Gard:2000}, but uses perturbation theory as in \cite{FLO'C:1996} as opposed to the Van Kampen cumulant expansion \cite{VKampen:1982}. It starts by deriving the Quantum Langevin Equation as done in \cite{FLO'C:QLE1988, Gard:NRWME, Gard:2000}. The adjoint equation follows and finally the perturbation method is used to arrive at the master equation expressed in the Non-Rotating Wave approximation and quantum Brownian limit. The second part of this chapter gives a solution to the master equation for a free-particle system Hamiltonian.

\section{Derivation of the Master Equation}

\subsection*{The Quantum Langevin Equation}

\sbp To derive the Quantum Langevin Equation, we follow \cite{FLO'C:QLE1988}. Let us consider an independent oscillator heat bath, which we couple to a particle. The resulting Hamiltonian has the following form
\begin{align}
 H =& \frac{p ^2}{2m} + V(x) + \frac{1}{2} \sum_j \left\{ \frac{p_j ^2}{m_j} + m_j \og_j ^2 (q_j - x)^2 \right\} 
\end{align}
where the sum is over all of the bath's degrees of freedom. The bath operators$p_j$ and $q_j$ satisfy the commutations relations $[ p_j , Y ] = [q_j, Y ] = 0$.
The Heisenberg equations of motion are (with $H_s = \frac{p ^2}{2m} + V(x)$)
{\alld
\begin{align}
 {\dot x} =& \frac{\im}{\hb} \left[ H_s , x \right] = \frac{p}{m} \non 
\\ \label{qle_interm} {\dot p} =& \frac{\im}{\hb} \left[ H_s , p \right] = - V'(x) + \sum_j m_j \og_j ^2 (q_j - x) 
\\ {\dot q}_j =& \frac{\im}{\hb} \left[ H_s , q_j \right] = \frac{p_j}{m_j} \non 
\\ {\dot p}_j =& \frac{\im}{\hb} \left[ H_s , p_j \right] = - m_j \og_j ^2 (q_j - x) 
\end{align}}
It follows that
\begin{align}
 m_j {\ddot q}_j =& - m_j \og_j ^2 (q_j - x)
\end{align}
This equation can be solved for $q_j (t)$ in terms of $x(t)$
\begin{align}
 q_j (t) =& q_j ^h(t) + x (t) - \int_{-\infty} ^t \cos\left[\og_j (t-t')\right] {\dot x} (t') \, dt' 
\\ q_j ^h (t) =& q_j \cos(\og_j t) + \frac{p_j}{\og_j m_j} \sin (\og_j t)
\end{align}
Introducing
\begin{align}
 \mu (t) =& \sum_j m_j \og_j ^2 \cos(\og_j t) \Theta(t) 
\\ \xi(t) =& \sum_j m_j \og_j ^2 q_j ^h (t)
\end{align}
and inserting it in (\ref{qle_interm}) yields the Quantum Langevin Equation
\begin{align}
m {\ddot x} + \int_{-\infty} ^t \mu (t-t') {\dot x}(t')\, dt' + V'(x) = \, \xi(t) 
\end{align}
where the dot denotes the derivative with respect to time and the prime that with respect to $x$. $\mu(t)$ and $\xi(t)$ describe the interaction with the system as, respectively, a memory function and an operator-valued random force. $\Theta(t)$ is the Heaviside function. 

At $t \rightarrow - \infty$, one can assume the bath to be in thermal equilibrium at temperature T. Then one can write
\begin{align}
 \langle q_j q_k \rangle =& \frac{\hb}{2 m_j \og_j} \coth \left( \frac{\hb \og_j}{2 k T} \right) \dt_{jk} \non 
\\ \langle p_j p_k \rangle =& \frac{\hb m_j \og_j}{2} \coth \left( \frac{\hb \og_j}{2 k T} \right) \dt_{jk} \non 
\\ \langle q_j p_k \rangle =& - \langle p_j q_k \rangle = \frac{\im\hb}{2} \dt_{jk}
\end{align}
The full derivation of the above relations is recalled in Appendix~\ref{thermal-equilibrium-relations}. The autocorrelation of the random force, i.e. the expectation of the anti-commutator can then be written as (the $[ \, , \, ]_+$ denote anticommutation) 
{\alld
\begin{align}
 \lefteqn{\half \langle \left[ \xi (t), \xi (t') \right]_+ \rangle} \non 
\\ &\hspace{-0.1 in} = \half \langle \sum_j \sum_{j'} m_j ^2 \og_j ^4 \left[ q_j ^h(t), q_{j'} ^h (t') \right]_+ \rangle \non 
\\ &\hspace{-0.1 in} = \half \langle  \sum_j \sum_{j'} m_j ^2 \og_j ^4 \Bigl( \left( q_j \cos(\og_j t) + \frac{p_j}{\og_j m_j} \sin (\og_j t) \right) \left( q_{j'} \cos(\og_{j'} t') + \frac{p_{j'}}{\og_{j'} m_{j'}} \sin (\og_{j'} t') \right) \non 
\\*  &\hspace{0.5 in} + \left( q_{j'} \cos(\og_{j'} t') + \frac{p_{j'}}{\og_{j'} m_{j'}} \sin (\og_{j'} t') \right) \left( q_j \cos(\og_j t) + \frac{p_j}{\og_j m_j} \sin (\og_j t) \right) \Bigr) \rangle \non 
\\  &\hspace{-0.1 in}= \half \sum_j m_j ^2 \og_j ^4 \Bigl( 2 \langle q_j ^2 \rangle \cos (\og _j t) \cos (\og_j t') + 2 \frac{\langle p_j ^2 \rangle}{m_j ^2 \og_j ^2} \sin ( \og_j t) \sin (\og_j t') \non
\\*  &\hspace{0.5 in} + \frac{\langle q_j p_j \rangle}{m_j \og_j}\left( \cos (\og_j t) \sin (\og_j t') + \cos (\og_j t') \sin (\og_j t) \right) \non 
\\*  &\hspace{0.5 in} + \frac{\langle p_j q_j \rangle}{m_j \og_j}\left( \cos (\og_j t') \sin (\og_j t) + \cos (\og_j t) \sin (\og_j t') \right) \Bigr)  \non
\\  &\hspace{-0.1 in} = \half  \sum_j m_j ^2 \og_j ^4 \Bigl( 2 \frac{\hb}{2 m_j \og_j} \coth \left( \frac{\hb \og_j}{2 k T} \right) \cos \left[\og_j (t-t') \right] \non \\* & \hspace{0.5 in} + \frac{\im \hb}{2 m_j \og_j} \sin \left[ \og_j (t+t') \right] - \frac{\im \hb}{2 m_j \og_j} \sin \left[ \og_j (t+t') \right] \Bigr) \non 
\end{align}}
to get
\be
\half \langle \left[ \xi (t), \xi (t') \right]_+ \rangle = \half \sum_j \hb m_j \og_j ^3 \coth \left( \frac{\hb \og_j}{2 k T} \right) \cos \left[ \og_j (t-t')\right]
\ee
If one introduces ${\tilde \mu}$, the Fourier trasnform of the memory function as
\be {\tilde \mu}(z) = \int_0^{\infty} dt e^{\im z t} \mu(z)  \non \ee 
and the \textit{spectral distribution} $G(\og) = \Re \left[ {\tilde \mu} \left( \og + \im 0^+ \right) \right]$ as
\begin{align}
 G(\og) = \frac{\pi}{2} \sum_j m_j \og_j ^2 \left[ \dt(\og - \og_j) + \dt (\og + \og_j) \right] 
\end{align}
one can write the autocorrelation of $\xi(t)$ as
\begin{align}
 \frac{1}{2} \langle \left[ \xi(t), \xi(t') \right]_+ \rangle =& \frac{1}{\pi} \int_0^{\infty} G(\og) \hb \og \coth \left( \frac{\hb \og}{2 k T} \right) \cos \left[\og (t-t') \right] \, d\og 
\end{align}
For an arbitrary observable Y of the small system, the Heisenberg equations of motion read 
{\alld
\begin{align}
 {\dot Y} =  \frac{\im}{\hb} \left[ H_s, Y \right] &+ \frac{\im}{2\hb} \sum_j \left\{ \frac{1}{m_j} \left[ p_j ^2, Y \right] + m_j \og_j ^2 \left[(q_j - x)^2 , Y \right] \right\} \non 
\\ = \frac{\im}{\hb} \left[ H_s, Y \right] &+ \frac{\im}{2\hb} \sum_j \left\{ m_j \og_j ^2 \left[ \left[ q_j - x , Y \right], q_j - x \right]_+ \right\} \non 
\\ = \frac{\im}{\hb} \left[ H_s, Y \right] &- \frac{\im}{2\hb} \sum_j \left\{ m_j \og_j ^2 \left[ \left[ x , Y \right], q_j - x \right]_+ \right\}\non 
\\ = \frac{\im}{\hb} \left[ H_s, Y \right] &- \frac{\im}{2\hb} \sum_j \left\{ m_j \og_j ^2 \left[ \left[ x , Y \right], q_j ^h(t) \right]_+ \right\}  \non 
\\* & + \frac{\im}{2 \hb}  \sum_j \left\{ m_j \og_j ^2 \left[ \left[ x , Y \right] , \int_{-\infty} ^t dt' \, \cos \left[ \og_j (t-t')\right] {\dot x} (t) \right]_+ \right\} \non 
\\ = \frac{\im}{\hb} \left[ H_s, Y \right] &- \frac{\im}{2\hb} \left[ \left[ x ,Y \right], \xi(t) \right]_+ + \frac{\imath}{2\hb} \left[ \left[ x , Y \right] , \int_{-\infty} ^t \mu(t-t') {\dot x} (t') \, dt'\right]_+ 
\end{align}}
If we work with a Ohmic heat bath \cite{FLO'C:QLE1988}, then 
\begin{align}
 \int_{-\infty} ^t \mu(t') {\dot x}(t') \, dt' \rightarrow \gm {\dot x}(t) \hspace{0.1 in} \text{and} \hspace{0.1 in}  G(\og) \rightarrow \gm
\end{align}
so that the Quantum Langevin Equation for the observable Y is
\begin{align} \label{qle}
 {\dot Y} =& \frac{\im}{\hb} \left[ H_s, Y \right] - \frac{\im}{2\hb} \left[ \left[ x ,Y \right], \xi(t) \right]_+ + \frac{\im}{2\hb} \left[ \left[ x , Y \right] , \gm {\dot x} (t) \right]_+ 
\end{align}

\subsection*{The adjoint equation}

\sbp The Quantum Langevin Equation is an equation for the system operators, whereas a master equation is an (approximate) equation acting on the density operator of the quantum system under study. The adjoint equation provides a link between the two formalisms, being an exact equation upon which approximations can be made to obtain the required master equation. One can write $\rho(t) = \nu(t) \rho_B$ where $\nu(t)$ is the density matrix of the small system and $\rho_B$ that of the bath.
One defines
\begin{equation} \label{trs}
 \Tr_s \left\{ Y(t) \rho \right\} = \Tr_s \left\{ Y \rho(t) \right\}
\end{equation}
where $\Tr_s$ is the trace over the (small) system and $Y(t)$ is a random system observable. One then applies (\ref{trs}) to (\ref{qle}). Term by term analysis yields

\begin{itemize}
	\item{
\be
\frac{\im}{\hb} \Tr_s \left\{ \left[ H_s , Y(t) \right] \rho \right\} = \frac{\im}{\hb} \Tr_s \left\{ \left[ H_s , Y \right] \rho (t) \right\} = - \frac{\im}{\hb} \Tr_s \left\{ Y \left[ H_s , \rho (t) \right] \right\} 
\ee}

	\item{{\alld
\begin{align} \label{xi1}
 \lefteqn{\frac{\im}{2\hb} \Tr_s \left\{ \left[ \left[ x(t) , Y(t) \right], \xi(t) \right]_+ \rho \right\}} \non \\& \hspace{0.25 in} =\frac{\im}{2\hb} \Tr_s \left\{ \left[ \left[ x , Y \right], \xi(t) \right]_+ \rho (t) \right\} \non 
\\& \hspace{0.25 in} = \frac{\im}{2\hb} \Tr_s \left\{ x Y \xi(t) \rho(t) - Y x \xi (t) \rho(t) + \xi (t) x Y \rho (t) - \xi (t) Y x \rho (t) \right\} \non 
\\& \hspace{0.25 in} = \frac{\im}{2\hb} \Tr_s \left\{ Y \left[ \left[ \xi (t) , \rho(t) \right]_+ , x \right] \right\}
\end{align}}}

	\item{{\alld
\begin{align}
 \lefteqn{\frac{\im}{2\hb} \Tr_s \left\{ \left[ \left[ x (t) , Y(t) \right], \gm {\dot x}(t) \right]_+ \rho \right\}} \non \\& \hspace{0.25 in} =\frac{\im}{2\hb} \Tr_s \left\{ \left[ \left[ x , Y \right], \gm {\dot x} \right]_+ \rho (t) \right\} \non 
\\& \hspace{0.25 in} = \frac{\im}{2\hb} \Tr_s \left\{ x Y \gm {\dot x} \rho(t) - Y x \gm {\dot x} \rho(t) + \gm {\dot x} x Y \rho(t) - \gm {\dot x} Y x \rho(t) \right\} \non 
\\& \hspace{0.25 in}  = \frac{\im \gm}{2\hb} \Tr_s \left\{ Y \left[ \left[ {\dot x} , \rho(t)\right]_+, x \right] \right\}
\end{align}}}
\end{itemize}

\subparagraph*{Nota Bene} \textit{The transition from the first line to the second line in (\ref{xi1}) may require some explaining. The trace is over the system's variables ; however, $\xi(t)$ is a bath operator, since it effectively represents the noise. Hence it remains unaffected by the trace and retains its time dependency.}

\sbp The adjoint equation can finally be written as
\be \label{adj}
 {\dot \rho}(t) = - \frac{\im}{\hb} \left[ H_s , \rho(t) \right] - \frac{\im}{2\hb} \left[ \left[ \xi (t) , \rho(t) \right]_+ , x \right] + \frac{\im \gm}{2\hb} \left[ \left[ {\dot x} , \rho(t) \right]_+ , x \right] 
\ee

\subsection*{The master equation}

\sbp To derive the master equation from the adjoint equation, we follow \cite{FLO'C:1996}. The noise is assumed to be small. This assumption is not essential but merely allows for a simple derivation. To indicate this, we temporarily introduce a small parameter $\ep$ and replace $ \xi(t) \rightarrow \ep \xi(t)$. One can then write $\rho(t)$ to second order in $\ep$ as
\begin{equation}
 \rho(t) = \rho_0(t) + \ep \rho_1(t) + \ep ^2 \rho_2 (t) \non
\end{equation}
The bath and the system are also assumed to be initially decoupled at $t \rightarrow - \infty$ and the bath to be large so that it remains at thermal equilibrium throughout. Thus one can write $\rho_0(t) = \nu_0(t) \rho_B$, where $\rho_B$ is the equilibrium state of the bath at temperature T. This assumption is essential to the derivation. Then
\begin{equation}
 {\dot \rho}(t) = {\dot \rho}_0 (t) + \ep {\dot \rho}_1 (t) + \ep ^2 {\dot \rho}_2 (t) \non 
\end{equation} 
The expansion of (\ref{adj}) yields 
{\alld
\begin{align} \label{no-order}
 {\dot \rho}(t) = & - \frac{\imath}{\hb} \left[ H_s , \rho_0(t) \right] - \frac{\imath}{2\hb} \, \ep \left[ \left[ \xi (t) , \rho_0(t) \right]_+ , x \right] + \frac{\imath\gm}{2\hb} \left[ \left[ {\dot x} , \rho_0(t) \right]_+ , x \right] \non
\\ & - \frac{\imath}{\hb} \, \ep \left[ H_s , \rho_1(t) \right] - \frac{\imath}{2\hb} \, \ep^2 \left[ \left[ \xi (t) , \rho_1(t) \right]_+ , x \right] + \frac{\imath\gm}{2\hb} \, \ep \left[ \left[ {\dot x} , \rho_1(t) \right]_+ , x \right] \non 
\\ & - \frac{\imath}{\hb} \, \ep ^2 \left[ H_s , \rho_2(t) \right] + \frac{\imath\gm}{2\hb} \, \ep ^2 \left[ \left[ {\dot x} , \rho_2(t) \right]_+ , x \right] 
\end{align}}
Since $\rho_0(t) = \nu_0(t) \rho_B$, (\ref{no-order}) becomes (with reordered terms)
{\alld
\begin{align}
 {\dot \rho}(t) = & - \frac{\imath}{\hb} \left[ H_s , \nu_0(t) \right]\rho_B + \frac{\imath\gm}{2\hb} \left[ \left[ {\dot x} , \nu_0(t) \right]_+ , x \right] \rho_B \non
\\ & - \frac{\imath}{\hb} \, \ep \left[ H_s , \rho_1(t) \right] - \frac{\imath}{2\hb} \, \ep \left[ \left[ \xi (t) , \nu_0(t) \rho_B \right]_+ , x \right] + \frac{\imath\gm}{2\hb} \, \ep \left[ \left[ {\dot x} , \rho_1(t) \right]_+ , x \right] \non 
\\ & - \frac{\imath}{\hb} \, \ep ^2 \left[ H_s , \rho_2(t) \right] - \frac{\imath}{2\hb} \, \ep^2 \left[ \left[ \xi (t) , \rho_1(t) \right]_+ , x \right] + \frac{\imath\gm}{2\hb} \, \ep ^2 \left[ \left[ {\dot x} , \rho_2(t) \right]_+ , x \right] 
\end{align}}
Let us study each order separately. The zeroth order yields
\be {\dot \rho}_0(t) = {\dot \nu}_0 (t) \rho_B \rightarrow \Tr_B \left\{ {\dot \rho}_0(t) \right\} = {\dot \nu}_0(t) \ee
so that $\nu_0(t)$ satisfies the following equation
\be {\dot \nu}_0(t) = - \frac{\imath}{\hb} \left[ H_s , \nu_0(t) \right] + \frac{\imath\gm}{2\hb} \left[ \left[ {\dot x} , \nu_0(t) \right]_+ , x \right] \ee

The first order term is
\be
 {\dot \rho}_1(t) = - \frac{\imath}{\hb} \left[ H_s , \rho_1(t) \right] - \frac{\imath}{2\hb} \left[ \left[ \xi (t) , \nu_0(t) \rho_B \right]_+ , x \right] + \frac{\imath\gm}{2\hb} \left[ \left[ {\dot x} , \rho_1(t) \right]_+ , x \right]
\ee
Since $\nu_0$ is a system operator and $\xi$ is a bath operator, one can rewrite $\left[ \left[ \xi (t) , \nu_0(t) \rho_B \right]_+ , x \right] $ as 
\be \left[ \left[ \xi (t) , \nu_0(t) \rho_B \right]_+ , x \right] = \left[ \left[ \xi(t) , \rho_B \right]_+ \nu_0, x \right] = \left[ \xi(t) , \rho_B \right]_+ \left[ \nu_0(t) , x \right] \non \ee
so that the first order term becomes
\be
 {\dot \rho}_1(t) = - \frac{\imath}{\hb} \left[ H_s , \rho_1(t) \right] + \frac{\imath\gm}{2\hb} \left[ \left[ {\dot x} , \rho_1(t) \right]_+ , x \right] - \frac{\imath}{2\hb} \left[ \xi (t) ,\rho_B \right]_+ \left[ \nu_0(t) , x \right]  
\ee
This can be written in the form
\be {\dot \rho}_1 (t) = A_s \rho_1 (t) + f(t) \ee 
where 
\be f(t) = - \frac{\imath}{2\hb} \left[ \xi (t) ,\rho_B \right]_+ \left[ \nu_0(t) , x \right]  \ee
and $A_s$ is a "super operator" given by
\be A_s \rho_1 (t) = - \frac{\imath}{\hb} \left[ H_s , \rho_1(t) \right] + \frac{\imath\gm}{2\hb} \left[ \left[ {\dot x} , \rho_1(t) \right]_+ , x \right] \ee
Since $\Tr_B \left\{ \xi(t) \rho_B \right\} = 0$, it can be seen that $\Tr_B \left\{ \rho_1(t) \right\}$ satisfies the same equation as $ \nu_0(t)$, but with vanishing initial condition. It is then easy to see that $\rho_1(t)$ can be written as
\begin{align} \label{rho1}
  \rho_1(t) =& \int_{-\infty} ^t e^{A_s(t-t')} f(t') \, dt' \non
\\ =& - \frac{\imath}{2\hb} \int_{-\infty} ^t e^{A_s(t-t')} \left[ \nu_0(t') , x \right] \left[ \xi (t') ,\rho_B \right]_+ \, dt' 
\end{align}

The second order term is
\be \label{drho2}
 {\dot \rho}_2(t) =  - \frac{\imath}{\hb} \left[ H_s , \rho_2(t) \right] - \frac{\imath}{2\hb} \left[ \left[ \xi (t) , \rho_1(t) \right]_+ , x \right] + \frac{\imath\gm}{2\hb} \left[ \left[ {\dot x} , \rho_2(t) \right]_+ , x \right]
\ee
Using (\ref{rho1})
\begin{align}
 \left[\left[\xi(t) , \rho_1(t) \right]_+ , x \right] =& - \frac{\imath}{2\hb} \int_{-\infty} ^t \left[ \left[ \xi(t) , e^{A_s(t-t')} \left[ \nu_0(t') , x \right] \left[ \xi (t') ,\rho_B \right]_+ \right]_+ , x \right]\, dt' \non 
\\ =& - \frac{\imath}{2\hb} \int_{-\infty} ^t \left[ \left[ \xi(t), \left[ \xi (t') ,\rho_B \right]_+ \right]_+ e^{A_s(t-t')} \left[ \nu_0(t') , x \right] , x \right] \, dt' \non 
\\ =& - \frac{\imath}{2\hb} \int_{-\infty} ^t  \left[ \xi(t), \left[ \xi (t') ,\rho_B \right]_+ \right]_+ \left[ e^{A_s(t-t')} \left[ \nu_0(t') , x \right] , x \right] \, dt'
\end{align}
Inserting into (\ref{drho2}) yields
\begin{align}
 {\dot \rho}_2 (t) =& - \frac{\imath}{\hb} \left[ H_s , \rho_2(t) \right] + \frac{\imath\gm}{2\hb} \left[ \left[ {\dot x} , \rho_2(t) \right]_+ , x \right] \non 
\\ & - \frac{1}{4\hb ^2} \int_{-\infty} ^t \left[ e^{A_s(t-t')} \left[ \nu_0(t') , x \right] , x \right] \left[ \xi(t), \left[ \xi (t') ,\rho_B \right]_+ \right]_+  \, dt' 
\end{align}
Taking the trace over the bath
{\alld
\begin{align} \label{dnu2}
 {\dot \nu}_2 (t) =& - \frac{\imath}{\hb} \Tr_B \left\{ \left[ H_s , \rho_2(t) \right] \right\} + \frac{\imath\gm}{2\hb} Tr_B \left\{\left[ \left[ {\dot x} , \rho_2(t) \right]_+ , x \right] \right\} \non 
\\ & - \frac{1}{4\hb ^2} \int_{-\infty} ^t \Tr_B \left\{ \left[ e^{A_s(t-t')} \left[ \nu_0(t') , x \right] , x \right] \left[ \xi(t), \left[ \xi (t') ,\rho_B \right]_+ \right]_+ \right\} \, dt' \non 
\\ =& - \frac{\imath}{\hb} \left[ H_s , \nu_2(t) \right] + \frac{\imath\gm}{2\hb} \left[ \left[ {\dot x} , \nu_2(t) \right]_+ , x \right] \non 
\\ & - \frac{1}{4\hb ^2} \int_{-\infty} ^t \left[ e^{A_s(t-t')} \left[ \nu_0(t') , x \right] , x \right]  \Tr_B \left\{ \left[\xi(t), \left[\xi (t') ,\rho_B \right]_+ \right]_+ \right\} \, dt' 
\end{align}}
Since
{\alld
\begin{align}
 \Tr_B \left\{ \left[ \xi(t) , \left[ \xi(t'), \rho_B\right]_+ \right]_+ \right\} = 2 \Tr_B \left\{ \left[ \xi(t), \xi(t') \right]_+ \rho_B \right\} = 2 \langle \left[ \xi(t), \xi(t') \right]_+ \rangle
\end{align} }
(\ref{dnu2}) becomes
\begin{align}
 {\dot \nu}_2 (t) =& - \frac{\imath}{\hb} \left[ H_s , \nu_2(t) \right] + \frac{\imath\gm}{2\hb} \left[ \left[ {\dot x} , \nu_2(t) \right]_+ , x \right] \non 
\\ & - \frac{1}{\hb ^2} \int_{-\infty} ^t \left[ e^{A_s(t-t')} \left[ \nu_0(t') , x \right] , x \right] \half \langle \left[ \xi(t), \xi(t') \right]_+ \rangle \, dt' 
\end{align}
Finally, it remains to recombine all the terms together
{\alld
\begin{align} \label{alnu}
 \nu(t) =& \Tr_B \left\{ \rho (t) \right\} \non 
\\ {\dot \nu}(t) =& \Tr_B \left\{ {\dot \rho}_0(t) \right\} + \ep \, \Tr_B \left\{ {\dot \rho}_1(t) \right\} + \ep ^2 \, \Tr_B \left\{ {\dot \rho}_2 (t) \right\} \non 
\\ =& {\dot \nu}_0 (t) + \ep \, {\dot \nu}_1 (t) + \ep ^2 \, {\dot \nu}_2 (t) \non 
\\ =& - \frac{\imath}{\hb} \left[ H_s , \nu_0(t) \right] - \frac{\imath}{\hb} \, \ep \left[ H_s , \nu_1(t) \right] - \frac{\imath}{\hb} \, \ep ^2 \left[ H_s , \nu_2(t) \right] \non 
\\ & + \frac{\imath\gm}{2\hb} \left[ \left[ {\dot x} , \nu_0(t) \right]_+ , x \right] + \frac{\imath\gm}{2\hb} \, \ep \left[ \left[ {\dot x} , \nu_1(t) \right]_+ , x \right] \non
\\ & + \frac{\imath\gm}{2\hb} \, \ep ^2 \left[ \left[ {\dot x} , \nu_2(t) \right]_+ , x \right] \non 
\\ & - \frac{\ep ^2}{\hb ^2} \int_{-\infty} ^t \left[ e^{A_s(t-t')} \left[ \nu_0(t') , x \right] , x \right] \half \langle \left[ \xi(t), \xi(t') \right]_+ \rangle \, dt' \non 
\\ =& - \frac{\imath}{\hb} \left[ H_s , \nu(t) \right] + \frac{\imath\gm}{2\hb} \left[ \left[ {\dot x} , \nu(t) \right]_+ , x \right] \non 
\\ & - \frac{1}{\hb ^2} \int_{-\infty} ^t \left[ e^{A_s(t-t')} \left[ \nu(t') , x \right] , x \right] \half \langle \left[ \ep \xi(t), \ep \xi(t') \right]_+ \rangle \, dt' \non 
\end{align}}
If we recall our assumption that the bath be Ohmic and have a closer look at the autocorrelation function, we can see \cite{FLO'C:QLE1988} 
\begin{align}
 \frac{1}{2} \langle \left[ \xi(t), \xi(t') \right]_+ \rangle =& \frac{1}{\pi} \int_0^{\infty} G(\og) \hb \og \coth \left( \frac{\hb \og}{2 k T} \right) \cos \left[\og (t-t') \right] \, d\og \non 
\\ =& \frac{\gm}{\pi} \int_{-\infty} ^t \hb \og \coth \left(\frac{\hb \og}{2 k T}\right) \cos \left[\og(t-t')\right] \, d\og \non 
\\ =& kT \gm \frac{d}{dt} \coth \left( \frac{\pi kT }{\hb} (t-t') \right)
\end{align}
which in the classical limit ($\hb \rightarrow 0$) \cite{FO'C:COTH1996, FO'C:COTH2002} becomes
\begin{equation} \label{claslim}
 \frac{1}{2} \langle \left[ \xi(t), \xi(t') \right]_+ \rangle \rightarrow 2 k T \gm \dt(t-t') 
\end{equation}
Finally, inserting (\ref{claslim}) into (\ref{alnu}) yields the Non-Rotating Wave master equation in the Brownian motion limit (removing $\ep$)
\be \label{NRWME1}
 {\dot \nu}(t) = -\frac{\imath}{\hb} \left[H_s , \nu (t)\right] + \frac{\imath \gm}{2\hb} \left[ \left[ {\dot x} , \nu (t) \right]_+ , x \right]  - \frac{k T \gm}{\hb ^2} \left[ \left[ \nu(t) , x \right] , x \right] 
\ee

For a two-particle system, with each particle in its own heat bath, (\ref{NRWME1}) reads analogously
\begin{align}
 {\dot \nu}(t) &=& -\frac{\imath}{\hb} \left[H_s , \nu (t)\right] + \frac{\imath \gm_1}{2\hb} \left[ \left[ {\dot x}_1 , \nu (t) \right]_+ , x_1 \right] + \frac{\imath \gm_2}{2\hb} \left[ \left[ {\dot x}_2 , \nu (t) \right]_+ , x_2 \right] \non \\ && - \frac{k T_1 \gm_1}{\hb ^2} \left[ \left[ \nu(t) , x_1 \right] , x_1 \right]  - \frac{k T_2 \gm_2}{\hb ^2} \left[ \left[ \nu(t) , x_2 \right] , x_2 \right]
\end{align}
where we assume that though the baths follow the same dynamics, they are independent and are not necessarily at the same temperature nor have the same coupling constant $\gm$.

\section{Solution to the  Master Equation}

\sbp Let us recall the N.R.W. master equation for a single particle in the Brownian motion limit ($\nu \rightarrow \rho$), rearranging (\ref{NRWME1})
\begin{align} \label{NRWME2}
 {\dot\rho}(t)= -\frac{\imath}{\hb} \left[H_{sys}, \rho(t)\right] &- \frac{\imath \gm}{2\hb} \left[x , \left[{\dot x}, \rho(t) \right]_+\right] - \frac{\gm kT}{\hb ^2} \left[x,\left[x, \rho(t)\right]\right]
\end{align}
We consider the case where the particle undergoes free motion $H_{sys} = \frac{p ^2}{2m}$ (and the coupling is a position coupling  $X\rightarrow x$). Considering the position matrix $\langle x \vert \rho \vert y \rangle$ and $p \vert x \rangle = - \im \hb \frac{\partial}{\partial x} \vert x \rangle$ (\ref{NRWME2}) reads
{\alld
\begin{align}
\frac{\partial}{\partial t} \langle x \vert \rho \vert y \rangle =& \frac{\imath\hb}{2m} \left(\frac{\partial^{2}}{\partial x ^2} - \frac{\partial ^2}{\partial y ^2}\right) \langle x \vert \rho \vert y \rangle \non 
\\ & -\frac{\gm}{2m}(x - y) \left(\frac{\partial}{\partial x} - \frac{\partial}{\partial y} \right) \langle x \vert \rho \vert y \rangle - \frac{\gm kT}{\hb ^2}(x - y) ^2  \langle x \vert \rho \vert y \rangle
\end{align}}
Applying the change of variables $x = u + \hb z$, $y = u - \hb z$, $\langle x \vert \rho \vert y \rangle \rightarrow P(u,z)$ yields
\be \frac{\partial}{\partial x} = \frac{1}{2} \left(\frac{\partial}{\partial
u} + \frac{1}{\hb} \frac{\partial}{\partial z} \right)  \hspace{1 in}
\frac{\partial}{\partial y} = \frac{1}{2} \left( \frac{\partial}{\partial
u} - \frac{1}{\hb} \frac{\partial}{\partial z} \right) \non\ee
and 
\begin{align}
\frac{\partial}{\partial t} P(u,z,t)= \Biggl[ \frac{\imath}{2m} \left( \frac{\partial ^2}{\partial u\partial z} \right) -\frac{\gm}{m}\, z \frac{\partial}{\partial z} -4\gm kT z ^2 \Biggr]P(u, z, t)
\end{align}
We now apply a Fourier transform with respect to u
\be {\tilde P}(q,z) = \int du e^{-i  qu} {P}(u,z), \ee 
to get
\begin{align}
\frac{\partial}{\partial t} {\tilde P} (q, z, t) =& - \Biggl[ \left( \frac{\gm}{m} z + \frac{q}{2m} \right) \frac{\partial}{\partial z} + 4\gm kT z ^2 \Biggr]{\tilde P}(q, z, t) 
\end{align}
This equation can be solved using the method of characteristics \cite{Farlow:1982}.
The characteristic equation is  
\be \frac{dz}{dt} = \frac{\gm}{m} z + \frac{q}{2m} \non \ee
Writing
\be w = z + \frac{q}{2\gm} \non \ee  
we have 
\be \frac{dw}{dt} = \frac{\gm}{m} w \non\ee
which can be solved as 
\be w = c e^{\gm t/m} \non \ee
On a characteristic
\begin{align}
\frac{d}{dt}  {\tilde P} (q, z(t), t) =& - 4 k T \gm z(t) {\tilde P} = \biggl[ - 4\gm kT \left( c e^{\gm t/m} - \frac{q}{2\gm} \right) ^2  \biggr] {\tilde P} (q, z, t)
\end{align}
where $c = z_0 + \frac{q}{2\gm}$. Then
{\alld
\begin{align}
 \int \frac{d {\tilde P}}{{\tilde P}} =& - 4\gm kT \, \int \left( c ^2\, e^{2\gm t/m} - \frac{c\, q}{\gm} e^{\gm t/m} + \frac{q ^2}{4 \gm ^2} \right) dt \non \\ \non 
\\ \ln {\tilde P} - \ln {\tilde P}_0 =& - 4\gm kT \, \left[ \frac{c ^2\, m}{2\gm} e^{2\gm t/m} - \frac{c \, q\, m}{\gm ^2} + \frac{q ^2 \, t}{4 \gm ^2} \right]_0 ^t \non \\ \non 
\\ {\tilde P} (q, z(t), t) =& {\tilde P} (q, z_0, 0) \exp \left[  - \frac{kT \, t}{\gm} \, q ^2 \right] \non \\* & \quad \times \exp \biggl[ - 2mkT \, c ^2 \, (e^{2\gm t/m} - 1)  + \frac{4mkT}{\gm} \, c \, q \, (e^{\gm t/m} - 1)  \biggr] \non
\end{align}}
To obtain ${\tilde P}(q, z, t)$, we now have to express $c$ in terms of $z$ and $t$ 
with
\be c = \left( z + \frac{q}{2\gm}\right) e^{-\gm t/m} \non \ee
Finally the solution can be written as
{\alld
\begin{align}
{\tilde P}(q,z,t) =& {\tilde P} (q_0, z_0, 0) \exp \Biggl[ - \frac{kT \, t}{\gm} \, q ^2 - 2mkT \, \left( z + \frac{q}{2\gm}\right) ^2 \, (1 - e^{-2\gm t/m}) \non \\* & \hspace{1.1 in} + \frac{4mkT}{\gm} \, \left( z + \frac{q}{2\gm}\right) \, q \, (1 - e^{-\gm t/m}) \Biggr] \non 
\\ =& {\tilde P} (q_0, z_0, 0) \exp \biggl[ -\bt \, \left( z + \frac{q}{2\gm}\right) ^2  + \ap \, q \, \left( z + \frac{q}{2\gm}\right) - \tau\, t \,q ^2 \biggr] \label{solME}
\end{align}}
with
{\alld
\begin{align}
\bt =&  2mkT (1 - e^{-2\gm t/m}) 
\\ \ap =& \frac{4mkT}{\gm} (1 - e^{-\gm t/m}) 
\\ z_0 =& z\, e^{-\gm t/m} - \frac{q}{2\gm} \, \left(1 - e^{-\gm t/m}\right) 
\\ \tau =& \frac{kT}{\gm}
\end{align}}

We have thus found a simple solution to the non-rotating wave master equation, in the case that the system Hamiltonian is chosen to be a free particle one. Then we need only apply a Fourier transform to the initial state and substitute for ${\tilde P}(q_0, z_0, 0)$ in (\ref{solME}). The method used to derive (\ref{solME}) can also be used in the case where a harmonic potential is added to the system Hamiltonian, as will be shown in Chapter 4 in the one-particle setting and in the final chapter for a two-particle case.

\section{Comments on the coupling constant}

\sbp This short section is intended to give a few lines of comments on the coupling constant $\gm$. As shall be seen in the following chapter, the bigger $\gm$ is, the stronger the coupling. One may, for instance in the case of a harmonic oscillator with frequency $\og_0$, notice two different damping behaviours (under and over-damped), depending on how large $\gm$ is with respect to $\og_0$. The reservoir considered is always assumed to be large compared to the system and as such is not affected by the system's influence, i.e. one can assume that the bath remains essentially in thermal equilibrium (since we are using thermal baths) and as such its state is roughly time-independent. This means that at any time we may write $\rho(t) \sim \rho_s (t) \otimes \rho_B$. However, the evolution of the system will in general be strongly influenced by the coupling to the reservoir. Thus strong coupling is to be understood as strongly affecting the system. The noise is still small in the sense that it affects the reservoir little.


		\chapter*{The Non-Rotating Wave Master Equation for a Single Particle}
\stepcounter{chapter}
\addcontentsline{toc}{chapter}{The Non-Rotating Wave Master Equation for a Single Particle}

\sbp This chapter studies the time evolution of a single-particle Gaussian state, comparing the results obtained with the N.R.W. master equation to the results obtained by Savage and Walls in \cite{SavWall:85_1}. It also examines the case of a quantum harmonic oscillator and compares to work done by Savage and Walls in \cite{SavWall:85_2}. This will allow us to determine that the master equation chosen is satisfactory to our purpose.

\section{The case of a free particle coupled to a heat bath}

Consider a one-particle system with a Gaussian state wavefunction. \be
\Psi(x)=\frac{1}{(\pi s ^{2})^{1/4}} e^{-\frac{x^{2}}{2 s ^{2}}}
\ee

The corresponding  density matrix is : 
\be \rho(x,y) = \frac{1}{\sqrt{\pi s^{2}}} e^{-\frac{(x^{2}+y^2)}{2 s ^{2}}} \label{initFP}\ee

We evolve it according to (\ref{NRWME2}), with a free particle Hamiltonian. Rename $ x = u+\hb z$ and $y = u-\hb z$ so that 
\be
\rho(x,y)\rightarrow P(u,z) = \frac{1}{\sqrt{\pi s ^{2}}}
e^{-\frac{u^{2}+\hb^{2}z^{2}}{s ^{2}}} \ee 
Applying the Fourier transform 
\be {\tilde P}(q,z) = \int du e^{-i  qu} {P}(u,z), \ee 
yields
\be {\tilde P}(q_0, z_0) = \exp \left[ - \frac{\hb ^2 z_0 ^2}{s ^2} - \frac{q_0 ^2 s^2}{4} \right] \label{tildep0}\ee

The solution to the master equation is recalled here (see (\ref{solME})
\be{\tilde P} (q, z, t) = {\tilde P} (q, z_0, 0) \exp \biggl[ -\bt \, \left( z + \frac{q}{2\gm}\right) ^2  + \ap \, q \, \left( z + \frac{q}{2\gm}\right) - \tau\, t \,q ^2 \biggr] 
\ee
with
{\alld
\begin{align}
\bt =&  2mkT (1 - e^{-2\gm t/m}) 
\\ \ap =& \frac{4mkT}{\gm} (1 - e^{-\gm t/m}) 
\\ z_0 =& z\, e^{-\gm t/m} - \frac{q}{2\gm} \, \left(1 - e^{-\gm t/m}\right) 
\\ \tau =& \frac{kT}{\gm}
\end{align}}

Substituting into (\ref{solME}), we get
\be {\tilde P} (q, z, t) = e^{- {\tilde \ap}(t) q^2 - {\tilde \bt}(t) z^2 + {\tilde \dt}(t) z q} \ee
with 
{\alld
\begin{align} \label{csttilde}
 {\tilde \ap}(t) =& \frac{s ^{2}}{4} + \frac{kTt}{\gm} - \frac{mkT}{2\gm^{2}}(1-e^{-\gm t/ m})(3-e^{-\gm
t / m}) + \frac{\hbar^2}{4 \gm^2 s ^2}(1-e^{-\gm t/m})^2 
\\ {\tilde \bt}(t) =& 2mkT (1 - e^{-2\gm t/m}) + \frac{\hb ^2}{s ^2} e^{- 2 \gm t /m} 
\\ {\tilde \dt} (t) =& \frac{\hbar^2}{\gm s ^2} e^{-\gm t/m} (1-e^{-\gm t/m})
+ \frac{2 m k T}{\gm} (1-e^{-\gm t/m})^2
\end{align}}

Performing the inverse Fourier transform, we get
\begin{align}
 \lefteqn{\rho(x,y,t) = \frac{1}{2 \pi}\sqrt{\frac{\pi}{{\tilde \ap}(t)}}} \non 
\\ & \times \exp \left[ - \left({\tilde \bt}(t)-\frac{{\tilde \dt}(t)^2}{4{\tilde \ap}(t)}\right) \frac{(x - y)^2}{4} - \frac{(x + y) ^2}{16{\tilde \ap}(t)} - \imath \frac{{\tilde \dt}(t)}{8{\tilde \ap}(t)} (x^2 - y^2)\right]
\end{align}

\sbp Savage and Walls considered the following master equation in the zero frequency limit and at high temperature \cite{SavWall:85_1, SavWall:85_2}
\be \label{SavWall}\partial_t \rho = - \imath\hb ^{-1} \left[\frac{\tbf{P}^2}{2m}, \rho \right] - \imath\gm\hb ^{-1} \left[ \tbf{r}, \tbf{P}\rho + \rho\tbf{P} \right] - 2m \gm k_B T \hb^{-2} \left[\tbf{r} , \left[ \tbf{r}, \rho\right]\right] \ee

Using the characteristic equation approach, they solve the equation for the initial state
\be \Psi(\tbf{r}) = (\pi \sg /2)^{-3/4} e^{- \tbf{r}^2 / \sg} \ee
 to get
\begin{align} \label{init}
\langle \textbf{r} \vert \rho \vert \tbf{r} - \bds{\mu} \rangle = (2\pi c)^{-3/2} \exp \left[- \imath b \bds{\mu} \centerdot (\tbf{r}-\bds{\mu}/2)/c - (\tbf{r} - \bds{\mu} /2)^2 /2c - (a-b ^2/2c) \bds{\mu} ^2 \right]
\end{align}
with 
{\alld
\begin{align} \label{cst1}
a &= \frac{mkT}{2 \hb ^2} \left(1 - e^{-4\gm t} \left[1-\frac{\hb ^2}{mkT \sg}\right]\right) \non 
\\ b &= \frac{kT}{2\gm\hb} \left(1 - e^{-2\gm t} \right) \left(1 - e^{-2\gm t} \left[1-\frac{\hb ^2}{mkT \sg} \right] \right) \non 
\\ c &= \frac{\sg}{4} + \left( \frac{\hb}{2m\gm} \right) ^2 \frac{(1 - e^{-2\gm t})^2}{\sg} + \frac{kT}{m\gm ^2} \left( \gm t - \frac{3}{4} + e^{-2\gm t} - \frac{e^{-4\gm t}}{4} \right)
\end{align}}

To compare (\ref{csttilde}) and (\ref{cst1}), one may first remark the difference in the definition of the momentum term of (\ref{NRWME2}) and (\ref{SavWall}). This leads to $e^{-\gm t/m} \sim e^{- 2 \gm t}$. Then if one notes that $\sg \rightarrow s ^2$, one can see easily that $c \sim {\tilde \ap}$, $a \sim {\tilde \bt}$ and $b \sim {\tilde \dt}$. This allows us to say that both master equations are in exact correspondence.

\section{Study of a one-particle Gaussian state with a harmonic potential}

\sbp Although as we have seen, the master equation seems to be consistent with past results, it is useful to verify it for a particular case. The dynamics of the harmonic oscillator are very well known \cite{Senitzky:1960, Senitzky:1961, Weiss:1985} and thus make it a perfect candidate for study.
\subsection*{Evolving the density matrix}

Let us evolve (\ref{initFP}) according to (\ref{NRWME2}), where the harmonic potential is defined as  $H_{sys}=\frac{p^{2}}{2m}+\og x^2$.
The master equation then becomes 
\be \label{mast} \frac{\partial\rho}{\partial
t}=-\frac{\imath}{\hb}\left[H_{sys},\rho\right]-
\frac{\imath\gm}{2m\hb}\left[x,\left[p,\rho\right]_{+}\right]-\frac{\gm
kT}{\hb^{2}} \left[x,\left[x,\rho\right]\right] \ee 
Using the position space matrix form $\langle x|\rho|y\rangle$ of the
density matrix, (\ref{mast}) reads 
{\alld
\begin{align} \label{mastt} \frac{\ptl}{\ptl t}&\langle
x|\rho|y\rangle \non \\ =&
\Biggl(\frac{\imath\hb}{2m}(\frac{\partial^{2}}{\partial x^{2}}
-\frac{\partial^{2}}{\partial y^{2}})  -\frac{\imath\og}{2\hb}(x^2-y^2) \non 
\\* & -\frac{\gm}{2m}(x-y)(\frac{\partial}{\partial x}
-\frac{\partial}{\partial y})-\frac{\gm
kT}{\hb^{2}}(x-y)^{2}\Biggr) \langle x|\rho|y\rangle \non \\
&& 
\end{align}}
which we now solve. Rename $ x = u+\hb z$ and $y = u-\hb z$ so that (\ref{mastt}) becomes:
\be \label{maste}\frac{\partial P}{\partial
t}=\left(\frac{\imath}{2m}\frac{\partial^{2}}{\partial u \partial
z}-2\imath\og uz-\frac{\gm}{m}z\frac{\partial}{\partial z}-4\gm kTz^{2}\right)P
\ee

Now we write $P$ in terms of its Fourier transform ${\tilde P}$ with respect to u, 
and replace it into (\ref{maste}) to get :
\be \frac{\partial {\tilde P}}{\partial t}+(\frac{q+2\gm z}{2m})
\frac{\partial {\tilde P}}{\partial z} - 2 \og z \frac{\partial {\tilde P}}{\partial q} = -4\gm kT z^{2}{\tilde P}
\ee

The characteristic equations are
\be \partial_t \tbf{v} = \frac{M}{2m} \tbf{v} \ee 
with $\tbf{v}^T = (z , q)$ and \be M = \left( \begin{array}{cc} 2 \gm & 1 \\ - 4 m \og & 0 \end{array} \right) \non \ee

To solve the differential equation, we need the eigenvalues and eigenvectors of M. Those can readily be calculated : 
\be \bds{\lb}^T = \left( \gm + \sqrt{\gm ^2 - 4 \og m}  ,  \gm - \sqrt{\gm ^2 - 4 \og m} \right) = (\lb_+ , \lb_-) \ee
and \be Q = \left( \begin{array}{cc} -\frac{1}{\lb_-} & -\frac{1}{\lb_+} \\ 1 & 1 \end{array} \right) \non \ee
Since $Q^{-1} M Q = D$, we must calculate $Q^{-1}$ to be
\be Q^{-1} = \frac{1}{\lb_- - \lb_+} \left(\begin{array}{cc} \lb_- \lb_+ &  \lb_-  \\ -\lb_- \lb_+ & - \lb_+ \end{array} \right) \ee
The differential equation then becomes
\begin{align}
 2m\, \partial_t \, \tbf{v} \, =& M \, \tbf{v} = Q \, D \, Q^{-1} \, \tbf{v} \non 
\\ \partial_t \, \tbf{w} =& \frac{D}{2m} \, \tbf{w} \non 
\\ \tbf{w} =& Q^{-1} \, \tbf{v}
\end{align}
One can then write
\be \tbf{v}(t) = Q e^{D t / 2m} Q^{-1} \tbf{v}_0 \ee and also
\be \tbf{v}_0  = Q e^{- D t / 2m} Q^{-1} \tbf{v}(t)\ee 
with 
\be e^{\pm D t / 2m} = \left( \begin{array}{cc} e^{\pm \lb_+ t /2m } & 0 \\ 0 & e^{\pm \lb_- t / 2m} \end{array}
\right)\ee
This yields
\begin{align}
\tbf{v}(t) =& \left( \begin{array}{c} - \frac{\lb_+ e^{\lb_+ t / 2m} - \lb_- e^{\lb_- t / 2m}}{\lb_- - \lb_+} z_0 - \frac{e^{\lb_+ t / 2m} - e^{\lb_- t / 2m}}{\lb_- - \lb_+} q_0  
\\ \frac{\lb_- \lb_+}{ \lb_- - \lb_+} (e^{\lb_+ t / 2m} - e^{\lb_- t /2m}) z_0 + \frac{\lb_- e^{\lb_+ t / 2m} - \lb_+ e^{\lb_- t /2m}}{\lb_- - \lb_+} q_0 \end{array}\right)
\end{align}
and similarly
\begin{align}
\tbf{v}_0 =& \left( \begin{array}{c} - \frac{\lb_+ e^{- \lb_+ t / 2m} - \lb_- e^{- \lb_- t / 2m}}{\lb_- - \lb_+} z - \frac{e^{- \lb_+ t / 2m} - e^{- \lb_- t / 2m}}{\lb_- - \lb_+} q  
\\ \frac{\lb_- \lb_+}{ \lb_- - \lb_+} (e^{- \lb_+ t / 2m} - e^{- \lb_- t /2m}) z + \frac{\lb_- e^{- \lb_+ t / 2m} - \lb_+ e^{- \lb_- t /2m}}{\lb_- - \lb_+} q \end{array}\right) \label{1Pv0}
\end{align}

The equation for ${\tilde P}$ becomes
{\alld
\begin{align} 
\lefteqn{\frac{d}{dt} {\tilde P} = - 4 \gm k T z^2 = - \frac{4 \gm k T}{(\lb_- - \lb_+)^2} } \non 
\\* & \hspace{1 in} \times \Biggl\{ \left( \lb_+ e^{\lb_+ t / 2m} - \lb_- e^{\lb_- t / 2m} \right)^2 z_0 ^2 \non 
\\* & \hspace{1.2 in} + \left( e^{\lb_+ t / 2m} - e^{\lb_- t / 2m} \right)^2 q_0 ^2 \non 
\\* & \hspace{1.2 in} + 2 z_0 q_0 \left( \lb_+ e^{\lb_+ t / 2m} - \lb_- e^{\lb_- t / 2m} \right)\left( e^{\lb_+ t / 2m} - e^{\lb_- t / 2m} \right) \Biggr\} {\tilde P} \non 
\\ \lefteqn{ \frac{d {\tilde P}}{{\tilde P}} = - \frac{4 \gm k T}{(\lb_- - \lb_+)^2} } \non 
\\* & \hspace{1 in} \times \Biggl\{ \left( \lb_+ ^2 e^{\lb_+ t / m} + \lb_- ^2 e^{\lb_- t / m} - 2 \lb_+ \lb_- e^{(\lb_+ + \lb_-)t / 2m} \right) z_0 ^2 \non 
\\* & \hspace{1.2 in} + \left( e^{\lb_+ t / m} + e^{\lb_- t / m} - 2 e^{(\lb_+ + \lb_-)t / 2m} \right) q_0 ^2 \non 
\\* & \hspace{1.2 in} + 2 z_0 q_0 \left( \lb_+ e^{\lb_+ t / m} + \lb_- e^{\lb_- t / m} - (\lb_+ + \lb_-) e^{(\lb_+ + \lb_-)t / 2m} \right) \Biggr\} dt \non 
\end{align}}
Then
{\alld
\begin{align}
 \lefteqn{\ln {\tilde P} - \ln {\tilde P}_0 = - \frac{4 \gm m k T}{(\lb_- - \lb_+)^2}} \non 
\\ & \hspace{1 in} \times \Biggl[ z_0 ^2 \left( \lb_+ e^{\lb_+ t / m} + \lb_- e^{\lb_- t / m} - \frac{ 4 \lb_+ \lb_-}{\lb_+ + \lb_-} e^{(\lb_+ + \lb_-)t / 2m}  \right) \non 
\\ & \hspace{1.2 in} + q_0 ^2 \left( \frac{1}{\lb_+} e^{\lb_+ t / m} + \frac{1}{\lb_-} e^{\lb_- t / m} - \frac{4}{\lb_+ + \lb_-} e^{(\lb_+ + \lb_-)t / 2m} \right) \non 
\\ & \hspace{1.2 in} + 2 z_0 q_0 \left( e^{\lb_+ t / m} + e^{\lb_- t / m} - 2 e^{(\lb_+ + \lb_-)t / 2m} \right)\Biggr]_0 ^t \non 
\\ \lefteqn{ {\tilde P} = {\tilde P}_0 \exp \Biggl\{ - \frac{4 \gm m k T}{(\lb_- - \lb_+)^2}}  \non 
\\ & \hspace{1 in} \times \Biggl[ z_0 ^2 \bigl( \lb_+ (e^{\lb_+ t / m} - 1) + \lb_- (e^{\lb_- t / m} - 1) \non 
\\* & \hspace{1.5 in} - \frac{4 \lb_+ \lb_-}{\lb_+ + \lb_-} (e^{(\lb_+ + \lb_-)t / 2m} - 1) \bigr) \non 
\\ & \hspace{1.2 in} + q_0 ^2 \left( \frac{e^{\lb_+ t / m} - 1}{\lb_+} + \frac{ e^{\lb_- t / m} - 1}{\lb_-} - \frac{4 (e^{(\lb_+ + \lb_-)t / 2m} - 1) }{\lb_+ + \lb_-} \right) \non 
\\  & \hspace{1.2 in} + 2 z_0 q_0 \left( e^{\lb_+ t / m} + e^{\lb_- t / m} - 2 e^{(\lb_+ + \lb_-)t / 2m} \right)\Biggr] \Biggr\} \label{p0hp}
\end{align}}

Let us write (\ref{1Pv0}) in the form
\be z_0 = a_1(t) z + b_1(t) q \hspace{1 in} q_0 = a_2(t) z + b_2(t) q \non \ee 
Inserting into (\ref{tildep0}), we have
\begin{align} {\tilde P}(q_0, z_0) = \exp \Biggl[ &- z^2 (\frac{\hb ^2}{s ^2} a_1 ^2(t) + \frac{s ^2}{4} a_2 ^2(t)) - q^2 (\frac{\hb ^2}{s ^2} b_1 ^2(t) + \frac{s ^2}{4} b_2 ^2(t)) \non \\ & - 2 z q (\frac{\hb ^2}{s ^2} a_1(t) b_1(t) + \frac{s ^2}{4} a_2(t) b_2(t)) \Biggr] 
\end{align} 

Thus (\ref{p0hp}) becomes
{\alld
\begin{align}
 {\tilde P} =& \exp \left[ - \mathcal{A}(t) q^2 - \mathcal{B}(t) z^2 - \mathcal{C}(t) zq \right] \non \\ \text{with} \non 
\\ \lefteqn{ \mathcal{B}(t) = a_1 ^2(t) \frac{\hb ^2}{s ^2} + \frac{s ^2}{4} a_2 ^2(t) + \frac{4 m \gm k T}{(\lb_- - \lb_+)^2}} \non
\\ \times \Biggl[ & a_1 ^2(t)  \left( \lb_+ (e^{\lb_+ t / m} - 1) + \lb_- (e^{\lb_- t / m} - 1) - \frac{8 m \og}{\gm} (e^{\gm t / m} - 1) \right) \non 
\\ & + a_2 ^2(t) \left( \frac{e^{\lb_+ t / m} - 1}{\lb_+} + \frac{ e^{\lb_- t / m} - 1}{\lb_-} - \frac{2 (e^{\gm t / m} - 1) }{\gm} \right) \non 
\\ & + 2 a_1(t) a_2(t) \left( e^{\lb_+ t / m} + e^{\lb_- t / m} - e^{\gm t / m} \right) \Biggr]
\\ \lefteqn{ \mathcal{A}(t) = \frac{\hb ^2}{s ^2} b_1 ^2(t) + \frac{s ^2}{4} b_2 ^2(t) + \frac{4 m \gm k T }{(\lb_- - \lb_+)^2}} \non
\\ \times \Biggl[ & b_1 ^2(t) \left( \lb_+ (e^{\lb_+ t / m} - 1) + \lb_- (e^{\lb_- t / m} - 1) - \frac{8 m \og}{\gm} (e^{\gm t / m} - 1) \right) \non 
\\ & + b_2 ^2(t) \left( \frac{e^{\lb_+ t / m} - 1}{\lb_+} + \frac{ e^{\lb_- t / m} - 1}{\lb_-} - \frac{ (e^{\gm t / m} - 1) }{\gm} \right) \non 
\\ & + 2 b_1(t) b_2(t) \left( e^{\lb_+ t / m} + e^{\lb_- t / m} - e^{\gm t / m} \right) \Biggr]
\\ \lefteqn{ \mathcal{C}(t) = 2 \frac{\hb ^2}{s ^2} a_1(t) b_1(t) + \frac{s ^2}{2} a_2(t) b_2(t) + \frac{8 m \gm k T}{(\lb_- - \lb_+)^2}} \non
\\ \times \Biggl[ & a_1(t) b_1(t) \left( \lb_+ (e^{\lb_+ t / m} - 1) + \lb_- (e^{\lb_- t / m} - 1) - \frac{8 m \og}{\gm} (e^{\gm t / m} - 1) \right) \non 
\\ & + a_2(t) b_2(t) \left( \frac{e^{\lb_+ t / m} - 1}{\lb_+} + \frac{ e^{\lb_- t / m} - 1}{\lb_-} - \frac{2 (e^{\gm t / m} - 1) }{\gm} \right) \non 
\\ & + (a_1(t) b_2(t) + b_1(t) a_2(t))  \left( e^{\lb_+ t / m} + e^{\lb_- t / m} - e^{\gm t / m} \right) \Biggr]
\end{align}}
after one has noticed that $\lb_+ \lb_- = 4 \og m$ and $\lb_+ + \lb_- = 2 \gm$.

Recall that $\rho(u,z,t)= P(u,z,t)= \frac{1}{2\pi} \int e^{\im qu}
{\tilde P}(q,z,t) dq$. 
Hence 
{\alld
\begin{align} \rho(u,z,t) &= \frac{1}{2\pi} e^{-\mathcal{B}(t)\, z ^2} \int e^{- \mathcal{A}(t)\, q ^2 - \mathcal{C}(t)\,q\,z + \im \,u\,q}dq, \non
\\&= \frac{1}{2\pi}\sqrt{\frac{\pi}{\mathcal{A}(t)}} e^{-\mathcal{B}(t)\, z ^2} \exp\left[\frac{(\im \,u - \mathcal{C}(t)\, z )^2}{4\mathcal{A}(t)}\right]
\end{align}}

Transforming back to the variables $x,y$, we get $\rho(x,y,t)$:
{\alld
\begin{align} \rho(x,y,t) = \frac{1}{2\pi} \sqrt{\frac{\pi}{\mathcal{A}(t)}} \exp
\biggl[ &-\left(\mathcal{B}(t)-\frac{\mathcal{C}(t)^2}{4\mathcal{A}(t)}\right)\frac{(x-y)^2}{4\hb ^2} \non \\ &- \im\frac{\mathcal{C}(t)}{2\mathcal{A}(t)} \frac{x^2-y^2}{4\hb} - \frac{1}{4\mathcal{A}(t)} \frac{(x+y)^2}{4} \biggr]\non 
\\ = \frac{1}{2\pi} \sqrt{\frac{\pi}{\mathcal{A}(t)}} &e^{-\ep x ^2-{\bar \ep} y ^2 + 2\nu xy}
\end{align}}

with
{\alld
\begin{align}
\ep =& \left(\mathcal{B}(t)-\frac{\mathcal{C}(t)^2}{4\mathcal{A}(t)}\right)\frac{1}{4\hb ^2} + \frac{1}{16 \mathcal{A}(t)} + \im  \frac{\mathcal{C}(t)}{8 \hb \mathcal{A}(t)} = \xi + \im \eta\non
\\ \nu =& \left(\mathcal{B}(t)-\frac{\mathcal{C}(t)^2}{4\mathcal{A}(t)}\right)\frac{1}{4\hb ^2} -\frac{1}{16 \mathcal{A}(t)}
\end{align}}
and ${\bar \ep}$ the complex conjugate of $\ep$.
To check that the density matrix is normalised :
\be \int \rho(x,x,t)dx=\Og\int \exp \left(-(\ep+{\bar \ep}-2\nu) x^2\right) dx=1
\ee
if $\Og=\frac{1}{2 \pi} \sqrt{\frac{\pi}{\mathcal{A}(t)}}$.

\section*{Study of the entropy}

To calculate the entropy, one must get the eigenvalues of the density matrix. The eigenvalue equation is the following : 
\be \int \rho(x,y,t) \Psi(y) dy= {\tilde \lb}_0 (t) \Psi(x)\ee

Let us try $\Psi(y)=e^{-\dt y^{2}}$ as an eigenvector.

We get 
{\alld
\begin{align} \int \rho(x,y,t)\Psi(y) dy =& \,
\Og e^{-\ep x^{2}} \int e^{-{\bar \ep}y^{2} - \dt y^{2} + 2xy\nu} dy \non 
\\ =&\, \Og e^{-\ep x^{2}} \int e^{-(\dt+{\bar\ep})y^{2}} e^{2\nu xy} dy \non
\\ =&\, \Og \sqrt{\frac{\pi}{\dt+{\bar\ep}}} e^{-\ep x^{2}}
e^{\frac{(\nu x)^{2}}{\dt+{\bar\ep}}} 
\end{align}}

We must calculate $\dt$, where  \be \dt = \ep -
\frac{\nu ^{2}}{\dt+{\bar\ep}} \ee 
We have $(\dt+{\bar\ep})(\dt-\ep)= -{\nu ^{2}}$ so $\dt^{2}-\dt(\ep-{\bar\ep})-\vert \ep \vert^{2} =
- \nu ^{2}$, i.e. $\dt^{2}+2 \im \dt \eta -\eta ^2 = \xi ^2- \nu ^2 $ calling $\ep=\xi+\im \eta$.

\be \dt(t) = \sqrt{\xi ^2-\nu ^2}+\im \eta \ee and \be \dt(t)+{\bar\ep} = \sqrt{\xi ^2-\nu ^2}+\xi. \ee 
The corresponding eigenvalue is given by \be {\tilde \lb}_0(t) = \Og
\sqrt{\frac{\pi}{\dt+{\bar\ep}}}. \ee

If one recalls (\ref{vnelb0}) and (\ref{vnelbn}), one can write the other eigenvalues as 
\be {\tilde \lb}_{n}(t) = \Og \sqrt{\frac{\pi}{{\bar\ep}(t)+\dt(t)}}
\left(\frac{\nu(t)}{{\bar\ep}(t)+\dt(t)} \right)^{n} \ee

The sum of the eigenvalues must equal 1. In fact, 
\begin{align}
\sum_{n=0}^{\infty} {\tilde \lb}_{n} =& \Og \sqrt{\frac{\pi}{{\bar \ep}+\dt}} \sum_{n=0}^{\infty} \left(
\frac{\nu}{{\bar\ep}+\dt} \right)^{n} \non \\ =& \Og \sqrt{\frac{\pi}{{\bar \ep}+\dt}} \frac{{\bar\ep}+\dt}{{\bar\ep}+\dt-\nu}, \end{align} and
using $ {\bar\ep}+\dt = \sqrt{\xi ^2-\nu ^2} + \xi = \frac{1}{2}
(\sqrt{\xi+\nu}+\sqrt{\xi-\nu})^2 $ and $ {\bar\ep}+\dt-\nu = \sqrt{\xi-\nu}(\sqrt{\xi+\nu}+\sqrt{\xi-\nu}),$ we get 
\be \sum_{n=0}^{\infty} {\tilde \lb}_{n} = \Og\sqrt{4 \mathcal{A}(t) \pi} = 1 \ee

Using Lemma~\ref{LemmaEntropy}, the entropy is given by 
\be S(t) = - \ln \left( 1 - \frac{\nu}{{\bar \ep} + \dt} \right) - \frac{\frac{\nu}{{\bar \ep}+\dt}}{1-\frac{\nu}{{\bar \ep}+\dt}} \ln \left( \frac{\nu}{{\bar \ep} + \dt} \right)
\ee

Figures~(\ref{1PentUD}, \ref{1PentOD}) show the entropy in the under-damped and the over-damped case respectively, increasing as the equilibrium is destroyed over time. One can observe that the over-damped entropy increases more smoothly than its under-damped case counterpart. One can also see that the entropy rises to larger values for the over-damped case than for the under-damped case. A rather quick analysis of the behaviour of the entropy reveals that at large time, $e^{- \lb_+ t / 2m} = e^{- \lb_- t / 2m} = 0$ so that $a_1 = a_2 = b_1 = b_2 = 0$ leading to $\mathcal{A}(t) = \mathcal{B}(t) = \mathcal{C}(t) = 0$ regardless of damping. Then 
\be 1 - \frac{\nu}{{\bar \ep} + \dt} \sim 2 \non \ee and \be \frac{\nu}{\dt + {\bar \ep}} = \frac{4 \mathcal{A}(t) \mathcal{B}(t) - \mathcal{C}(t) ^2 - \hb ^2}{\sqrt{4 \mathcal{A}(t) \mathcal{B}(t) - \mathcal{C}(t) ^2} + \hb} \sim -1 \non \ee 

If we recall M's eigenvalues
\be \bds{\lb}^T = \left(  \gm + \sqrt{\gm ^2 - 4 \og m}  ,  \gm - \sqrt{\gm ^2 - 4 \og m} \right) \ee
we can readily observe the differences in damping. Indeed the bigger $\gm$ is with respect to $\og$, the more strongly damped the system is. The eigenvalues $\lb_+$ and $\lb_-$ will be real only in the over-damped case $ \gm ^2 > 4 \og m$. Figures~(\ref{1PentUD}, \ref{1PentOD}, \ref{1PcohUD}, \ref{1PcohOD}) illustrate the differences between the types of damping quite clearly. One can recall the results of Savage and Walls in \cite{SavWall:85_2} and study the off-diagonal terms by replacing $y = x - \mu$ and getting the variance of the $e^{ - \mu ^2}$ coefficients. Some algebra yields
\be \Delta \mu ^2 = \langle \mu ^2 \rangle - \langle \mu \rangle ^2 = \frac{2 \hb ^2 \sqrt{\pi}}{\mathcal{B}^{\, 3/2}} \non \ee

Figure~\ref{1PcohUD} shows the time evolution of $\Delta \mu ^2$ for the under-damped case. One can easily notice the decreasing oscillations whereas on Figure~\ref{1PcohOD} representing the over-damped case, one can see that the off-diagonal elements disappear quite quickly. Since the off-diagonal elements of the density matrix represent coherence \cite{Zurek:2003}, one can easily conclude that the environment tends to diagonalise the density matrix, all the faster the stronger the coupling. This agrees with the result of Savage and Walls \cite{SavWall:85_1, SavWall:85_2}. 

\begin{figure}
	\begin{center}
 		\includegraphics[scale=0.5]{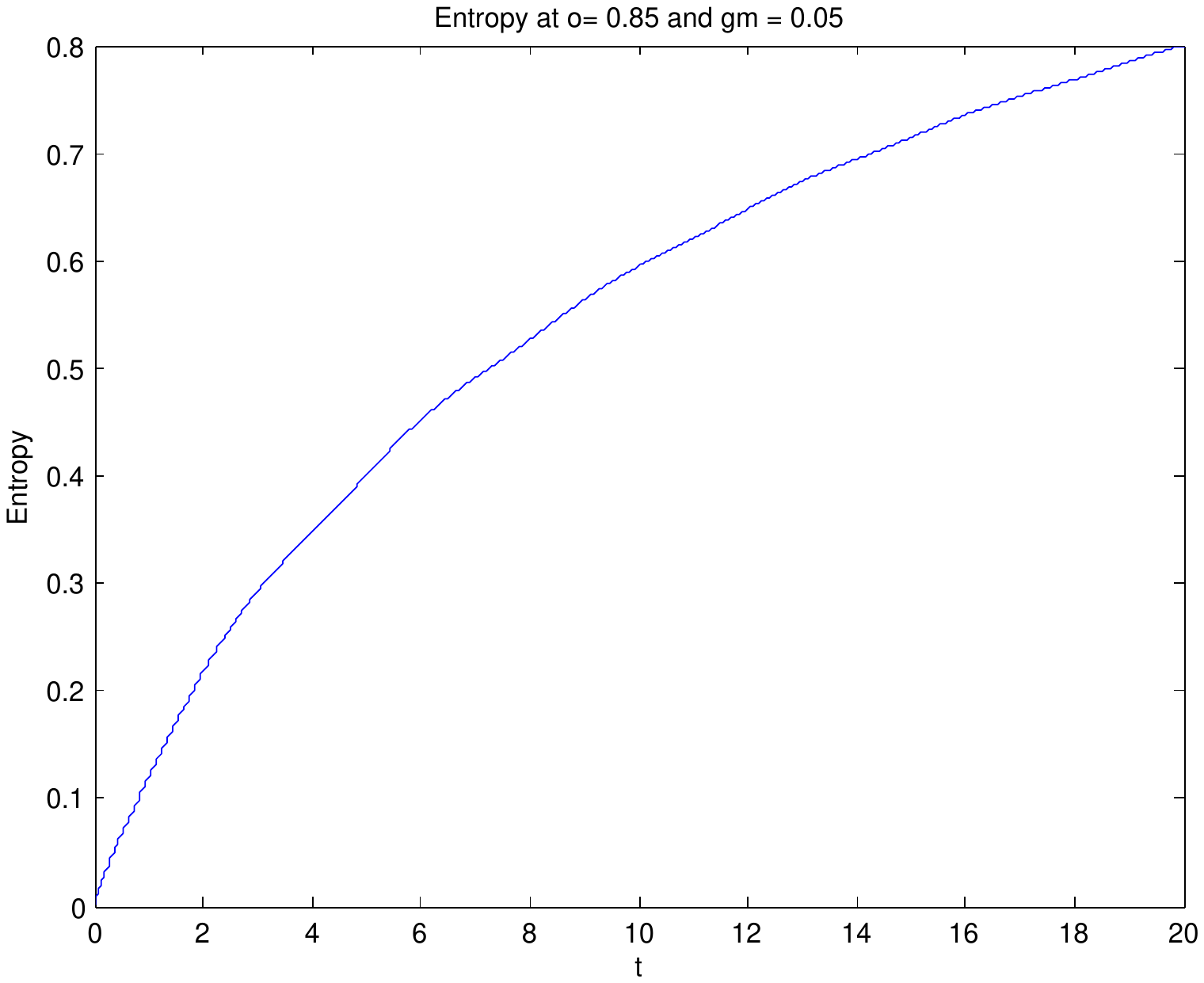}
		\caption{Entropy vs $t$ in the highly under-damped case}
		\medskip{This plot is obtained for $\gm = 0.05$ and $\og = 0.85$}
	\label{1PentUD}
	\end{center}
	\begin{center}
 		\includegraphics[scale=0.5]{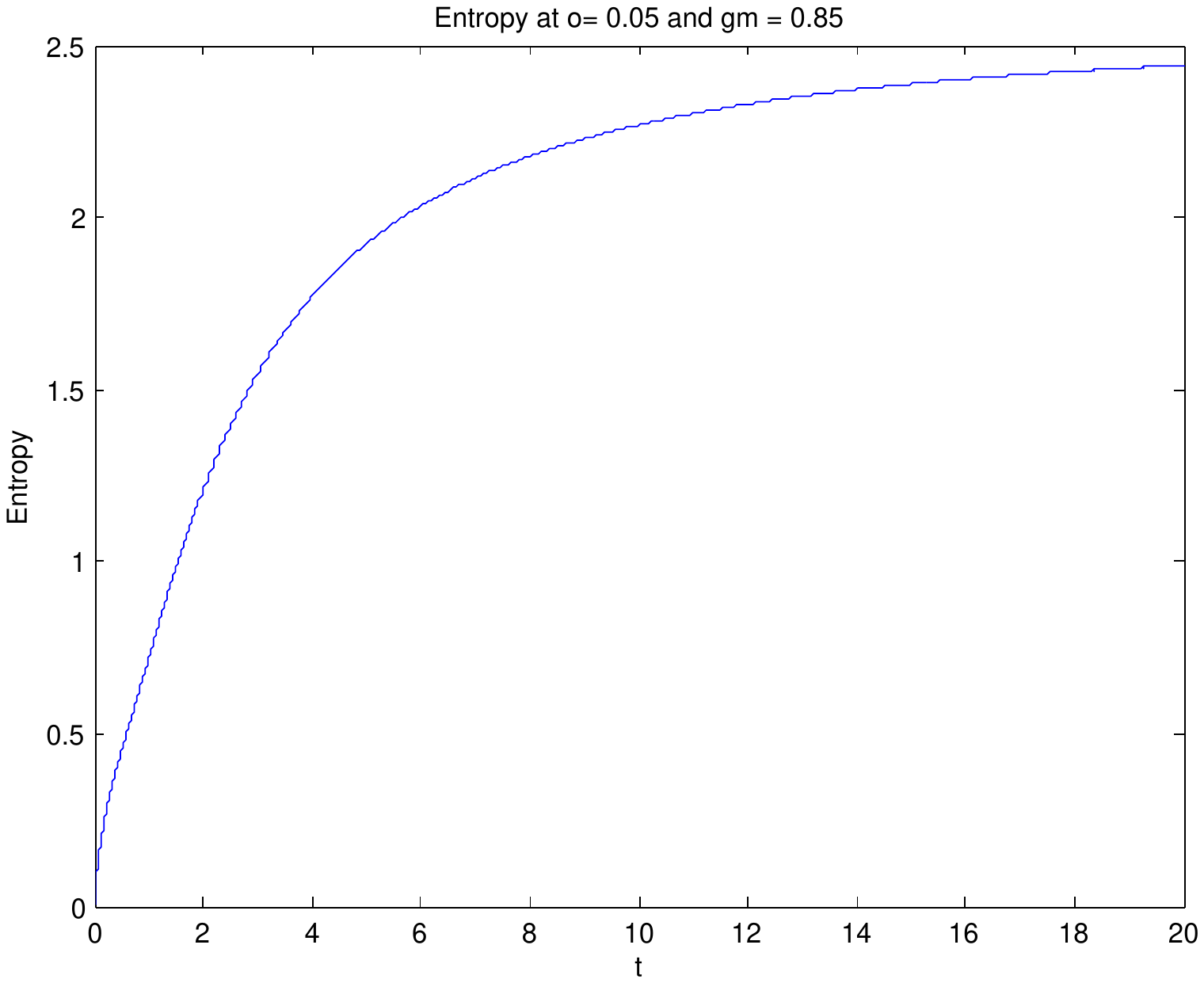}
		\caption{Entropy vs $t$ in the highly over-damped case}
		\medskip{This plot is obtained for $\gm = 0.85$ and $\og = 0.05$}
	\label{1PentOD}
	\end{center}
\end{figure}

\begin{figure}
	\begin{center}
 		\includegraphics[scale=0.5]{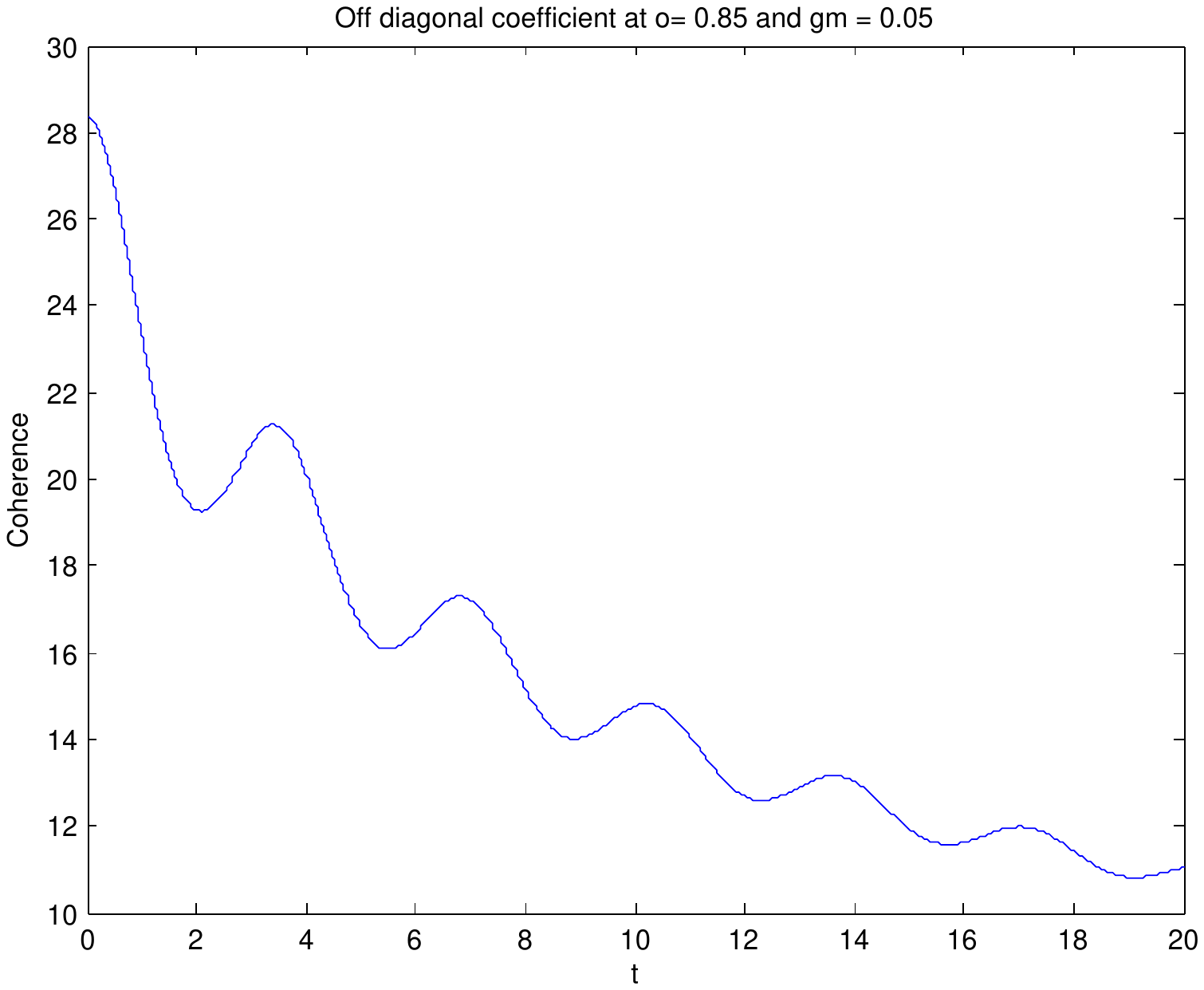}
		\caption{$\Delta \mu ^2$ vs $t$ in the highly under-damped case}
		\medskip{This plot is obtained for $\gm = 0.05$ and $\og = 0.85$}
	\label{1PcohUD}
	\end{center}
	\begin{center}
		\includegraphics[scale=0.5]{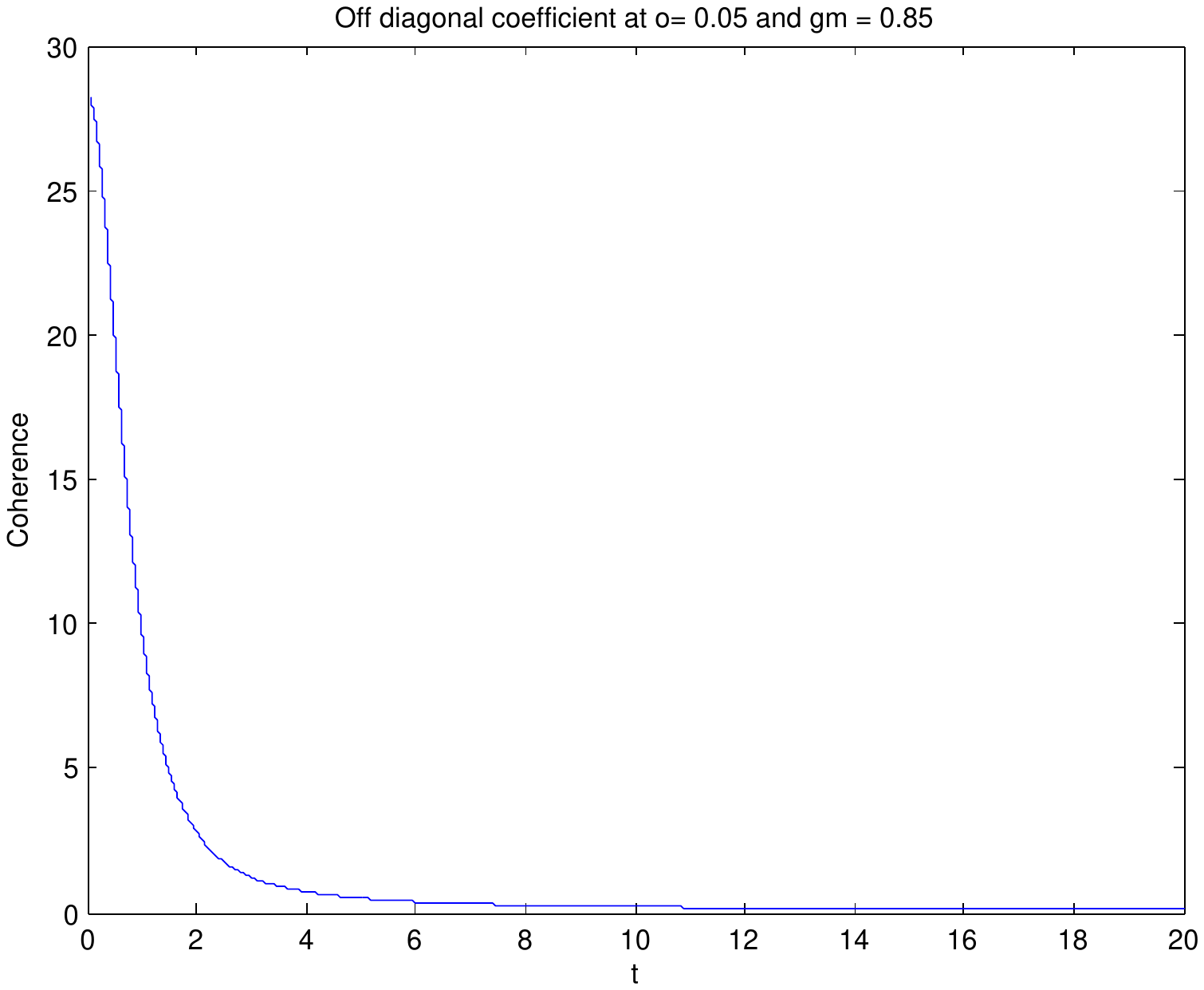}
		\caption{$\Delta \mu ^2$ vs $t$ the highly over-damped case}
		\medskip{This plot is obtained for $\gm = 0.85$ and $\og = 0.05$}
	\label{1PcohOD}	
	\end{center}
\end{figure}


		\chapter*{Free Particle Dynamics}
\stepcounter{chapter}
\addcontentsline{toc}{chapter}{Free Particle Dynamics}

\sbp The second chapter has established both the conservation of entanglement in closed systems dynamics and the formalism used for the main aim of this work. The present chapter contains one of the main results of this thesis and concerns the evolution of the entanglement in a two-particle system when it is left to evolve freely while subjected to an environment.

\section{Free particle Hamiltonian}

\subsection*{The time evolution}

\sbp Let us recall the entangled Gaussian initial state that we studied in Chapter 2

\be \Psi(x_1, x_2) = \frac{1}{\sqrt{2 \pi s d}} e^{- \frac{(x_1 - x_2) ^2}{4 s ^2}} e^{- \frac{(x_1 + x_2) ^2}{16 d ^2}} \ee

with corresponding density matrix
\be \rho(x_1,x_2;x'_1,x'_2;0) = \Og \, e^{-\frac{(x_1-x_2)^2}{4 s^2}-\frac{(x_1+x_2)^2}{16d^2}}e^{-\frac{(x'_1-x'_2)^2}{4 s^2}-\frac{(x'_1+x'_2)^2}{16d^2}}\ee
($\Og = \frac{1}{2 \pi s d}$)

The evolution of the state is modelled via the N.R.W. master equation 
\begin{align}
 {\dot \rho}(t) =& -\frac{\im}{\hb} \left[H_s , \rho \right] + \frac{\im \gm_1}{2\hb} \left[ \left[ {\dot x}_1 , \rho  \right]_+ , x_1 \right] + \frac{\im \gm_2}{2\hb} \left[ \left[ {\dot x}_2 , \rho \right]_+ , x_2 \right] \non \\ & - \frac{k T_1 \gm_1}{\hb ^2} \left[ \left[ \rho , x_1 \right] , x_1 \right]  - \frac{k T_2 \gm_2}{\hb ^2} \left[ \left[ \rho , x_2 \right] , x_2 \right]
\end{align}
whose solution is recalled here :

\begin{align}
 {\tilde P}(\tbf{q},\tbf{z},t)={\tilde P}(\tbf{q}_0,\tbf{z}_0,0)&\times \exp\left[-\bt_1(z_1+\frac{q_1}{2\gm_1})^2+\ap_1 q_1(z_1+\frac{q_1}{2\gm_1})-\tau_1 t {q_1}^2\right]\non\\ &\times \exp\left[-\bt_2(z_2+\frac{q_2}{2\gm_2})^2+\ap_2 q_2(z_2+\frac{q_2}{2\gm_2})-\tau_2 t {q_2}^2\right]
\end{align}

where \bea \tbf{z}_0&=&(z_i+\frac{q_i}{2\gm_i})\,e^{- \gm_i t / m}-\frac{q_i}{2\gm_i}
\\ \bt_i &=&2mkT_i \,(1-e^{- 2\gm_i t / m}) 
\\ \ap_i &=&\frac{4mkT_i}{\gm_i}(1-e^{- \gm_i t / m})  
\\ \tau_i &=&\frac{kT_i}{\gm_i}
\eea

\sbp In order to include the initial state in the solution, we must first perform the following change of variable 

\be x = u+\hb z \hspace{0.5 in} x' = u-\hb z\ee

to get
\begin{align}
 P(\tbf{u},\tbf{z};0)=\Og\, \exp\Big[&-\frac{(u_1+\hb z_1-u_2-\hb z_2)^2}{4\, s ^2}-\frac{(u_1+\hb z_1+u_2+\hb z_2)^2}{16\,d ^2}\non 
\\ & -\frac{(u_1-\hb z_1-u_2+\hb z_2)^2}{4 \, s ^2}-\frac{(u_1-\hb z_1+u_2-\hb z_2)^2}{16\,d ^2}\Big]
\end{align}

We can simplify the above expression as follows :

\bea (u_1+\hb z_1-u_2-\hb z_2)^2+(u_1-\hb z_1-u_2+\hb z_2)^2&=& 2\,(u_1-u_2)^2+2\,(\hb z_1-\hb z_2)^2\non
\\ (u_1+\hb z_1+u_2+\hb z_2)^2+(u_1-\hb z_1+u_2-\hb z_2)^2&=& 2\,(u_1+u_2)^2+2\,(\hb z_1+\hb z_2)^2 \non \\
\eea

to get 

\begin{align}
 P(\tbf{u},\tbf{z};0)=\Og\, \exp\Big[&-(\frac{1}{2\, s ^2}+\frac{1}{8\,d ^2})({u_1}^2+\hb ^2{z_1}^2+{u_2}^2+\hb ^2{z_2}^2) \non \\ &+2\,(\frac{1}{2\, s ^2}-\frac{1}{8\,d ^2})(u_1 u_2 +\hb ^2 z_1 z_2)\Big]
\end{align}

We now apply a Fourier transform to get

\begin{align}
{\tilde P}(\tbf{q}_0,\tbf{z}_0;0)=& \int P(\tbf{u},\tbf{z}_0;0)\, e^{-\imath q_1 u_1-\imath q_2 u_2}du_1\,du_2\non 
\\ =& \Og \exp\left[-\ep_{+}\hb ^2 {z_0}_1^2-\ep_{+}\hb ^2{z_0}_2^2 + 2\ep_{-}\hb ^2 {z_0}_1 {z_0}_2\right]\non
\\ &\times\int \exp\left[-\ep_{+}{u_1}^2-\ep_{+}{u_2}^2+2\ep_{-} u_1 u_2\right]\, e^{-\imath q_1 u_1-\imath q_2 u_2}du_1\,du_2
\end{align}

where $\ep_{+}=\frac{1}{2\, s ^2}+\frac{1}{8\,d ^2}$ and $\ep_{-}=\frac{1}{2\, s ^2}-\frac{1}{8\,d ^2}$. The integral with respect to $u_2$ is performed first, yielding

\begin{align}
{\tilde P}(\tbf{q}_0,\tbf{z}_0;0)=& \Og\,\sqrt{\frac{\pi}{\ep_{+}}}\, \exp\left[-\ep_{+}\hb ^2 {z_0}_1^2-\ep_{+}\hb ^2{z_0}_2^2 + 2\ep_{-}\hb ^2 {z_0}_1 {z_0}_2\right]\non
\\ &\int \exp\left[-\ep_{+} {u_1}^2-\imath q_1 u_1+\frac{(-2\ep_{-}u_1+\imath q_2)^2}{4\ep_{+}}\right] du_1 \non 
\\\non \\ {\tilde P} (\tbf{q}_0,\tbf{z}_0;0) =& \Og \, \sqrt{\frac{\pi}{\ep_{+}}}\, \sqrt{\frac{\pi}{\ep_{+}-\frac{\ep_{-}^2}{\ep_{+}}}} \exp\left[-\ep_{+}\hb ^2 {z_0}_1^2-\ep_{+}\hb ^2{z_0}_2 ^2 + 2\ep_{-}\hb ^2 {z_0}_1 {z_0}_2\right]\non
\\ & \exp\left[-\frac{{q_2}^2}{4\ep_{+}}-\frac{(q_1 + \frac{\ep_{-}q_2}{\ep_{+}})^2}{4(\ep_{+} -\frac{\ep_{-}^2}{\ep_{+}})}\right] \label{tildep02p}
\end{align}

If we now replace $z_0=(z+\frac{q}{2\gm})\,e^{- \gm t / m}-\frac{q}{2\gm}$ and using $\Og = \frac{\sqrt{\ep_+ ^2 - \ep_- ^2}}{\pi}$ we get 

\begin{align}
 \lefteqn{{\tilde P}(\tbf{q}_0,\tbf{z}_0;0)= }\non
\\ & \exp \left[ -\left(\frac{1}{4\ep_{+}} + \frac{\ep_{-}^2}{4\ep_{+}(\ep_{+}^2-\ep_{-}^2)}\right){q_2}^2-\frac{\ep_{+}}{4(\ep_{+}^2-\ep_{-}^2)}{q_1}^2-\frac{\ep_{-}}{2(\ep_{+}^2-\ep_{-}^2)}q_1 q_2\right]\non
\\ &\times \exp\left[-\ep_{+}\hb ^2 \left((z_1+\frac{q_1}{2\gm_1})\,e^{- \gm_1 t / m} - \frac{q_1}{2\gm_1}\right)^2 - \ep_{+}\hb ^2\left((z_2+\frac{q_2}{2\gm_2})\,e^{- \gm_2 t / m} - \frac{q_2}{2\gm_2}\right)^2\right]\non
\\ &\times \exp\left[2\ep_{-}\hb ^2\left((z_1+\frac{q_1}{2\gm_1})\,e^{- \gm_1 t / m}-\frac{q_1}{2\gm_1}\right)\left((z_2+\frac{q_2}{2\gm_2})\,e^{- \gm_2 t / m}-\frac{q_2}{2\gm_2}\right)\right]
\end{align}

Finally

\begin{align}
  {\tilde P}(\tbf{q},\tbf{z},t)&=  \exp\left[-\frac{\ep_{+}}{4(\ep_{+}^2-\ep_{-}^2)}{q_2}^2-\frac{\ep_{+}}{4(\ep_{+}^2-\ep_{-}^2)}{q_1}^2-\frac{\ep_{-}}{2(\ep_{+}^2-\ep_{-}^2)}q_1 q_2\right]\non
\\ &\times \exp\left[-\ep_{+}\hb ^2 \left((z_1+\frac{q_1}{2\gm_1})\,e^{- \gm_1 t / m}-\frac{q_1}{2\gm_1}\right)^2 \right] \non 
\\ &\times \exp \left[-\bt_1(z_1+\frac{q_1}{2\gm_1})^2+\ap_1 q_1(z_1+\frac{q_1}{2\gm_1})-\tau_1 t {q_1}^2\right]\non 
\\ &\times \exp\left[-\ep_{+}\hb ^2\left((z_2+\frac{q_2}{2\gm_2})\,e^{- \gm_2 t /m}-\frac{q_2}{2\gm_2}\right)^2 \right] \non 
\\ &\times \exp \left[-\bt_2(z_2+\frac{q_2}{2\gm_2})^2+\ap_2 q_2(z_2+\frac{q_2}{2\gm_2})-\tau_2 t {q_2}^2\right]\non
\\ &\times \exp\left[2\ep_{-}\hb ^2\left((z_1+\frac{q_1}{2\gm_1})\,e^{- \gm_1 t / m}-\frac{q_1}{2\gm_1}\right)\left((z_2+\frac{q_2}{2\gm_2})\,e^{- \gm_2 t / m}-\frac{q_2}{2\gm_2}\right)\right]\non \\
\end{align}

We can write this in the simpler form

\begin{align} \label{tildePt}
 {\tilde P}(\tbf{q},\tbf{z},t)=& e^{-\mathcal{A}_1 {q_1}^2-\mathcal{A}_2 {q_2}^2-\mathcal{B}_1 {z_1}^2-\mathcal{B}_2 {z_2}^2 - \mathcal{D} z_1 z_2 - \mathcal{E} q_1 q_2}\non
\\ &\times\,e^{- \mathcal{C}_{11} z_1 q_1 - \mathcal{C}_{22} z_2 q_2 - \mathcal{C}_{12} z_1 q_2 - \mathcal{C}_{21} z_2 q_1}
\end{align}

with

{\alld
\begin{align} 
\mathcal{A}_1 =&\frac{\ep_{+}}{4(\ep_{+} ^2-\ep_{-}^2)} + \tau_1 t - \frac{\ap_1}{2\gm_1} + \frac{\bt_1}{4{\gm_1}^2} + \frac{\hb ^2\ep_{+}}{4{\gm_1}^2}\,(1-e^{- \gm_1 t / m})^2\non
\\ \mathcal{A}_2 =&\frac{\ep_{+}}{4(\ep_{+}^2-\ep_{-}^2)} + \tau_2 t - \frac{\ap_2}{2\gm_2} + \frac{\bt_2}{4{\gm_2}^2} + \frac{\hb ^2\ep_{+}}{4{\gm_2}^2} \,(1-e^{- \gm_2 t /m})^2\non
\\ \mathcal{B}_1 =& \hb ^2\ep_{+}\,e^{- 2 \gm_1 t / m} + \bt_1 \hspace{1 in} \mathcal{B}_2 = \hb ^2\ep_{+}\,e^{- 2 \gm_2 t / m}+\bt_2 \non
\\ \mathcal{C}_{11} =& \frac{\bt_1}{\gm_1} - \ap_1 - \frac{\hb ^2\ep_{+}}{\gm_1} (e^{- \gm_1 t / m} - e^{- 2 \gm_1 t / m})\non
\\ \mathcal{C}_{22} =& \frac{\bt_2}{\gm_2} - \ap_2 - \frac{\hb ^2\ep_{+}}{\gm_2} (e^{- \gm_2 t /m} - e^{- 2 \gm_2 t / m})\non
\\ \mathcal{D} =& - 2\hb ^2\ep_{-}\,e^{- \gm_1 t / m}\,e^{- \gm_2 t / m} \non
\\  \mathcal{E} =&\frac{\ep_{-}}{2(\ep_{+} ^2-\ep_{-}^2)} - \frac{\hb ^2\ep_{-}}{2\gm_1 \gm_2}(1 - e^{- \gm_1 t / m})(1 - e^{- \gm_2 t /m})\non
\\  \mathcal{C}_{12} =&\frac{\hb ^2\ep_{-}}{\gm_2}\,e^{- \gm_1 t / m}\,(1 - e^{- \gm_2 t /m}) \hspace{1 in} \mathcal{C}_{21}=\frac{\hb ^2\ep_{-}}{\gm_1}\,e^{- \gm_2 t /m}\,(1 - e^{- \gm_1 t / m})\non \\
\end{align}}

Applying the inverse transform yields

{\alld
\begin{align}
 P(\tbf{u},\tbf{z},t) = \frac{1}{4\pi ^2}&\int {\tilde P}(\tbf{q},\tbf{z},t))e^{\imath q_1 u_1+\imath q_2 u_2} dq_1\,dq_2 \non
\\ =\frac{1}{4\pi ^2}&\,e^{-\mathcal{B}_1 z_1 ^2-\mathcal{B}_2 z_2 ^2 - \mathcal{D} z_1 z_2}\non
\\* &\int  e^{-\mathcal{A}_1 q_1 ^2-\mathcal{A}_2 q_2 ^2 - \mathcal{E} q_1 q_2 - \mathcal{C}_{11} z_1 q_1 - \mathcal{C}_{22} z_2 q_2 - \mathcal{C}_{12} z_1 q_2 - \mathcal{C}_{21} z_2 q_1} \non 
\\*& \qquad \times \exp\left[\imath q_1 u_1+\imath q_2 u_2\right] dq_1\,dq_2 \non 
\\ = \frac{1}{4\pi ^2}& \sqrt{\frac{\pi}{\mathcal{A}_2}} \,e^{-\mathcal{B}_1 z_1 ^2-\mathcal{B}_2 z_2 ^2 - \mathcal{D} z_1 z_2}\non
\\* & \int  \exp\left[ -\mathcal{A}_1 q_1 ^2 - (\mathcal{C}_{11} z_1 + \mathcal{C}_{21} z_2 - \imath u_1)\,q_1 \right] \non 
\\* & \qquad \times \exp \left[ \frac{(\mathcal{C}_{22} z_2 + \mathcal{C}_{12} z_1 + \mathcal{E} q_1 - \imath u_2)^2}{4\,\mathcal{A}_2}\right] dq_1\, \non
\\ =\frac{1}{4\pi ^2}& \sqrt{\frac{\pi}{\mathcal{A}_2}}\,\sqrt{\frac{\pi}{\mathcal{A}_1-\frac{\mathcal{E} ^2}{4\,\mathcal{A}_2}}} \,e^{-\mathcal{B}_1 z_1 ^2-\mathcal{B}_2 z_2 ^2 - \mathcal{D} z_1 z_2}\non
\\ &\times \exp \left[\frac{(\mathcal{C}_{22} z_2 + \mathcal{C}_{12} z_1 - \imath u_2)^2}{4\,\mathcal{A}_2}\right]\non 
\\ &\times \exp\left[\frac{\left(\mathcal{C}_{11} z_1 + \mathcal{C}_{21} z_2 - \imath u_1 - \frac{\mathcal{E}(\mathcal{C}_{22} z_2 + \mathcal{C}_{12} z_1 - \imath u_2)}{2\, \mathcal{A}_2}\right)^2}{4(\mathcal{A}_1-\frac{\mathcal{E} ^2}{4\,\mathcal{A}_2})}\right] \non 
\\ =\frac{1}{4\pi ^2}& \sqrt{\frac{4\pi ^2}{4\,\mathcal{A}_1\,\mathcal{A}_2-\mathcal{E} ^2}} \,e^{-\mathcal{B}_1 z_1 ^2-\mathcal{B}_2 z_2 ^2 - \mathcal{D} z_1 z_2}\non 
\\ &\times \exp\left[\frac{\mathcal{C}_{22} ^2 z_2 ^2 + \mathcal{C}_{12} ^2 z_1 ^2 - u_2 ^2 - 2\,\im \,(\mathcal{C}_{22} z_2 + \mathcal{C}_{12} z_1)u_2 + 2\,\mathcal{C}_{22} \mathcal{C}_{12} z_2 z_1}{4\,\mathcal{A}_2}\right]\non
\\ &\times \exp\Biggl[\frac{1}{4(\mathcal{A}_1-\frac{\mathcal{E} ^2}{4\,\mathcal{A}_2})}\Biggl\{\mathcal{C}_{11} ^2 z_1 ^2 + \mathcal{C}_{21} ^2 z_2 ^2 + 2 \mathcal{C}_{11} \mathcal{C}_{21} z_1 z_2 \non 
\\* &\hspace{0.75 in} - u_1 ^2+\frac{\mathcal{E} ^2}{4\,\mathcal{A}_2 ^2}(\mathcal{C}_{22} z_2 - \mathcal{C}_{12} z_1 + \imath u_2)^2 \non 
\\* &\hspace{0.75 in} - 2\,\imath\, u_1 \left( \mathcal{C}_{11} z_1 + \mathcal{C}_{21} z_2 - \frac{\mathcal{E}(\mathcal{C}_{22} z_2 + \mathcal{C}_{12} z_1 - \imath u_2)}{2\,\mathcal{A}_2} \right) \non
\\* & \hspace{0.75 in} -\,\frac{\mathcal{E}}{\mathcal{A}_2} (\mathcal{C}_{11} z_1 + \mathcal{C}_{21} z_2)(\mathcal{C}_{22} z_2 + \mathcal{C}_{12} z_1 - \imath u_2) \Biggr\}\Biggr]\non
\end{align}}

Finally
\begin{align}
P(\tbf{u},\tbf{z},t) = \frac{1}{4\pi ^2}& \sqrt{\frac{4\pi ^2}{4\,\mathcal{A}_1\,\mathcal{A}_2-\mathcal{E} ^2}} \non
\\ &\times \exp\big(-\og_1\, z_1 ^2-\og_2\, z_2 ^2 - \chi_1\, u_1 ^2 - \chi_2\, u_2 ^2 - \chi_{12}\, u_1 u_2 - \dt\, z_1 z_2\non \\ &\hspace{0.5 in} - \imath\ta_{11}\, u_1 z_1 - \imath\ta_{22}\, u_2 z_2  - \imath\ta_{12}\, u_1 z_2 - \imath\ta_{21}\, u_2 z_1\big)
\end{align}
with
{\alld
\begin{align} 
\chi_1 =&\frac{\mathcal{A}_2}{4\,\mathcal{A}_1\,\mathcal{A}_2-\mathcal{E} ^2} \hspace{0.5 in} \chi_2=\frac{\mathcal{A}_1}{4 \mathcal{A}_1 \mathcal{A}_2 - \mathcal{E} ^2} \hspace{0.5 in} \chi_{12} = \frac{- \mathcal{E}}{4\,\mathcal{A}_1\,\mathcal{A}_2-\mathcal{E} ^2} \non
\\ \og_1 =& \mathcal{B}_1 - \mathcal{C}_{12}^{\,2} \chi_2 - \mathcal{C}_{11}^{\,2} \chi_1 - \mathcal{C}_{12} \mathcal{C}_{11}\chi_{12} \non 
\\  \og_2 =& \mathcal{B}_2 - \mathcal{C}_{22}^{\,2} \chi_2 - \mathcal{C}_{21}^{\,2} \chi_1 - \mathcal{C}_{22} \mathcal{C}_{21}\chi_{12} \non 
\\  \dt =& \mathcal{D} - 2 \mathcal{C}_{22} \mathcal{C}_{12} \chi_2 - 2 \mathcal{C}_{11} \mathcal{C}_{21} \chi_1 - \chi_{12} \left( \mathcal{C}_{11} \mathcal{C}_{22} + \mathcal{C}_{12} \mathcal{C}_{21}\right) \non  
\\  \ta_{11} =& 2 \mathcal{C}_{11} \chi_1 + \mathcal{C}_{12} \chi_{12} \hspace{1 in}  \ta_{22} = 2 \mathcal{C}_{22} \chi_2 + \mathcal{C}_{21} \chi_{12}  \non 
\\ \ta_{12} =& 2 \mathcal{C}_{21} \chi_1 + \mathcal{C}_{22} \chi_{12} \hspace{1 in} \ta_{21} = 2 \mathcal{C}_{12} \chi_2 + \mathcal{C}_{11}  \chi_{12} 
\end{align}}

Going back to the original variables yields

\begin{align}
  \rho(x_1,x_2;x'_1,x'_2;t) = \Og'\, &e^{-\vsg_1 {x_1}^2- \vsg_2 {x_2}^2 - \vsg'_1 {x'_1}^2- \vsg'_2{x'_2}^2 + 2\nu_1 x_1 x'_1 + 2\nu_2 x_2 x'_2} \non \\ & e^{- \eta x_1 x_2 - \eta' x'_1 x'_2 - \zt x_1 x'_2 - \zt' x_2 x'_1}
\end{align}
with 
{\alld
\begin{align} \vsg_j =& \frac{\og_j}{4\hb ^2}+\frac{\chi_j}{4} + \imath\frac{\ta_{jj}}{4\hb} \hspace{.5 in} \vsg'_j= \frac{\og_j}{4\hb ^2}+\frac{\chi_j}{4} - \imath\frac{\ta_{jj}}{4\hb} \hspace{.5 in} 
 \nu_j = \frac{\og_j}{4\hb ^2} - \frac{\chi_j}{4} \non \\ \non 
\\ \eta =& \frac{\chi_{12}}{4} + \frac{\dt}{4\hb ^2} + \imath\frac{\ta_{12}}{4\hb} + \imath\frac{\ta_{21}}{4\hb} \hspace{1 in} \zt = \frac{\chi_{12}}{4} - \frac{\dt}{4\hb ^2} - \imath\frac{\ta_{12}}{4\hb} + \imath\frac{\ta_{21}}{4\hb} \non \\ \non
\\ \eta' =& \frac{\chi_{12}}{4} + \frac{\dt}{4\hb ^2} - \imath\frac{\ta_{12}}{4\hb} - \imath\frac{\ta_{21}}{4\hb} \hspace{1 in} \zt' = \frac{\chi_{12}}{4} - \frac{\dt}{4\hb ^2} + \imath\frac{\ta_{12}}{4\hb} - \imath\frac{\ta_{21}}{4\hb} \non \\
\end{align}}

\subsection*{Calculating the covariance matrix terms}

We now proceed to compute all the elements of the covariance matrix. In all the following, the integral with repect to $x_2$ is performed first. Note that $4 \chi_1 \chi_2 - \chi_{12} ^2 = \frac{1}{4 \mathcal{A}_1 \mathcal{A}_2 - \mathcal{E} ^2}$ and $\Og' = \frac{1}{2 \pi} \frac{1}{\sqrt{4 \mathcal{A}_1 \mathcal{A}_2 - \mathcal{E} ^2}}$.

{\alld 
\begin{align} 
\gins_{11} =& 2 \Rl \hspace{7 pt} \Tr[\rho{\hat X}_1{\hat X}_1] \non
\\ =&  2 \Rl\hspace{7 pt}\left\{ \Og'\,\int x_1 ^2 e^{-(\vsg_1+\vsg'_1-2\nu_1)x_1 ^2- (\vsg_2+\vsg'_2-2\nu_2) x_2 ^2} \, e^{-(\eta+\eta'+\zt+\zt') x_1 x_2} dx_1\,dx_2 \right\}\non 
\\ &=  2 \Rl\hspace{7 pt}\left\{ \Og'\,\int x_1 ^2 e^{-\chi_1 x_1 ^2-\chi_2 x_2 ^2 - \chi_{12} x_1 x_2}dx_1\,dx_2 \right\} \non 
\\ \gins_{11} &= 2 \Rl\hspace{7 pt} \left\{ \Og' \sqrt{\frac{4\pi ^2}{4\,\chi_1\,\chi_2-\chi_{12} ^2}} \, \frac{2 \chi_2}{4 \chi_1 \chi_2 - \chi_{12} ^2} \right\} = 4 \mathcal{A}_1 \non
\\ \non \\
 \gins_{12} =& 2 \Rl\hspace{7 pt} \Tr[\rho{\hat X}_1{\hat P}_1] = 2 \Rl\hspace{7 pt} \Tr[{\hat X}_1{\hat P}_1 \rho] = 2 \Rl\hspace{7 pt} \left\{ - \imath\hb \, \Tr\left[x_1\,\frac{\partial}{\partial x_1} \rho\right] \right\}\non
\\* =& 2 \Rl\hspace{7 pt} \Biggl\{ \imath\hb\Og' 2(\vsg_1-\nu_1) \,\int x_1 ^2 e^{-\chi_1 x_1 ^2-\chi_2 x_2 ^2 -\chi_{12} x_1 x_2} dx_1\,dx_2 \non 
\\* &\hspace{0.35 in} + \imath\hb\Og' (\eta+\zt)\,\int x_1 x_2 e^{-\chi_1 x_1 ^2-\chi_2 x_2 ^2 -\chi_{12} x_1 x_2} dx_1\,dx_2 \Biggr\}\non 
\\ \gins_{12} =& 2 \Rl\hspace{7 pt} \left\{ \imath\hb\Og'\,\sqrt{\frac{4\pi ^2}{4\,\chi_1\,\chi_2-\chi_{12} ^2}}\, \left(\frac{4 (\vsg_1-\nu_1) \chi_2}{4 \chi_1 \chi_2 - \chi_{12} ^2} -  \frac{\chi_{12} (\eta + \zt) }{4 \chi_1 \chi_2 - \chi_{12} ^2} \right)\right\} = - \mathcal{C}_{11}\non 
\\ \non \\ 
\gins_{13} =& 2 \Rl\hspace{7 pt} \Tr[\rho{\hat X}_1{\hat X}_2] \non
\\* =& 2 \Rl\hspace{7 pt} \left\{ \Og'\,\int x_1 x_2\,e^{-\chi_1 x_1 ^2-\chi_2 x_2 ^2 -\chi_{12} x_1 x_2} dx_1 dx_2 \right\}\non
\\ \gins_{13} =& 2 \Rl\hspace{7 pt} \left\{ \Og'\,\sqrt{\frac{4\pi ^2}{4\,\chi_1\,\chi_2-\chi_{12} ^2}}\, \frac{- \chi_{12}}{4 \chi_1 \chi_2 - \chi_{12} ^2} \right\} = 2 \mathcal{E} \non
\\ \non \\
\gins_{14} =& 2 \Rl\hspace{7 pt} \Tr[\rho{\hat X}_1{\hat P}_2] = 2 \Rl\hspace{7 pt} \Tr[{\hat X}_1{\hat P}_2\rho] = 2 \Rl\hspace{7 pt} \left\{ - \imath\hb \, \Tr\left[x_1\,\frac{\partial}{\partial x_2} \rho\right] \right\} \non
\\* =& 2 \Rl\hspace{7 pt} \Biggl\{ \imath\hb\Og' 2 (\vsg_2-\nu_2)\,\int x_1 x_2 \,e^{-\chi_1 x_1 ^2-\chi_2 x_2 ^2 -\chi_{12} x_1 x_2} dx_1\,dx_2 \non 
\\* &\hspace{0.35 in} + \imath\hb\Og' (\eta + \zt')\,\int x_1 ^2 \,e^{-\chi_1 x_1 ^2-\chi_2 x_2 ^2 -\chi_{12} x_1 x_2} dx_1\,dx_2 \Biggr\} \non 
\\ \gins_{14} =& 2 \Rl\hspace{7 pt} \left\{ \imath\hb\Og'\,\sqrt{\frac{4\pi ^2}{4\,\chi_1\,\chi_2-\chi_{12} ^2}}\, \left(- \frac{2 (\vsg_2-\nu_2) \chi_{12}}{4 \chi_1 \chi_2 - \chi_{12} ^2} + \frac{2 (\eta + \zt') \chi_2}{4 \chi_1 \chi_2 - \chi_{12} ^2} \right) \right\} = - \mathcal{C}_{21}\non 
\end{align}}

{\alld 
\begin{align}  
\gins_{21} =& 2 \Rl\hspace{7 pt} \Tr[\rho{\hat P}_1{\hat X}_1] = 2 \Rl\hspace{7 pt} \left\{ \imath\hb \Tr[\frac{\partial\rho}{\partial x'_1} x'_1] \right\} \non
\\* =&2 \Rl\hspace{7 pt} \Biggl\{ - \imath\hb\Og' 2(\vsg'_1-\nu_1) \,\int x_1 ^2 e^{-\chi_1 x_1 ^2-\chi_2 x_2 ^2 -\chi_{12} x_1 x_2} dx_1\,dx_2 \non 
\\* &\hspace{0.35 in} -\imath\hb\Og' (\eta'+\zt')\,\int x_1 x_2 e^{-\chi_1 x_1 ^2-\chi_2 x_2 ^2 -\chi_{12} x_1 x_2} dx_1\,dx_2 \Biggr\}\non 
\\ \gins_{21} =& 2 \Rl\hspace{7 pt} \left\{ \imath\hb\Og' \,\sqrt{\frac{4\pi ^2}{4\,\chi_1\,\chi_2-\chi_{12} ^2}}\, \left(- \frac{4 (\vsg'_1-\nu_1) \chi_2}{4 \chi_1 \chi_2 - \chi_{12} ^2} + \frac{\chi_{12} (\eta'+\zt')}{4 \chi_1 \chi_2 - \chi_{12} ^2} \right) \right\} = - \mathcal{C}_{11}\non 
\\ \non \\
\gins_{22} =& 2\Rl\hspace{7 pt} \Tr[\rho{\hat P}_1{\hat P}_1] = 2 \Rl\hspace{7 pt} \left\{ -\hb ^2 \Tr[\frac{\partial ^2}{\partial{x'_1}^2}\rho] \right\} \non
\\* =& 2 \Rl\hspace{7 pt} \Biggl\{ \hb ^2\Og' 2\vsg'_1\int e^{-\chi_1 x_1 ^2-\chi_2 x_2 ^2 -\chi_{12} x_1 x_2} dx_1\,dx_2 \non 
\\* &\hspace{0.45 in} - \hb ^2\Og'4(\vsg'_1-\nu_1)^2 \int x_1 ^2\, e^{-\chi_1 x_1 ^2-\chi_2 x_2 ^2 -\chi_{12} x_1 x_2} dx_1\,dx_2 \non 
\\* &\hspace{0.45 in} - \hb ^2\Og'4(\vsg'_1-\nu_1)(\eta'+\zt')\int x_1 x_2\, e^{-\chi_1 x_1 ^2-\chi_2 x_2 ^2 -\chi_{12} x_1 x_2} dx_1\,dx_2 \non 
\\* &\hspace{0.45 in} - \hb ^2\Og' (\eta'+\zt')^2 \int x_2 ^2\, e^{-\chi_1 x_1 ^2-\chi_2 x_2 ^2 -\chi_{12} x_1 x_2} dx_1\,dx_2 \Biggr\}\non 
\\ \gins_{22} =& 2 \Rl \hspace{7 pt} \Biggl\{\hb ^2\Og' \sqrt{\frac{4\pi ^2}{4\,\chi_1\,\chi_2-\chi_{12} ^2}}\non \\* & \hspace{0.35 in} \times \Big( 2\vsg'_1 - \frac{8 (\vsg'_1-\nu_1)^2 \chi_2}{4 \chi_1 \chi_2 - \chi_{12} ^2} - \frac{4(\vsg'_1-\nu_1)(\eta'+\zt') \chi_{12}}{4 \chi_1 \chi_2 - \chi_{12} ^2} - \frac{2 (\eta'+\zt')^2 \chi_1}{4 \chi_1 \chi_2 - \chi_{12} ^2} \Big) \Biggr\} = \mathcal{B}_1\non 
\\ \non \\  
\gins_{23} =& 2\Rl\hspace{7 pt}\Tr[\rho{\hat P}_1{\hat X}_2] = 2 \Rl\hspace{7 pt} \left\{ \imath\hb \Tr[\frac{\partial\rho}{\partial x'_1}x'_2] \right\} \non
\\* =& 2 \Rl\hspace{7 pt} \Biggl\{ -\imath\hb\Og'\,2(\vsg'_1-\nu_1)\,\int x_1x_2 e^{-\chi_1 x_1 ^2-\chi_2 x_2 ^2 -\chi_{12} x_1 x_2} dx_1\,dx_2 \non 
\\* &\hspace{0.35 in} -\imath\hb\Og'\, (\eta'+\zt')\,\int x_2 ^2 e^{-\chi_1 x_1 ^2-\chi_2 x_2 ^2 -\chi_{12} x_1 x_2} dx_1\,dx_2 \Biggr\} \non 
\\ \gins_{23} =& 2 \Rl \hspace{7 pt} \Biggl\{ \im \hb \Og' \sqrt{\frac{4\pi ^2}{4\,\chi_1\,\chi_2-\chi_{12} ^2}}\, \left(\frac{2 (\vsg'_1-\nu_1) \chi_{12}}{4 \chi_1 \chi_2 - \chi_{12} ^2} - \frac{2 (\eta'+\zt') \chi_1}{4 \chi_1 \chi_2 - \chi_{12} ^2} \right) \Biggr\}= - \mathcal{C}_{12} \non 
\\ \non \\
\gins_{24} =& 2\Rl\hspace{7 pt}\Tr[\rho{\hat P}_1{\hat P}_2] = 2 \Rl\hspace{7 pt} \left\{ - \hb ^2 \Tr[\frac{\partial}{\partial x'_1}\frac{\partial}{\partial x'_2}\rho] \right\} \non
\\* =& 2 \Rl\hspace{7 pt} \Biggl\{\hb ^2\Og'\eta'\,\int e^{-\chi_1 x_1 ^2-\chi_2 x_2 ^2 -\chi_{12} x_1 x_2} dx_1\,dx_2 \non 
\\* &\hspace{0.35 in} -\hb ^2\Og' 2(\vsg'_2-\nu_2)(\eta'+\zt')\,\int x_2 ^2 e^{-\chi_1 x_1 ^2-\chi_2 x_2 ^2 -\chi_{12} x_1 x_2} dx_1\,dx_2 \non 
\\* &\hspace{0.35 in} -\hb ^2\Og' 2(\vsg'_1-\nu_1)(\eta'+\zt)\,\int x_1 ^2 e^{-\chi_1 x_1 ^2-\chi_2 x_2 ^2 -\chi_{12} x_1 x_2} dx_1\,dx_2 \non 
\\* &\hspace{0.35 in} -\hb ^2\Og' \left(4(\vsg'_1-\nu_1)(\vsg'_2-\nu_2)+(\eta'+\zt')(\eta'+\zt)\right)\non 
\\*&\hspace{0.40 in} \times \int x_1 x_2 e^{-\chi_1 x_1 ^2-\chi_2 x_2 ^2 -\chi_{12} x_1 x_2} dx_1\,dx_2 \Biggr\} \non
\\ \gins_{24} =& 2 \Rl \hspace{7 pt} \Biggl\{ \hb ^2\Og' \sqrt{\frac{4\pi ^2}{4\,\chi_1\,\chi_2-\chi_{12} ^2}} \non 
\\* &\hspace{0.35 in} \times \Biggl( \eta' - \frac{4 (\vsg'_1-\nu_1)(\eta'+\zt) \chi_2}{4 \chi_1 \chi_2 - \chi_{12} ^2}  - \frac{4 (\vsg'_2-\nu_2)(\eta'+\zt') \chi_1}{4 \chi_1 \chi_2 - \chi_{12} ^2} \non 
\\*  &\hspace{0.75 in} + \frac{\chi_{12} \left((\eta'+\zt')(\eta'+\zt)+4(\vsg'_1-\nu_1)(\vsg'_2-\nu_2)\right)}{4\,\chi_1 \, \chi_2 -\chi_{12} ^2}  \Biggr) \Biggr\} = \frac{\mathcal{D}}{2}\non 
\end{align}}

{\alld 
\begin{align}  
\gins_{31} =& 2 \Rl \hspace{7 pt} \Tr[\rho{\hat X}_2{\hat X}_1] = \gins_{13} = 2 \mathcal{E} \non
\\ \non \\
\gins_{32} =& 2\Rl\hspace{7 pt}\Tr[\rho{\hat X}_2{\hat P}_1] = 2\Rl\hspace{7 pt}\Tr[ {\hat X}_2 {\hat P}_1 \rho] = 2 \Rl\hspace{7 pt} \left\{ - \imath\hb \Tr[x_2\frac{\partial\rho}{\partial x_1}] \right\} \non
\\* =& 2 \Rl\hspace{7 pt} \Biggl\{ \imath\hb\Og' 2(\vsg_1 - \nu_1) \, \int x_1 x_2 \, e^{-\chi_1 x_1 ^2 - \chi_2 x_2 ^2 - \chi_{12} x_1 x_2} dx_1 dx_2 \non 
\\* &\hspace{0.35 in} + \imath\hb\Og' (\eta + \zt) \, \int x_2 ^2 e^{-\chi_1 x_1 ^2 - \chi_2 x_2 ^2 - \chi_{12} x_1 x_2} dx_1 dx_2 \Biggr\} \non 
\\ \gins_{32} =& 2 \Rl \hspace{7 pt} \Biggl\{ \imath\hb\Og'\, \sqrt{\frac{4\pi ^2}{4\,\chi_1\,\chi_2-\chi_{12} ^2}} \Bigl(- \frac{2 (\vsg_1 - \nu_1) \chi_{12}}{4 \chi_1 \chi_2 - \chi_{12} ^2} + \frac{2 (\eta + \zt) \chi_1}{4 \chi_1 \chi_2 - \chi_{12} ^2} \Bigr)  \Biggr\} = - \mathcal{C}_{12} \non 
\\ \non \\
\gins_{33} =& 2 \Rl \hspace{7 pt} \Tr[\rho{\hat X}_2{\hat X}_2]\non
\\*  =& 2 \Rl\hspace{7 pt} \left\{ \Og'\,\int x_2 ^2 e^{-\chi_1 x_1 ^2-\chi_2 x_2 ^2 -\chi_{12} x_1 x_2} dx_1\,dx_2 \right\}\non 
\\ \gins_{33} =& 2 \Rl \hspace{7 pt} \left\{ \Og' \,\sqrt{\frac{4\pi ^2}{4\,\chi_1\,\chi_2-\chi_{12} ^2}} \frac{2 \chi_1}{4 \chi_1 \chi_2 - \chi_{12} ^2} \right\} = 4 \mathcal{A}_2\non 
\\ \non \\
\gins_{34} =& 2 \Rl \hspace{7 pt} \Tr[\rho {\hat X}_2 {\hat P}_2] = 2 \Rl \hspace{7 pt} \Tr[{\hat X}_2 {\hat P}_2 \rho] = 2 \Rl\hspace{7 pt} \left\{ - \imath\hb \Tr[x_2 \frac{\partial \, \rho}{\partial x_2}] \right\}\non
\\* =& 2 \Rl\hspace{7 pt} \Biggl\{ \imath\hb\Og' 2 (\vsg_2-\nu_2)\,\int x_2 ^2 \,e^{-\chi_1 x_1 ^2-\chi_2 x_2 ^2 -\chi_{12} x_1 x_2} dx_1\,dx_2 \non 
\\* &\hspace{0.35 in} + \imath\hb\Og'\,(\eta + \zt')\,\int x_1 x_2 \,e^{-\chi_1 x_1 ^2-\chi_2 x_2 ^2 -\chi_{12} x_1 x_2} dx_1\,dx_2 \Biggr\} \non 
\\ \gins_{34} =& 2 \Rl \hspace{7 pt} \Biggr\{ \imath\hb\Og'\, \sqrt{\frac{4\pi ^2}{4\,\chi_1\,\chi_2-\chi_{12} ^2}} \, \left(\frac{4 (\vsg_2 - \nu_2) \chi_1}{4 \chi_1 \chi_2 - \chi_{12} ^2} - \frac{\chi_{12}(\eta + \zt')}{4\,\chi_1\,\chi_2-\chi_{12} ^2} \right) \Biggr\} = - \mathcal{C}_{22}\non 
\end{align}}

{\alld 
\begin{align}  
\gins_{41} =& 2 \Rl \hspace{7 pt}\Tr[\rho{\hat P}_2{\hat X}_1] = 2 \Rl\hspace{7 pt} \left\{ \imath\hb \, \Tr[x'_1\frac{\partial\rho}{\partial x'_2}] \right\}\non
\\* =& 2 \Rl\hspace{7 pt} \Biggl\{ - \imath\hb\Og' 2(\vsg'_2-\nu_2) \, \int x_1 x_2 e^{-\chi x_1 ^2 - \chi_2 x_2 ^2 - \chi_{12} x_1 x_2} dx_1 dx_2 \non 
\\* &\hspace{0.35 in} - \imath\hb\Og' (\eta' + \zt) \, \int x_1 ^2 e^{-\chi x_1 ^2 - \chi_2 x_2 ^2 - \chi_{12} x_1 x_2} dx_1 dx_2 \Biggr\}\non 
\\ \gins_{41} =& 2 \Rl \hspace{7 pt} \left\{ \imath\hb\Og'\,\sqrt{\frac{4\pi ^2}{4\,\chi_1\,\chi_2-\chi_{12} ^2}}\, \left( \frac{2 (\vsg'_2-\nu_2) \chi_{12}}{4 \chi_1 \chi_2 - \chi_{12} ^2} - \frac{2 (\eta' + \zt) \chi_2}{4 \chi_1 \chi_2 - \chi_{12} ^2} \right) \right\} = - \mathcal{C}_{21}\non 
\\ \non \\ 
\gins_{42} =& 2 \Rl \hspace{7 pt} \Tr[\rho{\hat P}_2 {\hat P}_1] = \gins_{24} = \frac{\mathcal{D}}{2} \non 
\\ \non \\
\gins_{43} =& 2 \Rl \hspace{7 pt} \Tr[\rho{\hat P}_2 {\hat X}_2] = 2 \Rl\hspace{7 pt} \left\{ \imath\hb \Tr[x'_2 \frac{\partial\rho}{\partial x'_2}] \right\}\non 
\\* =& 2 \Rl\hspace{7 pt} \Biggl\{ -\imath\hb\Og' 2 (\vsg'_2-\nu_2)\,\int x_2 ^2 \,e^{-\chi_1 x_1 ^2-\chi_2 x_2 ^2 -\chi_{12} x_1 x_2} dx_1\,dx_2 \non 
\\* & \hspace{0.35 in} -\imath\hb\Og'\,(\eta'+\zt)\,\int x_1 x_2 \,e^{-\chi_1 x_1 ^2-\chi_2 x_2 ^2 -\chi_{12} x_1 x_2} dx_1\,dx_2 \Biggr\} \non 
\\ \gins_{43} =& 2 \Rl \hspace{7 pt} \Biggl\{ \imath\hb\Og'\, \sqrt{\frac{4\pi ^2}{4\,\chi_1\,\chi_2-\chi_{12} ^2}} \left(-\frac{4 (\vsg'_2-\nu_2) \chi_1}{4 \chi_1 \chi_2 - \chi_{12} ^2} + \frac{\chi_{12}(\eta'+\zt)}{4\,\chi_1\,\chi_2-\chi_{12} ^2}\right) \Biggr\} = - \mathcal{C}_{22}\non 
\\ \non \\ 
\gins_{44} =& 2\Rl \hspace{7 pt} \Tr[\rho{\hat P}_2{\hat P}_2] = 2 \Rl\hspace{7 pt} \left\{ -\hb ^2 \Tr[\frac{\partial}{\partial x'_2}\frac{\partial\rho}{\partial x'_2}] \right\}\non
\\* =& 2 \Rl\hspace{7 pt} \Biggl\{  \hb ^2\Og' 2\vsg'_2\, \int e^{-\chi_1 x_1 ^2-\chi_2 x_2 ^2 -\chi_{12} x_1 x_2} dx_1\,dx_2 \non 
\\* &\hspace{0.35 in} - \hb ^2\Og' 4(\vsg'_2-\nu_2)^2 \,\int x_2 ^2 \,e^{-\chi_1 x_1 ^2-\chi_2 x_2 ^2 -\chi_{12} x_1 x_2} dx_1\,dx_2 \non 
\\* &\hspace{0.35 in} - \hb ^2\Og'4 (\vsg'_2-\nu_2)(\eta'+\zt)\,\int x_1 x_2 \,e^{-\chi_1 x_1 ^2-\chi_2 x_2 ^2 -\chi_{12} x_1 x_2} dx_1\,dx_2 \non 
\\* &\hspace{0.35 in} - \hb ^2\Og'(\eta'+\zt)^2 \,\int x_1 ^2 \,e^{-\chi_1 x_1 ^2-\chi_2 x_2 ^2 -\chi_{12} x_1 x_2} dx_1\,dx_2 \non 
\\ \gins_{44} =& 2\Rl \hspace{7 pt} \Biggl\{ \hb ^2\Og'\sqrt{\frac{4\pi ^2}{4\,\chi_1\,\chi_2-\chi_{12} ^2}} \non 
\\* & \hspace{0.5 in} \times \Bigl( 2\vsg'_2  - \frac{8 (\vsg'_2-\nu_2)^2 \chi_1}{4 \chi_1 \chi_2 - \chi_{12} ^2} - \frac{4 (\vsg'_2-\nu_2)(\eta'+\zt) \chi_{12}}{4 \chi_1 \chi_2 - \chi_{12} ^2} - \frac{2 (\eta'+\zt)^2 \chi_2}{4 \chi_1 \chi_2 - \chi_{12} ^2}\Bigr)\Biggr\} = \mathcal{B}_2\non 
\end{align}} 

\subsection*{Alternative derivation of the covariance matrix terms}

\sbp The covariance matrix terms can also be calculated using the change of variables $x= u + \hb z$ and $x' = u - \hb z$, and (\ref{tildePt}) as follows :
{\alld
\begin{align}
 \langle X_i X_j \rangle =& \Tr \left[ \rho X_i X_j \right] \non
\\ =& \int x_i \, x_j \rho(x_i, x_j) dx_i \, dx_j \non 
\\ =& \int u_i \, u_j P(\tbf{u}, \tbf{z}=0, t) du_i \, du_j \non 
\\ =& - \left( \frac{\partial}{\partial q_i} \frac{\partial}{\partial q_j} {\tilde P}(\tbf{q}, \tbf{z}, t)\right) \vert_{z=0, q=0}
\end{align}}
if we notice that
\be
\frac{\partial}{\partial q} {\tilde P} = - \im \, \int u P e^{-\im q u} du
\ee

Similarly, we can get
{\alld 
\begin{align}
 \langle X_i P_j \rangle =& - \im \hb \int x_i \frac{\partial}{\partial x_j} \rho(x_i, x_j) dx_i \, dx_j \non 
\\ =& - \half \im \hb \int u_i \left( \frac{\partial}{\partial u_j} + \frac{1}{\hb} \frac{\partial}{\partial z_j} \right)\, P(\tbf{u}, \tbf{z}, t) du_i \, du_j \non
\\ =& - \half \im \hb \int u_i \frac{\partial}{\partial u_j}\, P(\tbf{u}, \tbf{z}, t) du_i \, du_j - \frac{\im}{2} \int u_i \frac{\partial}{\partial z_j}\, P(\tbf{u}, \tbf{z}=0, t) du_i \, du_j \non 
\\ =& - \half \im \hb \int u_i P(\tbf{u}, \tbf{z}, t) du_i \, du_j + \half \im \hb \int \frac{\ptl u_i}{\ptl u_j} P(\tbf{u}, \tbf{z}, t) \non 
\\ & \qquad - \frac{\im}{2} \int u_i \frac{\partial}{\partial z_j}\, P(\tbf{u}, \tbf{z}=0, t) du_i \, du_j \non
\\ =& \half \im \hb \dt_{ij} \int P(\tbf{u}, \tbf{z}, t) du_i \, du_j - \frac{\im}{2}\, \int u_i \frac{\partial}{\partial z_j}\, P(\tbf{u}, \tbf{z}, t) du_i \, du_j \non 
\\ =& \half \im \hb \dt_{ij} \int P(\tbf{u}, \tbf{z}, t) du_i \, du_j + \half \left( \frac{\partial}{\partial q_i} \frac{\partial}{\partial z_j}\, {\tilde P}(\tbf{q}, \tbf{z}=0, t)\right) \vert_{z=0, q=0} 
\\ \langle P_i X_j \rangle =& \im \hb \int x_j \frac{\partial}{\partial x'_i} \rho(x_i, x_j) dx_i \, dx_j \non 
\\ =& - \half \im \hb \dt_{ij} \int P(\tbf{u}, \tbf{z}=0, t) du_i \, du_j + \half \left( \frac{\partial}{\partial z_i} \frac{\partial}{\partial q_j}\, {\tilde P}(\tbf{q}, \tbf{z}=0, t) \right) \vert_{z=0, q=0}
\end{align}}
and
{\alld
\begin{align}
 \langle P_i P_j \rangle =& - \hb ^2 \int \frac{\partial}{\partial x_i} \frac{\partial}{\partial x_j} \rho(x_i, x_j) dx_i \, dx_j \non 
\\ =& - \frac{\hb ^2}{4} \int \left( \frac{\partial}{\partial u_i} + \frac{1}{\hb} \frac{\partial}{\partial z_i} \right) \, \left( \frac{\partial}{\partial u_j} + \frac{1}{\hb} \frac{\partial}{\partial z_j} \right)\, P(\tbf{u}, \tbf{z}, t) du_i \, du_j \non 
\\ =& - \frac{1}{4} \left( \frac{\partial}{\partial z_j}\, \frac{\partial}{\partial z_j} \, {\tilde P}(\tbf{q}, \tbf{z}, t)\right) \vert_{z=0, q=0}
\end{align}}

Finally,
\begin{align} \label{expxx}
 2 \Rl \langle X_i X_j \rangle =& - 2 \left( \frac{\partial}{\partial q_i} \frac{\partial}{\partial q_j} {\tilde P}(\tbf{q}, \tbf{z}, t)\right) \vert_{z=0, q=0} 
\\ \label{expxp} 2 \Rl \langle X_i P_j \rangle =& \left( \frac{\partial}{\partial q_i} \frac{\partial}{\partial z_j}\, {\tilde P}(\tbf{q}, \tbf{z}, t)\right) \vert_{z=0, q=0}  
\\ \label{exppx} 2 \Rl \langle P_i X_j \rangle =& \left( \frac{\partial}{\partial z_i} \frac{\partial}{\partial q_j}\, {\tilde P}(\tbf{q}, \tbf{z}, t) \right) \vert_{z=0, q=0}
\\ \label{exppp} 2 \Rl \langle P_i P_j \rangle =& - \frac{1}{2} \left( \frac{\partial}{\partial z_j}\, \frac{\partial}{\partial z_j} \, {\tilde P}(\tbf{q}, \tbf{z}, t)\right) \vert_{z=0, q=0}
\end{align}

This method allows us to get the same results as the explicit calculations of the covariance matrix but in a much easier fashion. 

\section{The Logarithmic Negativity}

Using these, we can write $\gins$ in matrix form :
\begin{align} \gins = \left[
\begin{array}{cccc} 4 \mathcal{A}_1 & - \mathcal{C}_{11} & 2 \mathcal{E} & - \mathcal{C}_{21} \\ - \mathcal{C}_{11} & \mathcal{B}_1 & - \mathcal{C}_{12} & \mathcal{D}/2 \\ 2 \mathcal{E} & - \mathcal{C}_{12} & 4 \mathcal{A}_2 & - \mathcal{C}_{22} \\ - \mathcal{C}_{21} & \mathcal{D}/2 & - \mathcal{C}_{22} & \mathcal{B}_2 
\end{array}
\right] \end{align}
To obtain the partial transpose, we set ${\hat p}_1 \rightarrow -{\hat p}_1$ so that $\gins$ becomes
\be \gins ^T = \left[
\begin{array}{cccc} 4 \mathcal{A}_1 & \mathcal{C}_{11} & 2 \mathcal{E} & - \mathcal{C}_{21} \\ \mathcal{C}_{11} & \mathcal{B}_1 & \mathcal{C}_{12} & -\mathcal{D}/2 \\ 2 \mathcal{E} & \mathcal{C}_{12} & 4 \mathcal{A}_2 & - \mathcal{C}_{22} \\ - \mathcal{C}_{21} & -\mathcal{D}/2 & - \mathcal{C}_{22} & \mathcal{B}_2 
\end{array}
\right] \ee
In order to use the logarithmic negativity as defined in previous chapters, let us calculate $\sg\gins ^T$ to get
\begin{align}\sg\gins ^T =
\left[
\begin{array}{cccc}
 \mathcal{C}_{11} & \mathcal{B}_1 & \mathcal{C}_{12} & - \mathcal{D}/2 \\ -4 \mathcal{A}_1 & - \mathcal{C}_{11} & -2 \mathcal{E} & \mathcal{C}_{21} \\ - \mathcal{C}_{21} & - \mathcal{D}/2 & - \mathcal{C}_{22} & \mathcal{B}_2 \\ -2 \mathcal{E} & - \mathcal{C}_{12} & -4 \mathcal{A}_2 & \mathcal{C}_{22}
\end{array}
\right] 
\end{align}

 Then we calculate $- \sg \gins ^T \sg \gins ^T$ to get
\begin{align} - \sg \gins ^T \sg \gins ^T = 
\left[
\begin{array}{cccc}
\mins_{11} & \mins_{12} & \mins_{13} & \mins_{14} 
\\ \mins_{21} & \mins_{22} & \mins_{23} & \mins_{24}
\\ \mins_{31} & \mins_{32} & \mins_{33} & \mins_{34}
\\ \mins_{41} & \mins_{42} & \mins_{43} & \mins_{44}
\end{array}
\right] 
\end{align}
where $\mins_{12} = \mins_{21} = \mins_{34} = \mins_{43} = 0$ and
\begin{align}
 \mins_{11} = \mins_{22} =& 4 \mathcal{A}_1 \mathcal{B}_1 - \mathcal{D} \mathcal{E} + \mathcal{C}_{12} \mathcal{C}_{21} -\mathcal{C}_{11} ^2 \non 
\\ \mins_{33} = \mins_{44} =& 4 \mathcal{A}_2 \mathcal{B}_2 - \mathcal{C}_{22} ^2 -\mathcal{D} \mathcal{E} + \mathcal{C}_{12} \mathcal{C}_{21} \non 
\\ \mins_{13} = \mins_{42} =& 2 \mathcal{E} \mathcal{B}_1 - 2 \mathcal{A}_2 \mathcal{D} - \mathcal{C}_{11} \mathcal{C}_{12} + \mathcal{C}_{12} \mathcal{C}_{22} \non 
\\ \mins_{14} = - \mins_{32}=& - \mathcal{C}_{12} \mathcal{B}_2 - \mathcal{C}_{21} \mathcal{B}_1 + \mathcal{C}_{11} \mathcal{D}/2 + \mathcal{C}_{22} \mathcal{D}/2 \non 
\\ \mins_{23} = - \mins_{41} =& - 2 \mathcal{E} \mathcal{C}_{11} + 4 \mathcal{A}_1 \mathcal{C}_{12} + 4 \mathcal{A}_2 \mathcal{C}_{21} - 2 \mathcal{E} \mathcal{C}_{22} \non 
\\ \mins_{24} = \mins_{31} =& 2 \mathcal{E} \mathcal{B}_2 - \mathcal{C}_{22} \mathcal{C}_{21} + \mathcal{C}_{11} \mathcal{C}_{21} - 2 \mathcal{A}_1 \mathcal{D}  
\end{align}

The eigenvalues of $-\sg\gins ^T \sg\gins ^T$ can then be determined to be :
\begin{align}
\lb_{1,2} ^T =& \frac{\mins_{11} + \,\mins_{33}}{2} + \half \, \sqrt{ (\mins_{11} - \, \mins_{33}) ^2 + 4 \mins_{13} \mins_{24} - 4 \mins_{14} \mins_{23}}\non 
\\ \lb_{3,4} ^T =& \frac{\mins_{11} + \,\mins_{33}}{2} - \half \, \sqrt{ (\mins_{11} - \, \mins_{33}) ^2 + 4 \mins_{13} \mins_{24} - 4 \mins_{14} \mins_{23}}
\end{align}
The logarithmic negativity then becomes
\be \mathcal{L_{\mathcal{N}}}(\rho) = -2\left( \log_2\left(\min(1,\vert \lb_{1,2} ^T \vert)\right)+\log_2\left(\min(1,\vert \lb_{3,4} ^T \vert)\right) \right) \ee 

Figure~\ref{diff_sg1} shows the time evolution of the logarithmic negativity for three values of s. The sharp decrease to zero is clearly visible, showing, as one would expect, the disentanglement between the particles as they are placed in their respective baths. Figure~\ref{diff_sg1} also shows that for a constant $d$, the greater the $s$, the faster the loss of entanglement. 

One may interpret our initial state as a correlated pair of wavepackets with width $d$ and the distance between them $s$. The entanglement then can be understood as the interference between the packets. At a distance $s = 2d$, the interferences are destructive and thus there is no entanglement. At any other distance, the interference pattern is more or less well-defined which is described by a certain value of entanglement. This picture is rather crude, yet serves quite well to illustrate why the entanglement would be lost more quickly when $s$ increases. 
If one recalls Figure~\ref{VNE0} where the entanglement entropy is plotted as a function of $s$, one may recall that as $s=2d$, the entanglement disappears, but that it is present for smaller and for larger $s$, though in smaller amount as $s$ increases. On Figure~\ref{LN_vs_sg}, one can also observe that as the time increases, the range of $s$ around $2d$ at which the entanglement vanishes increases. This suggests that as time increases, the wave-packets would already spread so that the distance around $2d$ at which the interferences become destructive is "blurred". It also suggests that entanglement will be present at large $s$, though in lesser amount still as time increases.

Dodd and Halliwell \cite{DHall:2004} studied disentanglement arising in a separated system and an EPR pair in a similar settings but with negligible dissipation and using merely a separability criterion, whereas our study concerns entanglement itself.

\begin{figure}
\begin{center}
	\includegraphics[scale=0.6]{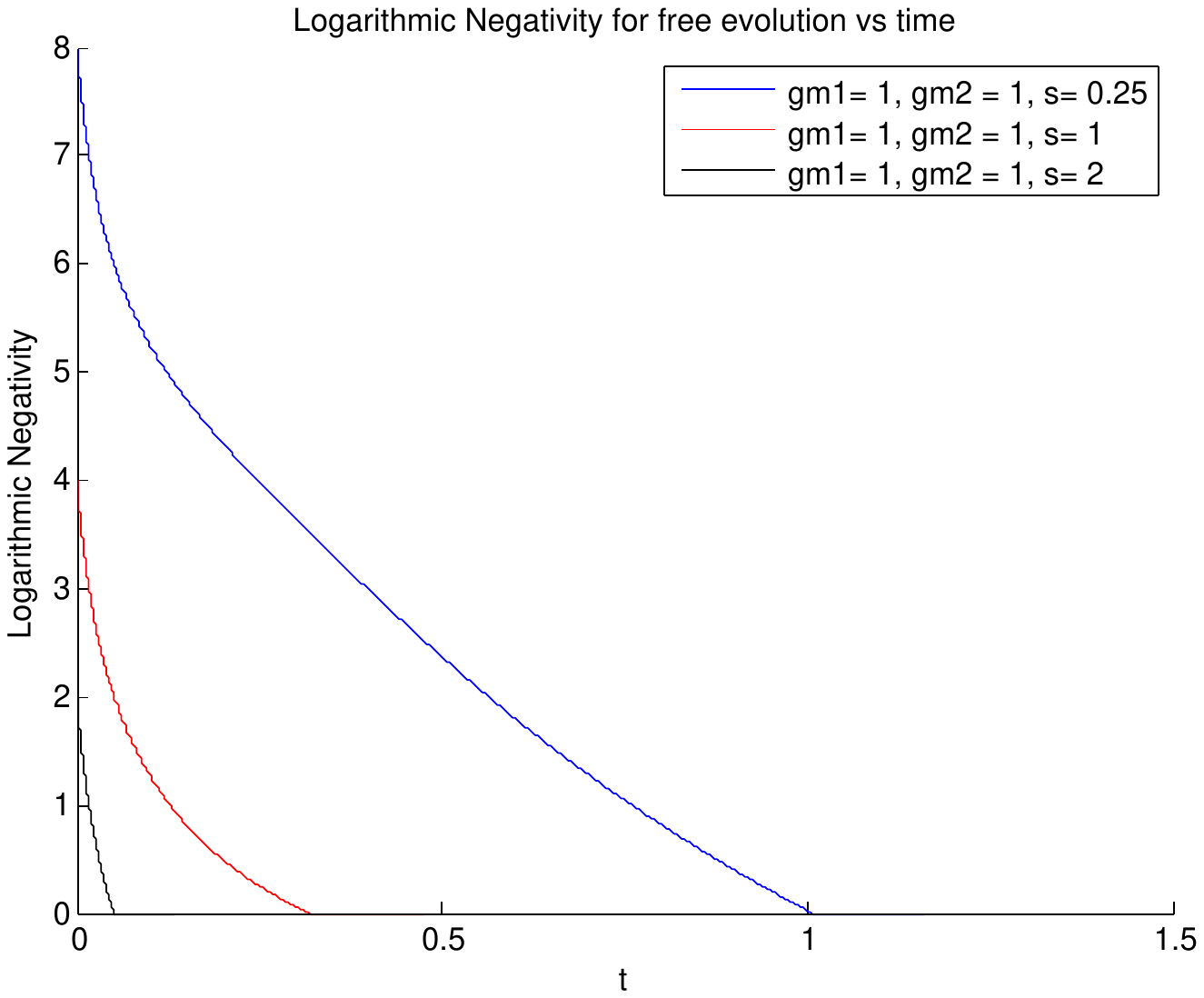}
	\caption{Logarithmic negativity vs t for three values of $s$ }
	\medskip{The values of $s$ are : blue : $s = 0.25$, red : $s = 1$, black : $s = 2$}
\label{diff_sg1}
\end{center}
\end{figure}

\begin{figure}
\begin{center}
	\includegraphics[scale=0.6]{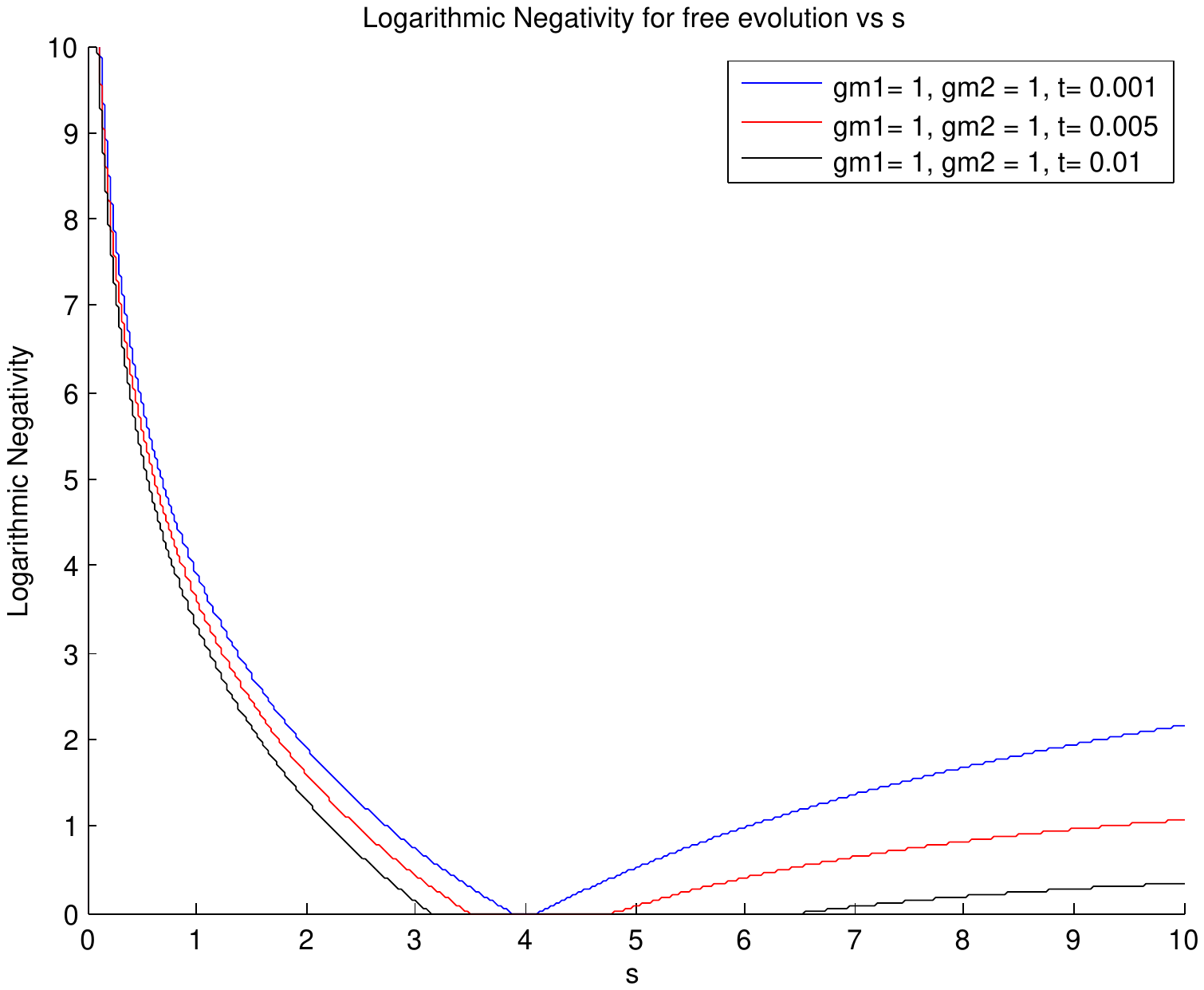} 
	\caption{Logarithmic negativity vs $s$ for three times}
	\medskip{The plot are for : blue : $t = 0.001$, red : $t = 0.005$, black : $t = 0.01$}
\label{LN_vs_sg}
\end{center}
\end{figure}


\chapter*{Harmonic Potential Between Two Particles in a Heat Bath}
\stepcounter{chapter}
\addcontentsline{toc}{chapter}{Harmonic Potential Between Two Particles in a Heat Bath}

\sbp The previous chapter has allowed us to observe that the entanglement between the two particles disappears quickly when the particles are free except for their interaction with the heat bath. In the present chapter we will now look at the influence of a harmonic potential between the particles. We show that allowing the particles to interact may delay the vanishing of the entanglement. We also find a striking difference in behaviour between the over and the under-damped cases.

\section{Time Evolution with a Harmonic Potential}

\sbp We would like to see how allowing the two particles to interact may influence the evolution of the entanglement and maybe slow down its decrease. For that, we add a harmonic interaction between the two-particles. A bipartite harmonic potential Hamiltonian can be written as
\be H_s = \frac{p_1  ^2}{2m} + \frac{p_2 ^2}{2m} + \frac{m \og_0 ^2}{2} (x_1 - x_2) ^2 \ee
where the $\frac{m \og_0 ^2}{2} (x_1 - x_2) ^2$ represents the interaction.

With position coupling, the master equation becomes
\begin{align}
 {\dot \rho} = - \frac{\im}{\hb} \left[ H_s , \rho \right] -& \frac{\im \gm_1}{2\hb} \left[ x_1, \left[ {\dot x}_1 , \rho \right]_+ \right] - \frac{\gm_1 k T_1}{\hb ^2} \left[ x_1, \left[ x_1, \rho \right] \right] \non 
\\ -& \frac{\im \gm_2}{2\hb} \left[ x_2, \left[ {\dot x}_2 , \rho \right]_+ \right] - \frac{\gm_2 k T_2}{\hb ^2} \left[ x_2, \left[ x_2, \rho \right] \right]
\end{align}
Writing the density matrix in position representation, $\rho(x_1, x_2 ; y_1, y_2)$, we get 
\begin{align}
 \frac{\partial \rho}{\partial t} =& \frac{\im \hb}{2 m} \left( \frac{\partial ^2}{\partial x_1 ^2} - \frac{\partial ^2}{\partial y_1 ^2} + \frac{\partial ^2}{\partial x_2 ^2} -\frac{\partial ^2}{\partial y_2 ^2} \right) \rho \non 
\\ &- \frac{\im m \og_0 ^2}{2\hb} \left( (x_1 - x_2)^2 - (y_1 - y_2)^2 \right) \rho \non 
\\ &- \frac{\gm_1}{2 m} \left( x_1 - y_1 \right) \left( \frac{\partial }{\partial x_1} - \frac{\partial}{\partial y_1} \right) \rho - \frac{\gm_1 k T_1}{\hb ^2} (x_1 - y_1)^2 \rho \non 
\\ &- \frac{\gm_2}{2 m} \left( x_2 - y_2 \right) \left( \frac{\partial }{\partial x_2} - \frac{\partial}{\partial y_2} \right) \rho - \frac{\gm_2 k T_2}{\hb ^2} (x_2 - y_2)^2 \rho
\end{align}
If we now perform the customary change of variables $ x = u + \hb z $, $ y = u - \hb z $ and $ \rho(\tbf{x}, \tbf{y}, t) \rightarrow P(\tbf{u}, \tbf{z}, t)$, then
\begin{align}
 (x_1 - x_2)^2 - (y_1 - y_2)^2 =& (u_1 + \hb z_1 - u_2 - \hb z_2)^2 - (u_1 - \hb z_1 - u_2 + \hb z_2)^2 \non 
\\ =&  4\hb \, (u_1 - u_2) (z_1 - z_2)
\end{align}
 so that
\begin{align}
 \frac{\partial P}{\partial t}(\tbf{u}, \tbf{z}, t) =& \Bigl\{ \frac{\im}{2 m} \left( \frac{\partial ^2}{\partial u_1 \partial z_1} + \frac{\partial^2}{\partial u_2 \partial z_2} \right) - 2 \im m \og_0 ^2 (u_1 - u_2) (z_1 - z_2) \non
\\ &- \frac{\gm_1}{m} \, z_1 \, \frac{\partial}{\partial z_1} - 4 \gm_1 k T_1 \, z_1 ^2 - \frac{\gm_2}{m} \, z_2 \, \frac{\partial}{\partial z_2} - 4 \gm_2 k T_2 \,  z_2 ^2 \Bigr\} P(\tbf{u}, \tbf{z}, t)
\end{align}
Let us apply the Fourier Transform
\be {\tilde P}(\tbf{q}, \tbf{z}, t) = \int du_1 \, du_2 \, P(\tbf{u}, \tbf{z}, t) e^{- \im q_1 u_1 - \im q_2 u_2}\ee
whose inverse is
\be P(\tbf{u}, \tbf{z}, t) = \frac{1}{4\pi ^2} \int dq_1 \, dq_2 \, {\tilde P}(\tbf{q}, \tbf{z}, t) e^{\im q_1 u_1 + \im q_2 u_2}\ee
Since
\be - \im (u_1 - u_2) P e^{- \im q_1 u_1} \, e^{- \im q_2 u_2} = \left( \frac{\partial}{\partial q_1} - \frac{\partial}{\partial q_2} \right) P e^{- \im q_1 u_1} \, e^{- \im q_2 u_2} \ee
we get the differential equation
\begin{align} \label{eqtosolve}
\lefteqn{ \frac{\partial {\tilde P}}{\partial t}(\tbf{q}, \tbf{z}, t)} \non 
\\ &= \Biggl\{ - \frac{1}{2 m} \left( q_1 \frac{\partial}{\partial z_1} + q_2 \frac{\partial}{\partial z_2} \right) - \frac{\gm_1}{m} \, z_1 \, \frac{\partial}{\partial z_1} - \frac{\gm_2}{m} \, z_2 \, \frac{\partial}{\partial z_2} \non 
\\ & + 2 m \og_0 ^2 \left( \frac{\partial}{\partial q_1} - \frac{\partial}{\partial q_2} \right) (z_1 - z_2) - 4 \gm_1 k T_1 \, z_1 ^2 - 4 \gm_2 k T_2 \, z_2 ^2  \Biggr\} {\tilde P}(\tbf{q}, \tbf{z}, t)
\end{align}

This equation can be solved using the method of characteristics. We have the characteristic equation
\be  \frac{\partial \tbf{v}}{\partial t} = \frac{M}{2m} \tbf{v} \ee
with $\tbf{v}^T = (z_1, z_2, q_1, q_2)$ and 

\be M= \left( \begin{array}{cccc} 2\gm_1 & 0 & 1 & 0 \\ 0 & 2\gm_2 & 0 & 1 \\ - 4 m^2 \og_0 ^2 & 4 m^2 \og_0 ^2 & 0 & 0 \\ 4 m^2 \og_0 ^2 & - 4 m^2 \og_0 ^2 & 0 & 0 \end{array} \right)
 \ee

On a characteristic we then have

\be \frac{d {\tilde P}}{d t} = - 4 k( \gm_1 T_1 z_1 ^2 + \gm_2 T_2 z_2 ^2) {\tilde P}\label{charac} \ee

Taking $\gm_1 = \gm_2=\gm$ and $T_1 = T_2 = T$ for simplicity, the eigenvalues and eigenvectors of M can be computed to be 
\be \bds{\lb}^T = \left(0, 2 \gm , \gm + \sqrt{\gm ^2 - 8 m^2 \og_0 ^2} , \gm - \sqrt{\gm ^2 - 8 m^2 \og_0 ^2} \right) = ( \lb_1, \lb_2, \lb_+, \lb_- )\ee
and
\be Q= \left( \begin{array}{cccc} -\frac{1}{2\gm} & 1 & \frac{1}{\lb_-} & \frac{1}{\lb_+} \\ -\frac{1}{2\gm} & 1 & - \frac{1}{\lb_-} & - \frac{1}{\lb_+} \\ 1 & 0 & -1 & -1 \\ 1 & 0 & 1 & 1 \end{array}
 \right) \ee

Since $Q^{-1} M Q = D$ where D is the diagonal matrix, we need $Q^{-1}$ :
\be Q^{-1} = \left( \begin{array}{cccc} 0 & 0 & \half & \half 
\\ \half & \half & \frac{1}{4 \gm}  & \frac{1}{4 \gm}
\\ \frac{\lb_+ \lb_-}{2 (\lb_+ - \lb_-)} & - \frac{\lb_+ \lb_-}{2 (\lb_+ - \lb_-)} & \frac{\lb_-}{2 (\lb_+ - \lb_-)} & -\frac{\lb_-}{2 (\lb_+ - \lb_-)}
\\ - \frac{\lb_+ \lb_-}{2 (\lb_+ - \lb_-)} & \frac{\lb_+ \lb_-}{2 (\lb_+ - \lb_-)} & - \frac{\lb_+}{2 (\lb_+ - \lb_-)} &  \frac{\lb_+}{2 (\lb_+ - \lb_-)} \end{array} \right)\ee

Using these results, one can rewrite the differential equation as
\begin{align} 2 m \frac{\partial \tbf{v}}{\partial t} =& Q D Q^{-1} \tbf{v} 
\\ \Longleftrightarrow \frac{\partial \tbf{w}}{\partial t} =& \frac{D}{2m} \tbf{w} 
\\ \text{with} \quad \tbf{w} =& Q^{-1} \tbf{v}
\end{align}
This is easily solved :
\be  \tbf{w}(t) = \tbf{w}(0) \, e^{Dt / 2m}  \ee
or more explicitly
\begin{align} \begin{array}{c}
w_1(t) = w_1 (0) 
\\ w_2(t) = w_2(0) e^{\gm t / m} 
\\ w_3(t) = w_3(0) e^{\lb_+ t / 2 m} 
\\ w_4(t) = w_4(0) e^{\lb_- t / 2 m} \end{array}  \non 
\end{align}
We can then write 
\be \tbf{v}(t) = Q e^{D t / 2m} Q^{-1} \tbf{v}_0 \ee
with \be e^{D t / 2m} = \left( \begin{array}{cccc} 1 & 0 & 0 & 
\\ 0 & e^{\gm t / m} & 0 & 0 \\ 0 & 0 & e^{\lb_+ t / 2m} & 0 
\\ 0 & 0 & 0 & e^{\lb_- t / 2m} \end{array} \right)
\ee
We get
{\alld
\begin{align}
 \tbf{w}_0 = Q^{-1} \tbf{v}_0 =& \left( \begin{array}{cccc} 0 & 0 & \half & \half 
\\ \half & \half & \frac{1}{4 \gm}  & \frac{1}{4 \gm}
\\ \frac{\lb_+ \lb_-}{2 (\lb_+ - \lb_-)} & - \frac{\lb_+ \lb_-}{2 (\lb_+ - \lb_-)} & \frac{\lb_-}{2 (\lb_+ - \lb_-)} & -\frac{\lb_-}{2 (\lb_+ - \lb_-)}
\\ - \frac{\lb_+ \lb_-}{2 (\lb_+ - \lb_-)} & \frac{\lb_+ \lb_-}{2 (\lb_+ - \lb_-)} & - \frac{\lb_+}{2 (\lb_+ - \lb_-)} &  \frac{\lb_+}{2 (\lb_+ - \lb_-)} \end{array} \right) \left( \begin{array}{c} {z_1}_0 \\ {z_2}_0 \\ {q_1}_0 \\ {q_2}_0 \end{array} \right) \non 
\\ =& \left( \begin{array}{c} \frac{{q_1}_0}{2} + \frac{{q_2}_0}{2} \\ \frac{{z_1}_0}{2} + \frac{{z_2}_0}{2} + \frac{{q_1}_0}{4 \gm} + \frac{{q_2}_0}{4 \gm}
\\ \frac{\lb_+ \lb_- {z_1}_0 - \lb_+ \lb_- {z_2}_0 + \lb_- {q_1}_0 - \lb_- {q_2}_0}{2 (\lb_+ - \lb_-)} \\ \frac{- \lb_+ \lb_- {z_1}_0 + \lb_+ \lb_- {z_2}_0 - \lb_+ {q_1}_0 + \lb_+ {q_2}_0}{2 (\lb_+ - \lb_-)}
      \end{array} \right)
\end{align}}
Then
{\alld
\begin{align}
 e^{D t / 2m} \tbf{w}_0 = & \left( \begin{array}{cccc} 1 & 0 & 0 & 
\\ 0 & e^{\gm t / m} & 0 & 0 \\ 0 & 0 & e^{\lb_+ t / 2m} & 0 
\\ 0 & 0 & 0 & e^{\lb_- t / 2m} \end{array} \right) \left( \begin{array}{c} \frac{{q_1}_0}{2} + \frac{{q_2}_0}{2} \\ \frac{{z_1}_0}{2} + \frac{{z_2}_0}{2} + \frac{{q_1}_0}{4 \gm} + \frac{{q_2}_0}{4 \gm}
\\ \frac{\lb_+ \lb_- {z_1}_0 - \lb_+ \lb_- {z_2}_0 + \lb_- {q_1}_0 - \lb_- {q_2}_0}{2 (\lb_+ - \lb_-)} \\ \frac{- \lb_+ \lb_- {z_1}_0 + \lb_+ \lb_- {z_2}_0 - \lb_+ {q_1}_0 + \lb_+ {q_2}_0}{2 (\lb_+ - \lb_-)} \end{array} \right) \non 
\\ =& \left( \begin{array}{c} \frac{{q_1}_0}{2} + \frac{{q_2}_0}{2} \\ \left(\frac{{z_1}_0}{2} + \frac{{z_2}_0}{2} + \frac{{q_1}_0}{4 \gm} + \frac{{q_2}_0}{4 \gm}\right) e^{\gm t / m}
\\ \left(\frac{\lb_+ \lb_- {z_1}_0 - \lb_+ \lb_- {z_2}_0 + \lb_- {q_1}_0 - \lb_- {q_2}_0}{2 (\lb_+ - \lb_-)} \right)e^{\lb_+ t / 2m}
\\ \left(\frac{- \lb_+ \lb_- {z_1}_0 + \lb_+ \lb_- {z_2}_0 - \lb_+ {q_1}_0 + \lb_+ {q_2}_0}{2 (\lb_+ - \lb_-)}\right) e^{\lb_- t / 2m}  \end{array}\right)
\end{align}}
and
{\alld
\begin{align}
 \tbf{v}(t) =& Q e^{D t / 2m} \tbf{w}_0 \non 
\\ =& \left( \begin{array}{cccc} -\frac{1}{2\gm} & 1 & \frac{1}{\lb_-} & \frac{1}{\lb_+} \\ -\frac{1}{2\gm} & 1 & - \frac{1}{\lb_-} & - \frac{1}{\lb_+} \\ 1 & 0 & -1 & -1 \\ 1 & 0 & 1 & 1 \end{array}
 \right) \left( \begin{array}{c} \frac{{q_1}_0}{2} + \frac{{q_2}_0}{2} \\ \left(\frac{{z_1}_0}{2} + \frac{{z_2}_0}{2} + \frac{{q_1}_0}{4 \gm} + \frac{{q_2}_0}{4 \gm}\right) e^{\gm t / m}
\\ \left(\frac{\lb_+ \lb_- {z_1}_0 - \lb_+ \lb_- {z_2}_0 + \lb_- {q_1}_0 - \lb_- {q_2}_0}{2 (\lb_+ - \lb_-)} \right)e^{\lb_+ t / 2m}
\\ \left(\frac{- \lb_+ \lb_- {z_1}_0 + \lb_+ \lb_- {z_2}_0 - \lb_+ {q_1}_0 + \lb_+ {q_2}_0}{2 (\lb_+ - \lb_-)}\right) e^{\lb_- t / 2m}  \end{array}\right) 
\end{align}}
We finally get $\tbf{v}_0$ as 
\be \left( \begin{array}{c} z_1 \\ z_2 \\ q_1 \\ q_2 \end{array} \right) = \mathcal{F}(t) \left( \begin{array}{c} {z_1}_0 \\ {z_2}_0 \\ {q_1}_0 \\ {q_2}_0 \end{array} \right)
 \non \ee
where the matrix $\mathcal{F}(t)$ is given by
{\alld
\begin{align} \label{matriceF}
 \mathcal{F}(t) =& \left( \begin{array}{cccc} \ap_+ (t) & \ap_- (t) & \dt_+ (t) & \dt_- (t)
\\ \ap_- (t) & \ap_+ (t)  & \dt_- (t) & \dt_+ (t)
\\  - \bt (t) & \bt (t) & \nu_+ (t) & \nu_- (t)
\\ \bt (t) & - \bt (t) & \nu_- (t) & \nu_+ (t) \end{array} \right) 
\\ \text{with} \non 
\\ \ap_{\pm} (t) =& \frac{e^{\gm t / m}}{2} \pm \frac{\lb_+ \, e^{\lb_+ t / 2m} - \lb_- \, e^{\lb_- t / 2m}}{2 (\lb_+ - \lb_-)} = \ap_1 (t) \pm \ap_2 (t)\non 
\\ \bt (t) =& \lb_+ \lb_- \frac{e^{\lb_+ t / 2m} - e^{\lb_- t / 2m}}{2 (\lb_+ - \lb_-)} \non 
\\ \dt_{\pm} (t) =&  - \frac{1}{4 \gm} + \frac{e^{\gm t / m}}{4 \gm} \pm \frac{e^{\lb_+ t / 2m} - e^{\lb_- t / 2m}}{2 (\lb_+ - \lb_-)} = \dt_1 (t) \pm \dt_2 (t) \non
\\ \nu_{\pm} (t) =&  \half \pm \frac{\lb_+ \, e^{\lb_- t / 2m} - \lb_- \, e^{\lb_+ t / 2m}}{2 (\lb_+ - \lb_-)} = \half \pm \nu (t)
\end{align}}

We can then insert $z_1 ^2(t)$ and $z_2 ^2(t)$ into (\ref{charac}), we get
\begin{align}
\frac{d {\tilde P}}{d t} = - 4 \gm k T (&z_1 ^2 + z_2 ^2) {\tilde P} \non 
\\ = - 4 \gm k T \Biggl\{ & \left({z_1}_0 ^2 + {z_2}_0 ^2  \right) \left( {\ap_+ (t)}^2 + {\ap_- (t)}^2 \right)  + 4 {z_1}_0 {z_2}_0 {\ap_+ (t)} {\ap_- (t)} \non 
\\ & + \left( {q_1}_0 ^2 + {q_2}_0 ^2\right) \left( {\dt_+ (t)}^2 + {\dt_- (t)}^2 \right) + 4 {q_1}_0 {q_2}_0 {\dt_- (t)} {\dt_+ (t)} \non 
\\ & + 2({z_1}_0 {q_1}_0 + {z_2}_0 {q_2}_0) \left( \ap_- (t) \dt_- (t) + \ap_+ (t) \dt_+ (t) \right) \non 
\\ &+ 2({z_1}_0 {q_2}_0 + {z_2}_0 {q_1}_0) \left( \ap_- (t) \dt_+ (t) + \ap_+ (t) \dt_- (t) \right) 
\Biggr\} {\tilde P}
\end{align}
After simplification, we can write
{\alld
\begin{align}
 \lefteqn{\frac{d {\tilde P}}{d t} = - 4 \gm k T \Biggl\{} \non 
\\ & ({z_1}_0 ^2 + {z_2}_0 ^2) \left( \frac{e^{2 \gm t /m}}{2} + \frac{\lb_+ ^2 e^{\lb_+ t /m}}{2 (\lb_+ - \lb_-) ^2} + \frac{\lb_- ^2 e^{\lb_- t /m}}{2 (\lb_+ - \lb_-) ^2} - \frac{\lb_+ \lb_- e^{\gm t / m}}{(\lb_+ - \lb_-) ^2}\right) \non 
\\ & + 4 {z_1}_0 {z_2}_0 \left( \frac{e^{2 \gm t /m}}{4} - \frac{\lb_+ ^2 e^{\lb_+ t /m}}{4 (\lb_+ - \lb_-) ^2} - \frac{\lb_- ^2 e^{\lb_- t /m}}{4 (\lb_+ - \lb_-) ^2} + \frac{\lb_+ \lb_- e^{\gm t / m}}{2 (\lb_+ - \lb_-) ^2}\right) \non
\\ & + ({q_1}_0 ^2 + {q_2}_0 ^2) \non 
\\ & \quad \times \left( \frac{1}{8 \gm ^2} + \frac{e^{2 \gm t /m}}{8 \gm ^2} - \frac{e^{\gm t /m}}{4 \gm  ^2} + \frac{e^{\lb_+ t /m}}{2 (\lb_+ - \lb_-) ^2} + \frac{e^{\lb_- t /m}}{2 (\lb_+ - \lb_-)^2} - \frac{ e^{\gm t / m}}{(\lb_+ - \lb_-) ^2}\right) \non 
\\ & + 4 {q_1}_0 {q_2}_0 \non 
\\ & \quad \times \left( \frac{1}{16 \gm ^2} + \frac{e^{2 \gm t /m}}{16 \gm ^2} - \frac{e^{\gm t /m}}{8 \gm  ^2} - \frac{e^{\lb_+ t /m}}{4 (\lb_+ - \lb_-) ^2} - \frac{e^{\lb_- t /m}}{4 (\lb_+ - \lb_-) ^2} + \frac{ e^{\gm t / m}}{2 (\lb_+ - \lb_-) ^2}\right) \non 
\\ & + ({z_1}_0 {q_1}_0 + {z_2}_0 {q_2}_0) \non 
\\ & \quad \times \left( \frac{e^{2 \gm t /m}}{2 \gm} - \frac{e^{\gm t /m}}{2 \gm} + \frac{\lb_+ e^{\lb_+ t /m}}{(\lb_+ - \lb_-) ^2} + \frac{\lb_- e^{\lb_- t /m}}{(\lb_+ - \lb_-) ^2} - \frac{2 \gm e^{\gm t / m}}{(\lb_+ - \lb_-) ^2} \right)\non 
\\ & + ({z_1}_0 {q_2}_0 + {z_2}_0 {q_1}_0) \non 
\\* & \times \left( \frac{e^{2 \gm t /m}}{2 \gm} - \frac{e^{\gm t / m}}{2 \gm} - \frac{\lb_+ e^{\lb_+ t /m}}{(\lb_+ - \lb_-) ^2} - \frac{\lb_- e^{\lb_- t /m}}{(\lb_+ - \lb_-) ^2} + \frac{2 \gm e^{\gm t / m}}{(\lb_+ - \lb_-) ^2} \right) \Biggr\} {\tilde P}
\end{align}}
This can easily be integrated :
{\alld
\begin{align}
 {\tilde P} = {\tilde P}_0 & \exp \left[ - 4 \gm k T \left(\chi_1 ({z_1}_0 ^2 + {z_2}_0 ^2) + \ta_1 {z_1}_0 {z_2}_0 + \chi_2 ({q_1}_0 ^2 + {q_2}_0 ^2) + \ta_2 {q_1}_0 {q_2}_0 \right) \right] \non 
\\ & \times \exp \left[ - 4 \gm k T \left( \Lb_{1} ({z_1}_0 {q_1}_0 + {z_2}_0 {q_2}_0) + \Lb_{2} ({z_1}_0 {q_2}_0 + {z_2}_0 {q_1}_0) \right) \right] \non 
\\ \text{with} \non 
\\ \chi_1 =& \frac{m}{4 \gm} (e^{2 \gm t /m} - 1) + \frac{m \lb_+}{2 (\lb_+ - \lb_-) ^2} (e^{\lb_+ t /m} - 1) + \frac{m \lb_-}{2 (\lb_+ - \lb_-) ^2} (e^{\lb_- t / m} - 1) \non \\* & - \frac{8 m ^3 \og_0 ^2}{\gm (\lb_+ - \lb_-) ^2} (e^{\gm t / m} - 1) \non 
\\ \ta_1 =& \frac{m}{2 \gm} (e^{2 \gm t /m} - 1) - \frac{m \lb_+}{(\lb_+ - \lb_-) ^2} (e^{\lb_+ t /m} - 1) - \frac{m \lb_-}{(\lb_+ - \lb_-)^2} (e^{\lb_- t / m} - 1) \non \\* & + \frac{16 m ^3 \og_0 ^2}{\gm (\lb_+ - \lb_-)^2} (e^{\gm t / m} - 1) \non 
\\ \chi_2 =& \frac{t}{8 \gm ^2} + \frac{m}{16 \gm ^3} (e^{2 \gm t /m} - 1) - \frac{m}{4 \gm ^3} (e^{\gm t / m} - 1) + \frac{m (e^{\lb_+ t /m} - 1)}{2 (\lb_+ - \lb_-) ^2 \lb_+} \non \\* & + \frac{m (e^{\lb_- t / m} - 1)}{2 (\lb_+ - \lb_-) ^2 \lb_-} - \frac{m (e^{\gm t / m} - 1)}{\gm (\lb_+ - \lb_-)^2} \non 
\\ \ta_2 =& \frac{t}{4 \gm ^2} + \frac{m}{8 \gm ^3} (e^{2 \gm t /m} - 1) - \frac{m}{2 \gm ^3} (e^{\gm t / m} - 1) - \frac{m (e^{\lb_+ t /m} - 1)}{(\lb_+ - \lb_-) ^2 \lb_+} \non \\* &  - \frac{m (e^{\lb_- t / m} - 1)}{(\lb_+ - \lb_-) ^2 \lb_-} + \frac{2 m (e^{\gm t / m} - 1)}{\gm (\lb_+ - \lb_-) ^2} \non 
\\ \Lb_{1} =& \frac{m}{4 \gm ^2} (e^{2 \gm t /m} - 1) - \frac{m}{2 \gm ^2} (e^{\gm t / m} - 1) + \frac{m (e^{\lb_+ t /m} - 1)}{(\lb_+ - \lb_-) ^2} + \frac{m (e^{\lb_- t / m} - 1)}{(\lb_+ - \lb_-) ^2}  \non \\* & - \frac{2 m (e^{\gm t / m} - 1)}{(\lb_+ - \lb_-) ^2} \non 
\\ \Lb_{2} =& \frac{m}{4 \gm ^2} (e^{2 \gm t /m} - 1) - \frac{m}{2 \gm ^2} (e^{\gm t / m} - 1) - \frac{m (e^{\lb_+ t /m} - 1)}{(\lb_+ - \lb_-) ^2} - \frac{m (e^{\lb_- t / m} - 1)}{(\lb_+ - \lb_-) ^2}  \non \\* & + \frac{2 m (e^{\gm t / m} - 1)}{(\lb_+ - \lb_-) ^2} 
\end{align}}

We recall the initial state (\ref{tildep02p})
\begin{align}
 {\tilde P}(\tbf{q}_0,\tbf{z}_0;0)=& \exp\left[-\ep_{+}\hb ^2 {z_1}_0 ^2-\ep_{+}\hb ^2 {z_2}_0 ^2 + 2\ep_{-}\hb ^2 {z_1}_0 {z_2}_0\right]\non
\\ &\times \exp\left[-\frac{\ep_{+}}{4(\ep_{+}^2-\ep_{-}^2)}{q_2}_0 ^2 -\frac{\ep_{+}}{4(\ep_{+}^2-\ep_{-}^2)} {q_1}_0 ^2-\frac{\ep_{-}}{2(\ep_{+}^2-\ep_{-}^2)}{q_1}_0 {q_2}_0 \right]
\end{align}

Using $\tbf{v}_0 = \mathcal{F}(-t) \tbf{v}$, we can readily get 
\be \left( \begin{array}{c} {z_1}_0 \\ {z_2}_0 \\ {q_1}_0 \\ {q_2}_0 \end{array} \right) = \mathcal{F}(-t) \left( \begin{array}{c} z_1 \\ z_2 \\ q_1 \\ q_2 \end{array} \right)
 \non \ee

Then we can write 
{\alld
\begin{align}
 {\tilde P} =&  \exp \left[ -(\ep_+ \hb ^2 + 4 \gm k T \chi_1) \left(\ap_+ (-t) z_1 + \ap_+ (-t) z_2 + \dt_+ (-t) q_1 + \dt_- (-t) q_2 \right)^2 \right] \non 
\\ & \times \exp \left[ - (\ep_+ \hb ^2 + 4 \gm k T \chi_1) \left( \ap_+ (-t) z_1 + \ap_+ (-t) z_2 + \dt_- (-t) q_1 + \dt_+ (-t) q_2 \right)^2 \right] \non 
\\ & \times \exp \bigl[ (2 \ep_- \hb ^2 - 4 \gm k T \ta_1)  \non \\* & \hspace{0.45 in} \times \left( \ap_+ (-t) z_1 + \ap_+ (-t) z_2 + \dt_+ (-t) q_1 + \dt_- (-t) q_2 \right) \non \\* & \hspace{0.75 in} \times  \left( \ap_+ (-t) z_1 + \ap_+ (-t) z_2 + \dt_- (-t) (-t) q_1 + \dt_+ (-t) q_2 \right) \bigr] \non 
\\ & \times \exp \left[- (\frac{\ep_+}{4(\ep_+ ^2 - \ep_- ^2)} + 4 \gm k T \chi_2) \left( -\bt (-t) z_1 + \bt (-t) z_2 + \nu_+ (-t) q_1 + \nu_- (-t) q_2 \right)^2\right] \non 
\\ & \times \exp \left[ - (\frac{\ep_+}{4(\ep_+ ^2 - \ep_- ^2)} + 4 \gm k T \chi_2) \left( \bt (-t) z_1 - \bt (-t) z_2 + \nu_- (-t) q_1 + \nu_+ (-t) q_2 \right)^2 \right] \non 
\\ & \times \exp \bigl[ - (\frac{\ep_-}{2(\ep_+ ^2 - \ep_- ^2)} + 4 \gm k T \ta_2) \non \\* & \hspace{0.45 in} \times \left( -\bt (-t) z_1 + \bt (-t) z_2 + \nu_+ (-t) q_1 + \nu_- (-t) q_2 \right) \non \\* & \hspace{0.75 in} \times  \left( \bt (-t) z_1 - \bt (-t) z_2 + \nu_- (-t) q_1 + \nu_+ (-t) q_2 \right) \bigr] \non 
\\ & \times \exp \bigl[ - 4 \gm k T \Lb_{1} \left(\ap_+ (-t) z_1 + \ap_+ (-t) z_2 + \dt_+ (-t) q_1 + \dt_- (-t) q_2 \right) \non \\* & \hspace{0.75 in} \times \left( -\bt (-t) z_1 + \bt (-t) z_2 + \nu_+ (-t) q_1 + \nu_- (-t) q_2 \right) \bigr] \non 
\\ & \times \exp \bigl[ - 4 \gm k T \Lb_{1}\left(\ap_+ (-t) z_1 + \ap_+ (-t) z_2 + \dt_- (-t) q_1 + \dt_+ (-t) q_2 \right) \non \\* & \hspace{0.75 in} \times \left( \bt (-t) z_1 - \bt (-t) z_2 + \nu_- (-t) q_1 + \nu_+ (-t) q_2 \right) \bigr] \non 
\\ & \times \exp \bigl[ - 4 \gm k T \Lb_{2} \left(\ap_+ (-t) z_1 + \ap_+ (-t) z_2 + \dt_+ (-t) q_1 + \dt_- (-t) q_2 \right) \non \\* & \hspace{0.75 in} \times \left( \bt (-t) z_1 - \bt (-t) z_2 + \nu_- (-t) q_1 + \nu_+ (-t) q_2 \right) \bigr] \non 
\\ & \times \exp \bigl[ - 4 \gm k T \Lb_{2}  \left(\ap_+ (-t) z_1 + \ap_+ (-t) z_2 + \dt_- (-t) q_1 + \dt_+ (-t) q_2 \right) \non \\* & \hspace{0.75 in} \times \left( -\bt (-t) z_1 + \bt (-t) z_2 + \nu_+ (-t) q_1 + \nu_- (-t) q_2 \right) \bigr] 
\end{align}}

After some more unpleasant algebra, we can write
{\alld
\begin{align}
 {\tilde P} =& \exp \bigl[- \mathcal{A} q_1 ^2 - \mathcal{A} q_2 ^2 - \mathcal{E} q_1 q_2 - \mathcal{B} z_1 ^2 - \mathcal{B} z_2 ^2 - \mathcal{D} z_1 z_2 \non 
\\ & \hspace{0.1 in} - \mathcal{C}_1 z_1 q_1 - \mathcal{C}_1 z_2 q_2 - \mathcal{C}_2 z_1 q_2 - \mathcal{C}_2 z_2 q_1 \bigr] 
\non \\ \text{with} \non 
\\ \mathcal{A} =& ({\dt_1 (-t)} ^2 + {\dt_2 (-t)} ^2) (\ep_+ \hb ^2 + 4 \gm k T \chi_1) - ({\dt_1 (-t)} ^2 - {\dt_2 (-t)} ^2) (2 \ep_- \hb ^2 - 4 \gm k T \ta_1) \non
\\ & + (\frac{1}{4} + {\nu (-t)} ^2) (\frac{\ep_+}{4(\ep_+ ^2 - \ep_- ^2)} + 4 \gm k T \chi_2) + (\frac{1}{4} - {\nu (-t)} ^2) (\frac{\ep_-}{2(\ep_+ ^2 - \ep_- ^2)} + 4 \gm k T \ta_2) \non 
\\ & + 4 \gm k T \left( (\dt_1 (-t) + 2 \dt_2 (-t) \nu (-t)) \Lb_{1} + (\dt_1 (-t) - 2 \dt_2 (-t) \nu (-t)) \Lb_{2} \right)  \\ \non
\\ \mathcal{B} =& ({\ap_1 (-t)} ^2 + {\ap_2 (-t)} ^2) (\ep_+ \hb ^2 + 4 \gm k T \chi_1) - ({\ap_1 (-t)} ^2 - {\ap_2 (-t)} ^2) (2 \ep_- \hb ^2 - 4 \gm k T \ta_1) \non
\\ & + 2 {\bt (-t)} ^2  (\frac{\ep_+}{4(\ep_+ ^2 - \ep_- ^2)} + 4 \gm k T \chi_2) - {\bt (-t)} ^2 (\frac{\ep_-}{2(\ep_+ ^2 - \ep_- ^2)} + 4 \gm k T \ta_2) \non 
\\ & + 8 \gm k T  \bt (-t) \ap_2 (-t) \left(\Lb_{2} - \Lb_{1} \right)  \\ \non 
\\ \mathcal{D} =& 4 ({\ap_1 (-t)} ^2 - {\ap_2 (-t)} ^2)  (\ep_+ \hb ^2 + 4 \gm k T \chi_1)  - ({\ap_1 (-t)} ^2 + {\ap_2 (-t)} ^2) (2 \ep_- \hb ^2 - 4 \gm k T \ta_1) \non
\\ & - 4 {\bt (-t)} ^2 (\frac{\ep_+}{4(\ep_+ ^2 - \ep_- ^2)} + 4 \gm k T \chi_2) + 2 {\bt (-t)} ^2 (\frac{\ep_-}{2(\ep_+ ^2 - \ep_- ^2)} + 4 \gm k T \ta_2) \non 
\\ & + 16 \gm k T \bt (-t) \ap_2 (-t) \left( \Lb_{1} - \Lb_{2} \right) \\ \non
\\ \mathcal{E} =& 4 ({\dt_1 (-t)} ^2 - {\dt_2 (-t)} ^2) (\ep_+ \hb ^2 + 4 \gm k T \chi_1) - ( {\dt_1 (-t)} ^2 + {\dt_2 (-t)} ^2) (2 \ep_- \hb ^2 - 4 \gm k T \ta_1) \non
\\ & + 4 (\frac{1}{4} - {\nu (-t)} ^2) (\frac{\ep_+}{4(\ep_+ ^2 - \ep_- ^2)} + 4 \gm k T \chi_2) + (\frac{1}{4} + {\nu (-t)} ^2)  (\frac{\ep_-}{2(\ep_+ ^2 - \ep_- ^2)} + 4 \gm k T \ta_2) \non 
\\ & + 8 \gm k T \left( (\dt_1 (-t) - 2 \dt_2 (-t) \nu (-t)) \Lb_{1} + (\dt_1 (-t) + 2 \dt_2 (-t) \nu (-t)) \Lb_{2} \right) \\ \non
\\ \mathcal{C}_1 =& 4 (\ap_1 (-t) \dt_1 (-t) + \ap_2 (-t) \dt_2 (-t)) (\ep_+ \hb ^2 + 4 \gm k T \chi_1) \non
\\ & - 2 (\ap_1 (-t) \dt_1 (-t) - \ap_2 (-t) \dt_2 (-t)) (2 \ep_- \hb ^2 - 4 \gm k T \ta_1) \non
\\ & + 4 \bt (-t) \nu (-t) (\frac{\ep_+}{4(\ep_+ ^2 - \ep_- ^2)} + 4 \gm k T \chi_2) - 2 \bt (-t) \nu (-t) (\frac{\ep_-}{2(\ep_+ ^2 - \ep_- ^2)} + 4 \gm k T \ta_2) \non 
\\ & + 4 \gm k T (\ap_1 (-t) + 2 \ap_2 (-t) \nu (-t) - 2 \bt (-t) \dt_2 (-t)) \Lb_{1} \non 
\\ & +  4 \gm k T (\ap_1 (-t) - 2 \ap_2 (-t) \nu (-t) + 2 \bt (-t) \dt_2 (-t) ) \Lb_{2} \\ \non
\\ \mathcal{C}_2 =& 4 (\ap_1 (-t) \dt_1 (-t) - \ap_2 (-t) \dt_2 (-t)) (\ep_+ \hb ^2 + 4 \gm k T \chi_1) \non 
\\ & - 2 (\ap_1 (-t) \dt_1 (-t) + \ap_2 (-t) \dt_2 (-t)) (2 \ep_- \hb ^2 - 4 \gm k T \ta_1) \non
\\ & - 4 \bt (-t) \nu (-t) (\frac{\ep_+}{4(\ep_+ ^2 - \ep_- ^2)} + 4 \gm k T \chi_2) + 2 \bt (-t) \nu (-t) (\frac{\ep_-}{2(\ep_+ ^2 - \ep_- ^2)} + 4 \gm k T \ta_2) \non 
\\ & + 4 \gm k T (\ap_1 (-t) - 2 \ap_2 (-t) \nu (-t) + 2 \bt (-t) \dt_2 (-t) ) \Lb_{1} \non 
\\ & +  4 \gm k T (\ap_1 (-t) + 2 \ap_2 (-t) \nu (-t) - 2 \bt (-t) \dt_2 (-t)) \Lb_{2} 
\end{align}}

\subsection*{Covariance Matrix and Logarithmic Negativity}

Using (\ref{tildePt}) and (\ref{expxx} - \ref{exppp}), the covariance matrix can then be written as
\begin{align} \gins = \left[
\begin{array}{cccc} 4 \mathcal{A} & - \mathcal{C}_1 & 2 \mathcal{E} & - \mathcal{C}_2 \\ - \mathcal{C}_1 & \mathcal{B} & - \mathcal{C}_2 & \mathcal{D}/2 \\ 2 \mathcal{E} & - \mathcal{C}_2 & 4 \mathcal{A} & - \mathcal{C}_1 \\ - \mathcal{C}_2 & \mathcal{D}/2 & - \mathcal{C}_1 & \mathcal{B} 
\end{array}
\right] \end{align}
The partial transpose is then
\be \gins ^T = \left[
\begin{array}{cccc} 4 \mathcal{A} & \mathcal{C}_1 & 2 \mathcal{E} & - \mathcal{C}_2 \\ \mathcal{C}_1 & \mathcal{B} & \mathcal{C}_2 & - \mathcal{D}/2 \\ 2 \mathcal{E} & \mathcal{C}_2 & 4 \mathcal{A} & - \mathcal{C}_1 \\ - \mathcal{C}_2 & - \mathcal{D}/2 & - \mathcal{C}_1 & \mathcal{B} 
\end{array}
\right] \ee

One can again calculate $- \sg \gins ^T \sg \gins ^T$ to get
\be - \sg \gins ^T \sg \gins ^T = \left( \begin{array}{cccc}
\nins_{11} & \nins_{12} & \nins_{13}  & \nins_{14} 
\\ \nins_{21} & \nins_{22} & \nins_{23}  & \nins_{24} 
\\ \nins_{31} & \nins_{32}  & \nins_{33}  & \nins_{34}
\\ \nins_{41} & \nins_{42}  & \nins_{43} & \nins_{44}
\end{array}
\right)
\ee
where $\nins_{12} = \nins_{21} = \nins_{34} = \nins_{43} = 0$
\begin{align}
 \nins_{11} = \nins_{22} = \nins_{33} = \nins_{44} =& 4 \mathcal{A} \mathcal{B} - \mathcal{D} \mathcal{E} + \mathcal{C}_2 ^2 - \mathcal{C}_1 ^2 \non 
\\ \nins_{13} = \nins_{24} = \nins_{31} = \nins_{42} =& 2 \mathcal{E} \mathcal{B} - 2 \mathcal{A} \mathcal{D} \non 
\\ \nins_{14} = - \nins_{32}  =& \mathcal{C}_1 \mathcal{D} - 2 \mathcal{C}_2 \mathcal{B} \non 
\\ \nins_{23} = - \nins_{41}  =& 8 \mathcal{A} \mathcal{C}_2 - 4 E \mathcal{C}_1 \non 
\end{align}

The eigenvalues of $-\sg \gins ^T \sg \gins ^T$ can then be determined to be :
\begin{align}
 \lb_{1,2} ^T =& \nins_{11} + \sqrt{ \nins_{13} ^2 - \nins_{14} \nins_{23} }\non 
\\ \lb_{3,4} ^T =& \nins_{11} - \sqrt{ \nins_{13} ^2 - \nins_{14} \nins_{23} }
\end{align}
The logarithmic negativity then becomes
\be \mathcal{L_{\mathcal{N}}}(\rho) = - 2 \left( \log_2\left(\min(1,\vert \lb_{1,2} ^T  \vert)\right) + \log_2\left(\min(1,\vert \lb_{3,4} ^T  \vert)\right) \right) \ee

\section{Observations and Remarks}

\sbp Figure~\ref{OverToUnd} shows how the harmonic potential influences the entanglement. One can notice that in the highly over-damped case ($\gm = 3$ and $\og_0 = 1$), the entanglement vanishes as quickly with the potential as it does without the potential. Figure~\ref{VeryOD} illustrates the over-damped behaviour. One may easily notice that all the curves seem to coincide, suggesting that an over-damped harmonic interaction may do little to improve on ESD. Moreover Figure~\ref{BLN_VOD_detail} shows that in fact, the entanglement vanishes at shorter times, though exponentially close to the free evolution vanishing time. In the slightly under-damped case ($\gm = 1.5$), the entanglement also disappears around the same time as in the free evolution case. In the highly-under-damped case ($\gm = 0.2$), one can observe that $\mathcal{L}$ decreases non-uniformly, disappears then re-appears for a short while. This leads one to wonder how keeping the system slightly or more under-damped may help the entanglement. Indeed, Figure~\ref{BLN_slightUD} lets us observe that as the damping decreases, the logarithmic negativity vanishes at longer times. However, if the damping decreases further, as can be seen on Figure~\ref{Slight_to_VUD}, $\mathcal{L}$ oscillates to a constant value greater than 0. Figure~\ref{VeryUD} shows more example of this behaviour. 

\sbp To understand better this behaviour, one may recall M's eigenvalues, namely 
\be \bds{\lb}^T = \left(0, 2 \gm , \gm + \sqrt{\gm ^2 - 8 m^2 \og_0 ^2} , \gm - \sqrt{\gm ^2 - 8 m^2 \og_0 ^2} \right) = ( \lb_1, \lb_2, \lb_+, \lb_- ) \non \ee
and notice that the critical eigenvalues $\lb_+$ and $\lb_-$ become complex as $\gm ^2 < 8 m ^2 \og_0 ^2 $. This brings an oscillatory term into the covariance matrix terms and consequently, into the symplectic eigenvalues and the logarithmic negativity, resulting into the entanglement's oscillatory convergence towards a constant. 

One may recall our physical interpretation of the initial state. The harmonic potential may now be interpreted as a pulse, such as that used to increase the quality factor in lasers. In the over-damped case, the pulse conflicts with the wavepacket's original interference, resulting in a faster loss of entanglement. In the slightly under-damped case, this oscillatory behaviour begins to resonate with the wavepacket, sustaining the interference pattern for a while. In the highly under-damped case, this resonance dominates and the initial interference pattern becomes drowned into a larger interference packet resulting from the harmonic potential's oscillations.

\sbp In \cite{Ficek:2006}, Ficek and Tan\'as study a two qubits system coupled to a radiation field where they allow spontaneous decay of the atoms. They show that the entanglement vanishes but is revived twice, with different reasons for each revival. The first revival is due to the regaining of coherence due to the spontaneous emission, while the second is related to the asymmetric state population. In \cite{Ficek:2008}, the authors study the emergence of entanglement between two initially non-entangled qubits due to spontaneous emission, provided both atoms are initially excited and in the asymmetric state. They show this creation of entanglement to be a function of the separation between the atoms. Their results suggest that allowing an interaction between two particles initially entangled will delay the vanishing of the entanglement and revive it, or create entanglement between two initially non-entangled particles. We show that when using a harmonic potential as the interaction, the entanglement revival depends on how strong the coupling is with respect to the oscillator's frequency. In fact, we show that if the damping is sufficiently low, the entanglement survives for very long times.

\begin{figure}
 \begin{center}
	\includegraphics[scale=0.5]{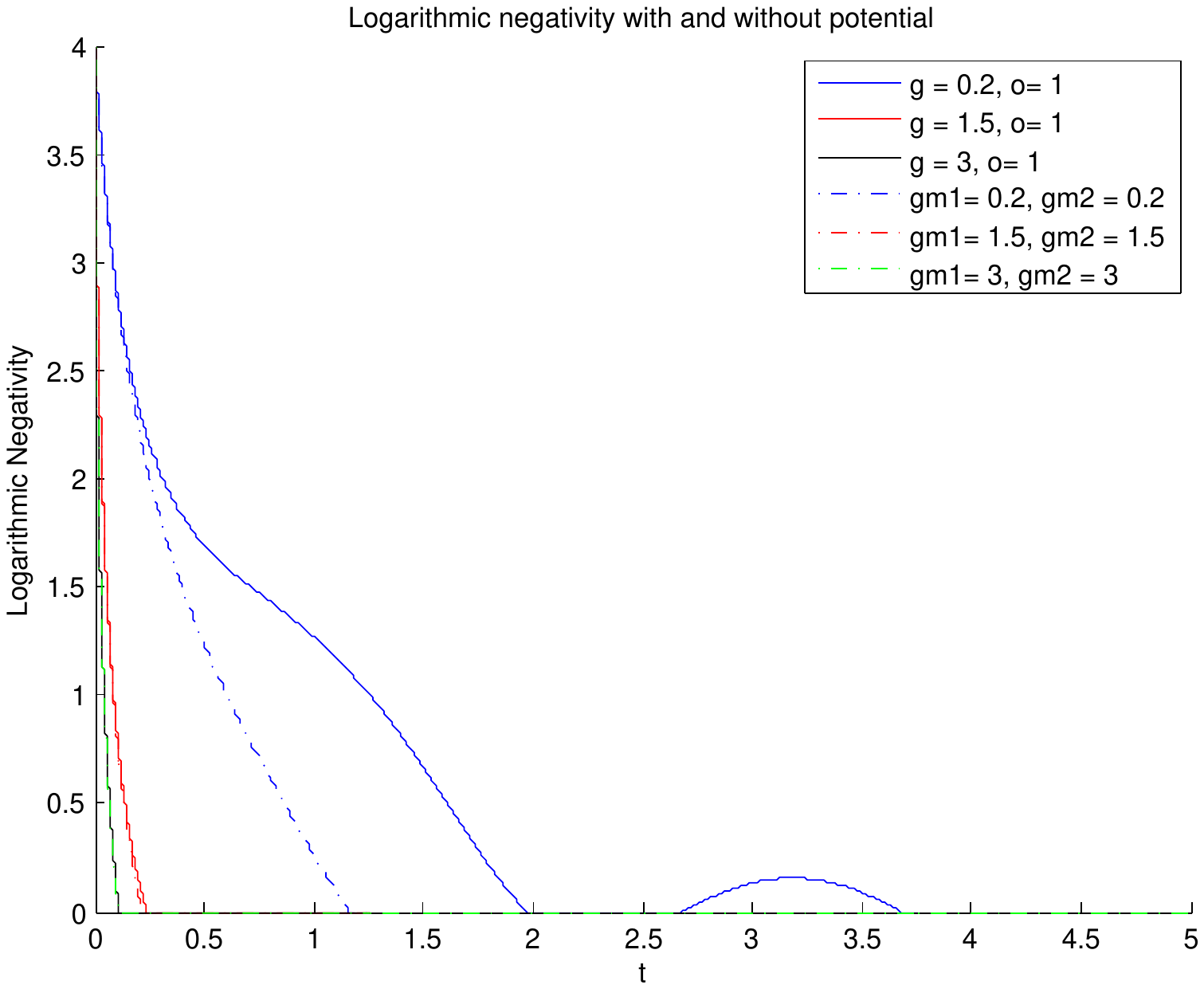}
	\caption{Logarithmic Negativities with and without potential for $\og_0 = 1$} 
	\medskip{The dashed lines represent $\mathcal{L}$ plotted without the potential. The plots are ($\gm = \gm_1 = \gm_2$): green : $\gm = 3$, red : $\gm = 1.5$ and blue : $\gm = 0.2$}
	\label{OverToUnd}
\end{center}
\end{figure}

\begin{figure}
 \begin{center}
	\includegraphics[scale=0.5]{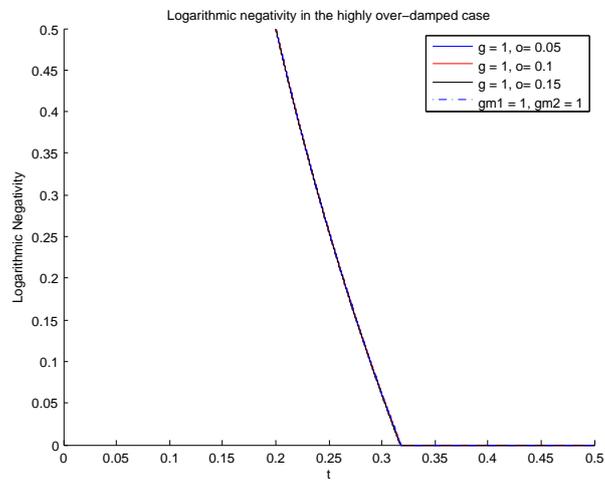}
	\caption{$\mathcal{L}$ in the highly over-damped case}
	\medskip{The plots are obtained keeping $\gm = \gm_1 = \gm_2 = 1$ and letting $\og_0$ vary as : dashed : $\og_0 = 0$, blue : $\og_0 = 0.05$, red : $\og_0 = 0.1$ and black : $\og_0 = 0.15$}
	\label{VeryOD}
\end{center}
\end{figure}

\begin{figure}
 \begin{center}
	\includegraphics[scale=0.5]{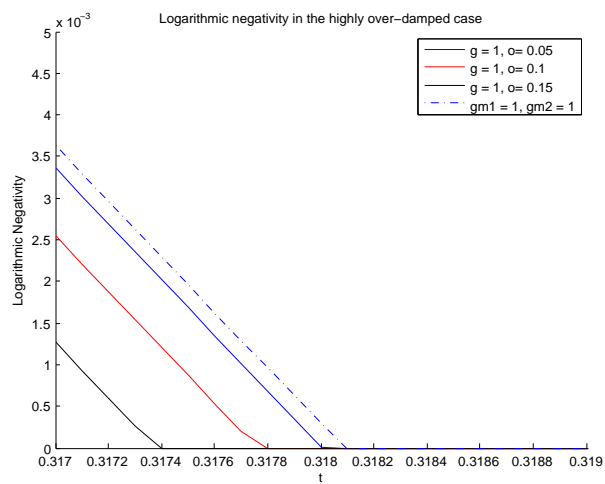}
	\caption{Logarithmic negativity in the highly over-damped case}
	\medskip{This is a detail of Figure~\ref{VeryOD} and has the same legend}
	\label{BLN_VOD_detail}
\end{center}
\end{figure}

\begin{figure}
 \begin{center}
	\includegraphics[scale=0.5]{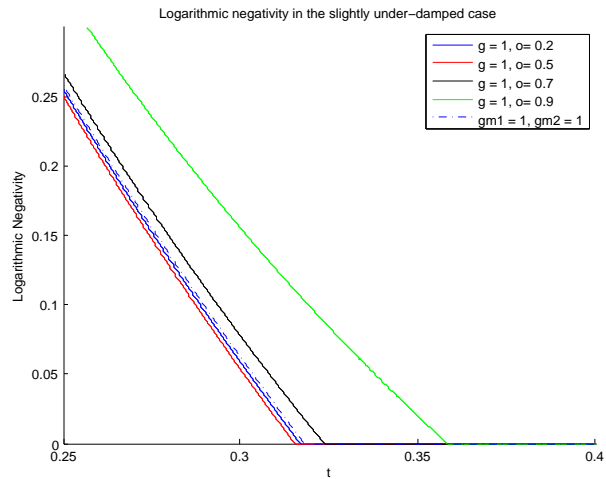}
	\caption{Logarithmic negativity in the slightly under-damped case}
	\medskip{The plots are again obtained while fixing $\gm = \gm_1 = \gm_2 = 1$ and letting $\og_0$ vary as : dashed : $\og_0 = 0$, blue : $\og_0 = 0.2$, red : $\og_0 = 0.5$, black : $\og_0 = 0.7$ and green : $0.9$}
	\label{BLN_slightUD}
\end{center}
\end{figure}

\begin{figure}
 \begin{center}
	\includegraphics[scale=0.5]{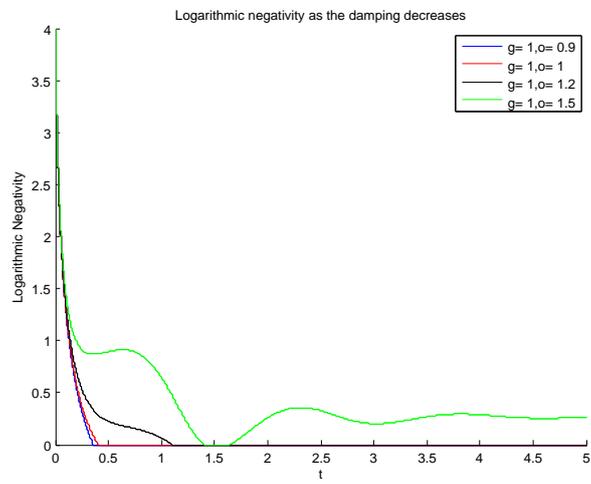}
	\caption{$\mathcal{L}$ as the damping decreases}
	\medskip{The plots are obtained with $\gm = 1$ and : blue : $\og_0 = 0.9$, red : $\og_0 = 1$, black : $\og_0 = 1.2$ and green : $\og_0 = 1.5$}
	\label{Slight_to_VUD}
\end{center}
\end{figure}

\begin{figure}
 \begin{center}
	\includegraphics[scale=0.5]{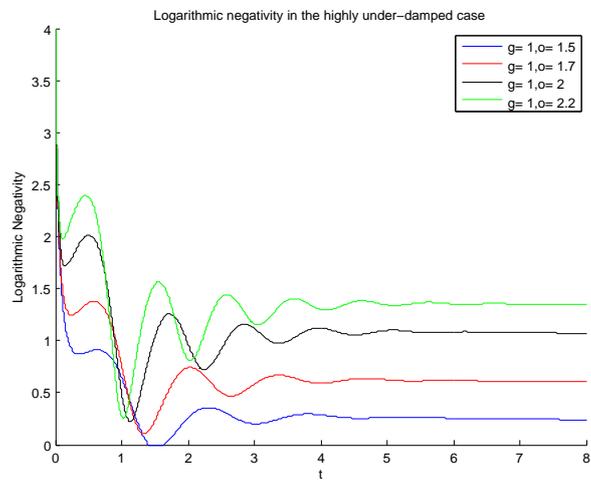}
	\caption{$\mathcal{L}$ in the highly under-damped case}
	\medskip{The plots are obtained with $\gm = 1$ and : blue : $\og_0 = 1.5$, red : $\og_0 = 1.7$, black : $\og_0 = 2$ and green : $\og_0 = 2.2$}
	\label{VeryUD}
\end{center}
\end{figure}

\newpage

\sbp Setting $x' = x - \mu$, we can obtain the off-diagonal variances $\Dt \mu_1 ^2$ and $\Dt \mu_2 ^2$ and study the decoherence. If we recall $x = u + \hb z$ and $x' = u - \hb z$, we can notice $\mu = 2 \hb z$. Then
{\alld
\begin{align}
 \langle \mu_1 ^2 \rangle =& \int \mu_1 ^2 \rho(0, \tbf{x}', t) dx_1 \, dx_2 \non 
\\ =& \int \mu_1 ^2 P(0, \tbf{z}, t) dz_1 \, dz_2 \non 
\\ =& 4 \hb ^2 \int z_1 ^2 {\tilde P}(0 , \tbf{z}, t) dz_1 \, dz_2 \non 
\\ =& 4 \hb ^2  \int z_1 ^2 e^{ - \mathcal{B} z_1 ^2 - \mathcal{B} z_2 ^2 - \mathcal{D} z_1 \, z_2} dz_1 \, dz_2 \non 
\\ =& 4 \hb ^2 \sqrt{\frac{4 \pi ^2}{4 {\mathcal{B}}^2 - {\mathcal{D}}^2}} \frac{2 \mathcal{B}}{4 {\mathcal{B}}^2 - {\mathcal{D}}^2}
\end{align}}
We can easily see that $\langle \mu_1 \rangle$ vanishes as
{\alld
\begin{align}
\langle \mu_1 \rangle =& \int \mu_1 ^2 \rho(0, \tbf{x}', t) dx_1 \, dx_2 \non 
\\ =& 4 \hb ^2  \int z_1 e^{ - \mathcal{B} z_1 ^2 - \mathcal{B} z_2 ^2 - \mathcal{D} z_1 \, z_2} dz_1 \, dz_2 \non 
\\ =& 0
\end{align}}
Thus we get 
\be \Dt \mu_1 ^2 = \frac{16 \hb ^2 \pi}{(4 {\mathcal{B}}^2 - {\mathcal{D}}^2)^{3/2}} \ee
Similarly
\be \Dt \mu_2 ^2 = \Dt \mu_1 ^2 \ee

\sbp Figure~\ref{CohVeryUD} shows how the coherence inside the system evolves alongside $\mathcal{L}$ in the highly under-damped case. The variance was rescaled by a factor of one third. It can easily be seen that the coherence peaks as the negativity dips. This, however, can be seen not to happen in the slightly under-damped case, as shown on Figure~\ref{CohSlightUD} (note that this time, the variance is resized by a factor of one tenth). It is important to note that the coherence in this case never vanishes, but converges towards its minimal value at roughly the same time the logarithmic negativity vanishes. One may recall Chapter 4's results and observe that the coherence in both the slightly and the highly under-damped case act in a similar fashion to that of the one-particle case. 

It is quite remarkable that as $\mathcal{L}$ converges, so does the coherence but with opposite oscillations. This suggests that as the entanglement dies, the state regains some coherence, then as it regains coherence, the system regains some entanglement and the entanglement is high as the lowest coherence.

\begin{figure}
 \begin{center}
	\includegraphics[scale=0.5]{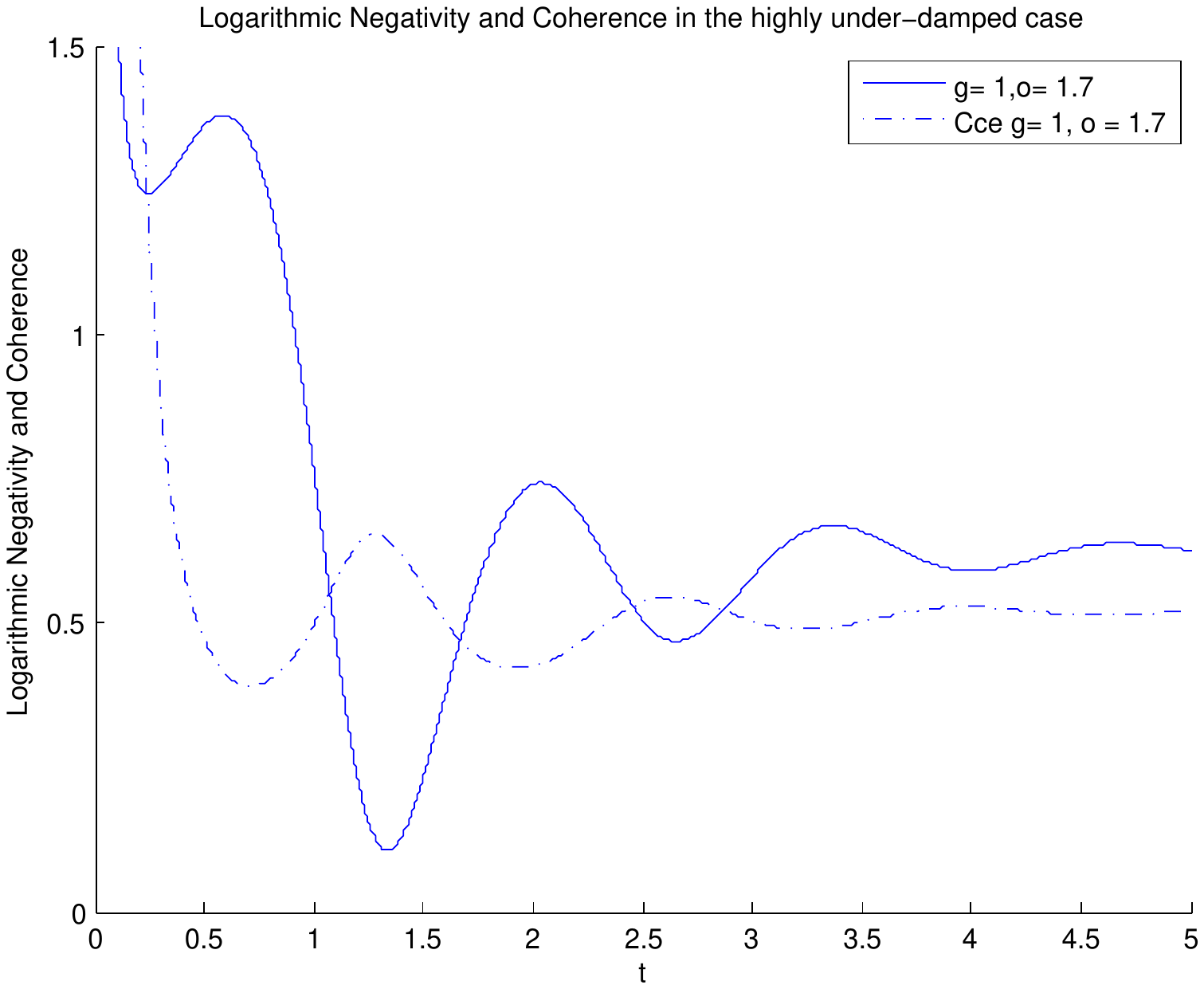}
	\caption{$\mathcal{L}$ and Coherence in the highly under-damped case}
	\medskip{The dashed line represent the coherence. The variance is rescaled by 1/3.}
	\label{CohVeryUD}
\end{center}
\end{figure}

\begin{figure}
 \begin{center}
	\includegraphics[scale=0.5]{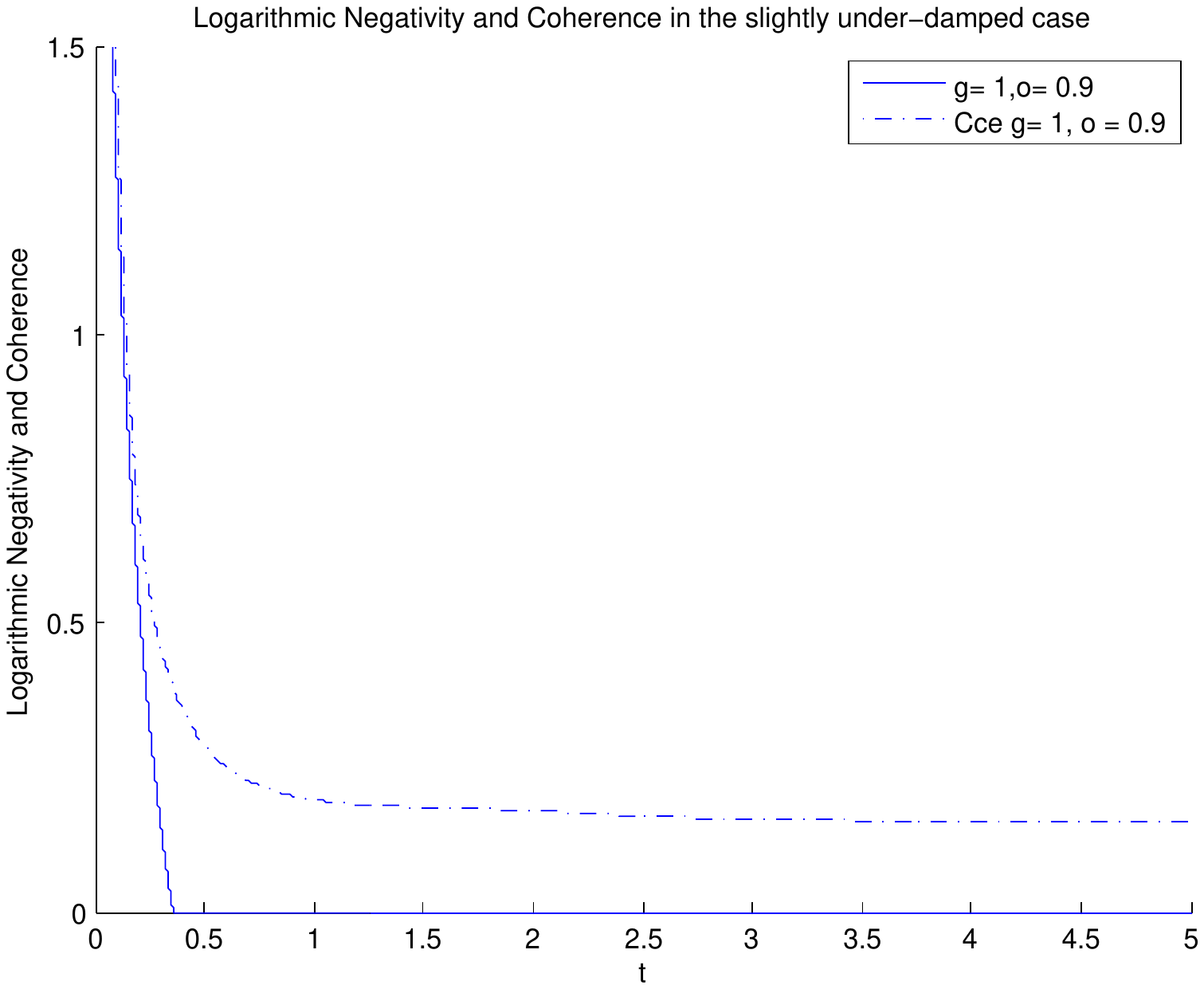}
	\caption{$\mathcal{L}$ and Coherence in the sightly under-damped case}
	\medskip{The dashed line represent the coherence. The variance is rescaled by 1/10.}
	\label{CohSlightUD}
\end{center}
\end{figure}


\chapter*{Conclusion}
\stepcounter{chapter}
\addcontentsline{toc}{chapter}{Conclusion}

\sbp This work set out to study how the entanglement in a bipartite Gaussian state evolves with time, when the state is subjected to an environment. To this end, a master equation approach was chosen and the entanglement studied. Thus in Chapter 2, the Von Neumann entropy was obtained for an initial bipartite Gaussian state. It was then formally shown to be invariant under closed system dynamics. A method of studying Gaussian state was introduced, namely covariance matrices, which make it very easy to estimate the entanglement since it is directly related to the density matrix. Then the logarithmic negativity can easily be obtained using the symplectic eigenvalues of the covariance matrix. In Chapter 3, a Non-Rotating Wave master equation was derived for a general system Hamiltonian using the Quantum Langevin Equation as derived in \cite{FLO'C:QLE1988} and a perturbation technique. The master equation obtained is similar to that of Savage and Walls. 

 In Chapter 4, a single particle state was evolved using the N.R.W. master equation in the case of two different Hamiltonians. Results obtained with a free-particle Hamiltonian in the first section and with a hamonic potential in the second section allowed us to verify the results obtained by Savage and Walls \cite{SavWall:85_1, SavWall:85_2}. Studying the off-diagonal elements of the resulting density matrix, we observed two different damping behaviours, depending on the coupling constant $\gm$ and the frequency $\og$ of the oscillator.

Still using the N.R.W. master equation, the subsequent chapters were dedicated to the study of the entanglement in a two-particle Gaussian state.

In Chapter 5, the bipartite system was coupled to two independent heat baths, one for each particle. The independence of the baths ensured that no entanglement would be created from the interaction between the particles and the baths. This resembles a crudely simple set-up where information may be coded in a pair of correlated wave-packets that is afterwards shared between two independent parties. 

A free-particle evolution revealed that the entanglement vanishes at very short times. In Chapter 6, a harmonic potential was added between the two particles to obtain some quite fascinating results. If the systems are over-damped, i.e. if the coupling is strong compared to the frequency of the oscillator, the entanglement behaves similarly as it would without the potential, vanishing at exponentially shorter times. If the systems are slightly under-damped, the entanglement decreases at longer times. In the highly under-damped case, the logarithmic negativity converges towards a constant in an oscillatory manner.

This result suggests that under certain conditions, it is possible to maintain entanglement in a system for a long time. One can them imagine a set-up where two parties share an entangled state in which some information has been encoded. Then the information contained in the entanglement would endure. However, the study concerns one type of environment, one particular inital state and one particular interaction. It may be that other types of baths will prove less dissipative or that other types of interactions or initial states will prove more useful in practical applications.


\begin{appendices}
\chapter*{Appendices}
\stepcounter{chapter}
\addcontentsline{toc}{chapter}{Appendices}

\section{Extra derivations}

\subsection*{Thermal equilibrium relations \label{thermal-equilibrium-relations}}
\sbp The relations $\langle q_j q_k \rangle$, $\langle p_j p_k \rangle$, $\langle q_j p_k \rangle$ will here derived. The bath Hamiltonian is
\be H_B = \half \sum_j \left[ \frac{p_j ^2}{m_j} + m_j \og_j ^2 q_j ^2 \right] \ee
It can be written in terms of the creation and annihilation operators. 
\be a_j = \frac{m_j \og_j q_j + \im p_j}{\sqrt{2 \hb \og_j m_j}} \hspace{1 in} a_j ^{\dg} = \frac{m_j \og_j q_j - \im p_j}{\sqrt{2 \hb \og_j m_j}} \ee
so that
\be q_j = \sqrt{\frac{\hb}{2 m_j \og_j}} (a_j + a_j ^{\dg}) \hspace{1 in} p_j = \im \sqrt{\frac{\hb \og_j m_j}{2}} (a_j ^{\dg} - a_j) \ee
Then
\begin{align}
 \frac{p_j ^2}{m_j} + m_j \og_j ^2 q_j ^2 =& - \frac{\hb \og_j m_j}{2 m_j} (a_j ^{\dg} - a_j)^2 + m_j \og_j ^2 \frac{\hb}{2 \og_j m_j} (a_j + a_j ^{\dg}) ^2 \non \\ =& \frac{\hb \og_j}{2} \left[ a_j ^2 + 2 a_j a_j ^{\dg} + {a_j ^{\dg}}^2 - a_j ^2 + 2 a_j ^{\dg} a_j - {a_j ^{\dg}}^2 \right] \non 
\\ =& \hb \og_j (a_j a_j ^{\dg} + a_j ^{\dg} a_j) = \hb \og_j (2 a_j ^{\dg} a_j + 1) = \hb \og_j (2 {\hat n}_j +1)
\end{align}
So
\be H_B= \sum_j \hb \og_j (a_j ^{\dg} a_j + \half) = \sum_j \hb \og_j ({\hat n}_j + \half) \ee
Now 
\begin{align}
 \langle q_j q_k \rangle =& \frac{\hb}{2 m_j \og_j} \langle \left( a_j + a_j ^{\dg} \right)^2 \rangle \, \dt_{jk} = \frac{\hb}{2 m_j \og_j} \langle a_j a_j ^{\dg} + a_j ^{\dg} a_j \rangle  \, \dt_{jk} = \frac{\hb}{2 m_j \og_j} \langle 2 {\hat n}_j + 1 \rangle \, \dt_{jk} \non 
\\ =& \frac{\hb}{2 m_j \og_j} \frac{\Tr\left[ (2 {\hat n} + 1) e^{- \bt \hb \og_j ( {\hat n} + \half)} \right]}{\Tr \left[ e^{-\bt \hb \og_j ({\hat n} + \half)} \right]}  \, \dt_{jk} \non
\\ =& \frac{\hb}{2 m_j \og_j} \sum_{n=0} ^{\infty}  \frac{(2 n + 1) e^{- \bt \hb \og_j n} e^{-\bt \hb (\og_j + \half)}}{e^{- \bt \hb \og_j n} e^{-\bt \hb (\og_j + \half)}}  \, \dt_{jk} \non
\\ =& \frac{\hb}{2 m_j \og_j} \sum_{n=0} ^{\infty} \frac{(2 n + 1) e^{-\bt \hb \og_j n}}{e^{-\bt \hb \og_j n}}  \, \dt_{jk}
\end{align}
The sums can be express as follows
\be
 \sum_n  e^{-\bt \hb \og_j n} = \frac{1}{1 - e^{- \bt \hb \og_j}}\ee
and
\be \sum_n n\, e^{-\bt \hb \og_j n} = -\frac{1}{\hb \og_j} \frac{\partial}{\partial \bt} e^{-\bt \hb \og_j n} = \frac{e^{-\bt \hb \og_j}}{(1 - e^{\bt \hb \og_j})^2}
\ee

Hence
\begin{align}
 \langle q_j q_k \rangle =& \frac{\hb}{2 m_j \og_j} \sum_{n=0} ^{\infty} \left[ \frac{2n e^{-\bt \hb \og_j n}}{e^{-\bt \hb \og_j n}} + \frac{ e^{-\bt \hb \og_j n}}{e^{-\bt \hb \og_j n}} \right]  \, \dt_{jk} \non 
\\ =& \frac{\hb}{2 m_j \og_j} \left( \frac{2 \, e^{\bt \hb \og_j}}{1 - e^{-\bt \hb \og_j}} + 1 \right)  \, \dt_{jk} \non 
\\ =& \frac{\hb}{2 m_j \og_j} \frac{1 + e^{- \bt \hb \og_j}}{1 - e^{-\bt \hb \og_j}}  \, \dt_{jk} \non 
\\ =& \frac{\hb}{2 m_j \og_j} \coth \left( \frac{\bt \hb \og_j}{2} \right)  \, \dt_{jk} \non 
\\ \langle q_j q_k \rangle =& \frac{\hb}{2 m_j \og_j} \coth \left( \frac{\hb \og_j}{2 k T} \right)  \, \dt_{jk}
\end{align}

Similarly
\begin{align}
 \langle p_j p_k \rangle =& - \frac{\hb m_j \og_j}{2} \langle (a_j ^{\dg} - a_j)^2 \rangle \, \dt_{jk} = - \frac{\hb m_j \og_j}{2} \langle - a_j ^{\dg} a_j - a_j a_j ^{\dg} \rangle \, \dt_{jk} \non 
\\ =& - \frac{\hb m_j \og_j}{2} \langle - 2 {\hat n}_j - 1 \rangle \, \dt_{jk} \non 
\\ =& - \frac{\hb m_j \og_j}{2} \frac{\Tr \left[ -(2 {\hat n} + 1) e^{-\bt \hb \og_j ({\hat n} + \half)} \right]}{\Tr \left[ e^{-\bt \hb \og_j ({\hat n} + \half)} \right]} \, \dt_{jk} \non 
\\ =& - \frac{\hb m_j \og_j}{2} \sum_{n=0} ^{\infty} \frac{- (2 n +1) e^{-\bt \hb \og_j n} e^{-\bt \hb (\og_j + \half)}}{e^{-\bt \hb \og_j n} e^{-\bt \hb (\og_j + \half)}} \, \dt_{jk} \non 
\\ =& \frac{\hb m_j \og_j}{2} \sum_{n=0} ^{\infty} \frac{(2 n +1) e^{-\bt \hb \og_j n}}{e^{-\bt \hb \og_j n}} \, \dt_{jk} \non
\\ =& \frac{\hb m_j \og_j}{2} \coth \left( \frac{\bt \hb \og_j}{2} \right) \, \dt_{jk} \non 
\\ \langle p_j p_k \rangle =& \frac{\hb m_j \og_j}{2} \coth \left( \frac{\hb \og_j}{2 k T} \right) \, \dt_{jk}
\end{align}
 and
\begin{align}
 \langle q_j p_k \rangle =& \, \im \sqrt{\frac{\hb m_j \og_j}{2}} \sqrt{\frac{\hb}{2 m_j \og_j}} \langle (a_j + a_j ^{\dg}) (a_j ^{\dg} + a_j) \rangle \, \dt_{jk} =\, \im \frac{\hb}{2} \langle a_j a_j ^{\dg}- a_j ^{\dg} a_j \rangle \, \dt_{jk} \non 
\\ \langle q_j p_k \rangle =& \, \im \frac{\hb}{2} \, \dt_{jk}
\end{align}
and finally
\begin{align}
 \langle p_j q_k \rangle =& \, \im \sqrt{\frac{\hb m_j \og_j}{2}} \sqrt{\frac{\hb}{2 m_j \og_j}} \langle (a_j ^{\dg} + a_j)  (a_j + a_j ^{\dg}) \rangle \, \dt_{jk} = \im \frac{\hb}{2} \langle a_j ^{\dg} a_j -  a_j a_j ^{\dg} \rangle \, \dt_{jk} \non 
\\ \langle p_j q_k \rangle =& - \im \frac{\hb}{2} \, \dt_{jk} = - \langle q_j p_k \rangle
\end{align}

\end{appendices}


		\bibliographystyle{alpha}
		\bibliography{bibli.bib}

\end{document}